\pdfoutput=1
\documentclass[mnsc]{informs_modified}

\SingleSpacedXI

\usepackage[textsize=tiny]{todonotes}

\usepackage[ruled,noend]{algorithm2e}

\usepackage[colorlinks=true,citecolor=blue]{hyperref}
\usepackage{cleveref}
\usepackage{threeparttable} 
\usepackage{array} 

\usepackage{natbib}
 \bibpunct[, ]{(}{)}{,}{a}{}{,}%
 \def\bibfont{\small}%
\usepackage{color}
\usepackage{setspace}
\usepackage{bm}
\usepackage{tabu}
\usepackage{array}
\usepackage{multirow}
\usepackage{comment}
\usepackage[toc,page]{appendix}
\usepackage{appendix}
\usepackage{setspace}
\usepackage{hyperref}

\TheoremsNumberedThrough     
\ECRepeatTheorems

\EquationsNumberedThrough    

\usepackage{tabularx}
\usepackage{hyperref}
\usepackage{url}

\usepackage{natbib}
 \bibpunct[, ]{(}{)}{,}{a}{}{,}%
 \def\bibfont{\small}%
\newcommand{\calA}{\mathcal{A}}
\newcommand{\Rev}{\text{Rev}}

\newcommand{\Reg}{\text{Regret}}
\newcommand{\HO}{\Rev^{\star}}
\newcommand{\RDLP}{\Rev(\AlgRDLP)}
\newcommand{\DLP}{\Rev(\AlgDLP)}
\newcommand{\Revp}{\Rev(\pi)}

\newcommand{\Rogd}{\text{Rev}(\AlgDBPC)}
\newcommand{\Rdlpp}{\text{Rev}(\GPRDLP)}
\newcommand{\CR}{\text{CR}}

\newcommand{\A}{\mathbf{A}}
\newcommand{\bbR}{\mathbb{R}}

\usepackage{color-edits}
\addauthor{cp}{red}

\usepackage{xspace}
\newcommand{\xhdr}[1]{\vspace{1mm} \noindent{\bf #1}}
\newcommand{\FCFS}{\textsf{FCFS}\xspace}
\newcommand{\AlgDLP}{\textsf{DLP-PA}\xspace}
\newcommand{\AlgRDLP}{\textsf{RDLP-PA}\xspace}
\newcommand{\AlgSBPC}{\textsf{S-BPC}\xspace}
\newcommand{\AlgDBPC}{\textsf{BPC-OGD}\xspace}
\newcommand{\AlgBL}{\textsf{BL}\xspace}
\newcommand{\AlgN}{\textsf{NESTING}\xspace}

\newcommand{\GPDLP}{\textsf{GP-Enhanced-DLP}\xspace}
\newcommand{\GPRDLP}{\textsf{GP-Enhanced-RDLP}\xspace}
\newcommand{\GPDBPC}{\textsf{GP-Enhanced-BPC-OGD}\xspace}
\newcommand{\GPBL}{\textsf{GP-Enhanced-BL}\xspace}
\newcommand{\GPN}{\textsf{GP-Enhanced-NESTING}\xspace}

\def\E{\mathbb{E}}

\usepackage{enumitem}

\begin{document}

\ARTICLEAUTHORS{
\AUTHOR{Patrick Jaillet}\AFF{Department of Electrical Engineering and Computer Science, Massachusetts Institute of Technology,
\EMAIL{jaillet@mit.edu}}
\AUTHOR{Chara Podimata}\AFF{Sloan School of Management, Massachusetts Institute of Technology, Cambridge, MA, 02139,
\EMAIL{podimata@mit.edu}}
\AUTHOR{Zijie Zhou}\AFF{Operations Research Center,
Massachusetts Institute of Technology, Cambridge, MA, 02139,
\EMAIL{zhou98@mit.edu}}
}

\RUNAUTHOR{Jaillet, Podimata, and Zhou}
\RUNTITLE{Grace Period is All You Need}

\TITLE{Grace Period is All You Need: Individual Fairness without Revenue Loss in Revenue Management}

\ABSTRACT{\textbf{Problem definition: 
 } Imagine you and a friend purchase identical items at a store, yet only your friend received a discount. Would your friend's discount make you feel unfairly treated by the store? And would you be less willing to purchase from that store again in the future? Based on a large-scale online survey that we ran on Prolific, it turns out that the answers to the above questions are \emph{positive}. Therefore, when allocating resources to different customers, sellers should consider both the total reward and individual fairness.
 
 \xhdr{Methodology/results: } Motivated by these findings, in this work we propose a notion of \emph{individual fairness} in online revenue management and an algorithmic module (called \emph{``Grace Period''}) that can be embedded in traditional revenue management algorithms and guarantee individual fairness. Specifically, we show how to embed the Grace Period in five common revenue management algorithms including Deterministic Linear Programming with Probabilistic Assignment (\AlgDLP), Resolving Deterministic Linear Programming with Probabilistic Assignment (\AlgRDLP), Static Bid Price Control (\AlgSBPC), Booking Limit (\AlgBL), and Nesting (\AlgN), thus covering both \emph{stochastic} and \emph{adversarial} customer arrival settings. Embedding the Grace Period does \emph{not} incur additional regret for any of these algorithms. This finding indicates that, {in an asymptotic regime}, there is \emph{no} tradeoff between a seller maximizing their revenue and guaranteeing that each customer feels fairly treated. 

 \xhdr{Managerial implications: } The core intuition behind the Grace Period is that independent randomized decisions for each customer often lead to unfair outcomes. However, we cannot eliminate the randomness, as it plays a crucial role in maximizing profit. The Grace Period addresses this by shifting randomness away from individual decisions and applying it instead to the total number of customers receiving a particular decision. This approach preserves revenue potential while mitigating fairness issues. Furthermore, implementing the Grace Period depends on the specific decision algorithm, with tailored methods provided for each to incorporate this mechanism effectively.}

\maketitle

\section{Introduction} \label{sec:intro}

Revenue management plays a pivotal role in various sectors including airline, retail, advertising, and hospitality, as evidenced by previous research \cite{williamson1992airline,talluri1998analysis, talluri2004revenue, talluri2004theory}. In the most standard version, sequential decision-making for revenue management encompasses the following seller-customer interaction: The seller has access to limited-capacity \emph{resources} of potentially more than one \emph{types}. Customers of different \emph{types} arrive sequentially (the order of arrival can be either \emph{stochastic} or \emph{adversarial}) over a finite period of time. Upon arrival, the seller can choose either to \emph{accept} (i.e., sell the item to them) or \emph{reject} (i.e., not sell) the customer. Acceptance leads to revenue generation equal to the payment made by the customer, alongside the consumption of specific units of resources based on the customer's type. On the other hand, rejection of a customer does not lead to any revenue or resource usage. The goal of the seller is to maximize the total anticipated revenue throughout the given period, while adhering to the capacity constraints of all resources.

This work emphasizes the equitable treatment of customers through the requirement of \emph{individual fairness}. The concept of fairness is crucial not only from an ethical standpoint but also for its impact on customer behavior and potential to increase revenue. For example, think about the following (hypothetical, for now) scenario. Two stores, A and B, offer an identical product at a same price. However, Store A selectively offers discounts leading to perceived unfairness among customers. One may think that such perceived unfair treatment may deter customers from buying at Store A, possibly benefiting Store B. { \citet{xia2004price} support this intuition by showing that when customers perceive a high degree of transaction similarity—such as purchasing the same product under similar conditions or within a short time frame—price discrepancies are more likely to be judged as unfair. They further argue that perceived unfairness can trigger both rational, money-focused responses (e.g., complaints or refund requests) and emotional reactions (e.g., anger or revenge), ultimately eroding customer trust in the seller. Reinforcing this view, \cite{okun2011prices} explains that fair treatment of each individual is instrumental for a profit-maximizing firm to the maximization of long-run profits. To support even further that customer reactions align with our aforementioned intuition, we conducted a large-scale online survey on Prolific, detailed in Section \ref{sec:experiment}. Our survey's findings validate both our priors and the previous research; customers do indeed perceive various forms of unfair treatment from sellers and this in turn affects the probability of them buying again from certain stores.}

{ To incorporate fairness into revenue management, we need a meaningful way to measure the extent of equitable treatment. A fully fair policy—where all customers requesting the same resource are either always accepted or always rejected—would eliminate any unfairness but could significantly reduce revenue by failing to allocate resources efficiently. Instead, a more practical approach is to quantify fairness \emph{probabilistically}, assessing the likelihood that two customers requesting the same resource receive different decisions. This allows us to evaluate fairness while preserving decision-making flexibility and avoiding rigid rules that could harm revenue optimization.

Building on the idea that fairness should be measured probabilistically, we must carefully consider \emph{which} probability matters. There are two natural ways to define fairness: one from the decision-maker's perspective and one from the customer's perspective. The decision-maker might consider fairness in terms of ensuring that all customers requesting the same resource are accepted with the same probability (see e.g., \citep{arsenis2022individual}). However, this does not fully capture the individual customer's perception of fairness, which is shaped by the likelihood that two such customers receive different outcomes—an event that can lead to a rejected customer perceiving unfair treatment. Since fairness concerns arise when one customer is accepted while another is rejected (\cite{xia2004price}), we focus on minimizing the probability of such disparities to ensure that perceived unfairness occurs as infrequently as possible.

To illustrate the difference, consider two algorithms. In the first, two customers requesting the same resource are each accepted with probability 0.3, while in the second, both are accepted with probability 0.5. From the decision-maker’s standpoint, both algorithms appear fair since they treat all such customers identically in terms of individual acceptance probabilities. However, their actual impact on customers differs significantly. In the first case, given that the second customer observes that the first one is accepted, the probability to be rejected is \( 0.3\), while in the second case, this probability increases to \( 0.5\). Thus, even when an algorithm assigns equal acceptance probabilities, the probability of \emph{observed} disparate treatment may vary, meaning fairness at the decision-making level does not necessarily translate to fairness as experienced by customers.  

In this paper, we focus on \emph{fairness from the customer's perspective}, ensuring that the probability of disparate treatment—where one customer is accepted while another is rejected—is explicitly bounded. The exact definition of the fairness metric is provided in Section \ref{sec:def}, where we closely examine its key aspects. We recognize that there is no single, unified way to define fairness; we have included an extensive discussion with our thoughts about how we came up with this definition in Section~\ref{sec:discussion}.
}

In this paper, our primary objective is to incorporate the principle of individual fairness into the most commonly used algorithms solving the quantity-based revenue management problem, including primal methods such as Deterministic Linear Programming with Probabilistic Assignment (\AlgDLP), Resolving Deterministic Linear Programming with Probabilistic Assignment (\AlgRDLP), dual methods such as  Static Bid Price Control (\AlgSBPC), Dynamic Bid Price Control (\AlgDBPC), quota-based methods such as Booking Limit (\AlgBL), and Nesting (\AlgN). We modify each algorithm one by one to meet individual fairness criteria as defined in Definition \ref{def:ifmetrics}. In addition to implementing individual fairness, we also strive for revenue optimization, ensuring that incorporating fairness does not lead to any loss compared to the original algorithm's performance. 
 By adapting these well-known algorithms, we aim to ensure that any revenue management problems addressed by them also uphold the principles of individual fairness.

\subsection{Main Contributions}

\xhdr{``Grace Period'' Design.} In Section \ref{sec:gpd}, we introduce the \emph{Grace Period} which we find to be the workhorse concept for achieving individual fairness in revenue management. At a simplified level, the Grace Period works as follows: 
As the seller's available resources approach depletion, an arriving customer is either accepted or rejected based on the treatment of the preceding customer. If the previous customer was accepted, the succeeding customer is likewise accepted with a probability of $(1-\alpha)$, where 
$\alpha$ is a small fairness parameter capturing the allowable probability of unequal treatment; and the rigorous definition can be found in Definition \ref{def:ifmetrics} in Section \ref{sec:def}. However, if the previous customer was rejected, the incoming customer is also declined.
This design (which gets utilized only when the resources are almost depleted) forms a grace period during which the probability of unequal treatment is confined. We remark that the Grace Period for customers of type $i$ can (and will in general) be different from the Grace Period for customers of type $j$ (for $i \neq j$). The tricky part is to guarantee that despite these distinct grace periods, the seller does not sell resources beyond their allowable capacity. We showcase how the Grace Period can be used via an example analysis of the First-Come-First-Serve ($\FCFS$) algorithm.


The rationale for the Grace Period Design lies in its mitigation of unfairness. If we consider a scenario where the Grace Period Design isn't applied and we adhere to a first-come-first-serve model, the pattern of acceptances and rejections is distinctly binary; the early arrivals are accepted, followed by a series of rejections as resources are exhausted. In such a scenario, the initial group of customers who are declined would be unfairly treated (per Definition~\ref{def:ifmetrics}). However, the application of the Grace Period Design spreads this perception of unfairness more equally among earlier arrivals. Hence, this design maintains individual fairness among every customer with high probability.

\xhdr{Grace-Period-Enhanced Algorithms for Revenue Management.} In Sections \ref{sec:stochatic} and \ref{sec:adversarial}, we integrate the Grace Period design into six well-known algorithm structures frequently used in revenue management (RM): \AlgDLP, \AlgRDLP, \AlgSBPC, \AlgDBPC, \AlgBL and \AlgN. The first four \AlgDLP, \AlgRDLP, \AlgSBPC, \AlgDBPC are used for stochastic customer arrivals, while \AlgBL and \AlgN are used for adversarial arrivals. For completeness, we outline each of these algorithms together with the explanation of why they are not individually fair in Section \ref{subsec:notfair}. Based on the different structure of these algorithms, we employ different methods to incorporate the grace period design into them. For example, in \AlgDLP, each type $i$ customer is accepted with certain probability $p_i$. This can lead to unfair outcomes: consider two consecutive type \( i \) customers—if the first is accepted, the second will be rejected with probability \( 1 - p_i \), which is the probability that the second customer perceives unfair treatment. The reverse situation also applies. When \( p_i \) is small, the overall probability that customers experience disparate treatment increases, amplifying perceived unfairness. To address this, we avoid random decisions for individual arrivals. Instead, we divide the timeline into segments, accepting a random number of customers in a first-come-first-serve manner per segment. Grace periods \emph{inside} each segment ensure fairness when decisions change from acceptance to rejection within each segment; grace periods \emph{between each segment} guarantee fairness when decisions change from rejection to acceptance between the segments. 
For clarity, we refer to the new algorithms as \textsf{GP-Enhanced}.

\xhdr{Minimal Trade-off between Revenue Maximization and Individual Fairness in an Asymptotic Regime.}  
Contrary to the common belief that there is a fundamental trade-off between maximizing revenue and ensuring individual fairness, we show that standard RM algorithms, when augmented with the Grace Period design, can achieve fairness with minimal revenue loss in the asymptotic regime.

In addition to the fairness parameter \( \alpha \) discussed above, our fairness definition (see Definition \ref{def:ifmetrics}) includes a second parameter \( \delta > 0 \), which captures the stochastic nature of the arrival process. Specifically, fairness is required to hold with probability at least \( 1 - \delta \), allowing the algorithm to remain robust against rare, extreme arrival patterns. In the following performance summaries, we set \( \delta = 1/T \), where \( T \) is the length of the time horizon. This choice ensures fairness holds with high probability while enabling a fair comparison across algorithms.  

In the stochastic arrival setting, algorithm performance is evaluated using \emph{regret}, defined as the expected gap between the revenue earned by the algorithm and the revenue of the offline optimal policy with full knowledge of demand. Table \ref{table:stochastic} summarizes the regret bounds for our fairness-aware \textsf{GP-Enhanced} algorithms and their original counterparts. All results in the table assume \( \delta = 1/T \), and for any fixed \( \alpha \in (0,1) \), the asymptotic regret remains unchanged—demonstrating that fairness can be achieved without sacrificing long-run revenue. For settings where the customer arrival process is adversarial, the performance of RM algorithms is measured in terms of the \emph{competitive ratio}; that is, the worst-case ratio (across all possible arrival instances) between the revenue generated by the algorithm and the hindsight optimal (i.e., the maximum revenue assuming full knowledge of the realized demand). Table \ref{table:adversarial} contains a summary of the competitive ratio for \textsf{GP-Enhanced} algorithms for adversarial arrivals versus the traditional non-individually fair counterparts.worst-case environments.


\begin{table}[ht]
\small
\caption{Regret of Original Algorithms and {\textsf{GP-Enhanced}} Algorithms under Stochastic Arrivals}\label{table:stochastic}
\begin{center}
\begin{tabular}{ | m{3.6cm} | r | r | }
\hline
{\bf Original Algorithm}       & {\bf Regret of Original Algorithm}    & {\bf Regret of {\textsf{GP-Enhanced}} Algorithm} \\
\hline
\AlgDLP & $O(\sqrt{T})$ & $O(\sqrt{T})$ (Theorem \ref{thm:RDLPrevise}) \\
\hline
\AlgRDLP (resolve once) & $O(T^{1/3})$ & $O(T^{1/3})$ (Theorem \ref{thm:RDLPrevise})  \\
\hline
\AlgSBPC & $O(\sqrt{T})$ & $O(\sqrt{T})$ (Theorem \ref{thm:FBPCrevise})  \\
\hline
\AlgDBPC & $O(\sqrt{T})$ & $O(T^{2/3} \log T)$ (Theorem \ref{thm:reviseogd})   \\
\hline
\end{tabular}
\end{center}
\end{table}

\begin{table}[ht]
\small
\centering 
\caption{Competitive Ratio of Original Algorithms and \textsf{GP-Enhanced} Algorithms under Adversarial Arrivals. $C \in [0,1]$ is a setting-specific value that is specified in the analysis, $m$ is the number of resources}.
\label{table:adversarial}
\begin{threeparttable}
\begin{tabular}{ | m{3.5cm} | >{\raggedleft\arraybackslash}m{5cm} | >{\raggedleft\arraybackslash}m{7cm} | }
\hline
{\bf Original Algorithm} & {\bf Competitive Ratio of Original Algorithm} & {\bf Competitive Ratio of {\textsf{GP-Enhanced}} Algorithm} \\
\hline
\AlgBL & $C$ & $C-O(\log m / m)$ (Theorem \ref{thm:revisebooking}) \\
\hline
\AlgN & $C$ & $C-O(\log m / m)$ (Theorem \ref{thm:revisenesting}) \\
\hline
\end{tabular}
\end{threeparttable}
\end{table}

\xhdr{\emph{No} Additional \emph{Computational} Cost for Implementing Individual Fairness.} One might conjecture that simultaneously achieving optimal revenue and being individually fair might necessitate additional computational cost for the algorithm; but this is \emph{not} the case. In all \textsf{GP-Enhanced} algorithms, we have preserved the algorithms' original structure to the maximum extent possible, maintaining the order of computational complexity equivalent to that of the original algorithm. 


\xhdr{Survey Highlighting the Significance of Individual Fairness in RM.} We conducted an online survey between October and November 2023 on Prolific about people's perception of unfair treatment in RM. Our findings suggest: (i) When customers witness behavior in a store that is considered ``unfair" (as described in Definition \ref{def:ifmetrics}) before they shop, their inclination to buy decreases. (ii) A customer who has bought from a store in the past and then experiences what they perceive as unfair treatment is likely to decrease their purchasing from that store in the future. (iii) Individuals perceive different treatment happening in the near term as more unfair compared to those in the distant past. With these observations, we can evaluate the importance of individual fairness to business and validate the mathematical definition of individual fairness (Definition \ref{def:ifmetrics}).

{We have included a discussion on additional related work in Appendix \ref{append:intro}, where we discuss the time variant individual fairness metrics that have been traditionally used in the literature. Importantly, we categorize most previous definitions via an axiomatic approach. We also provide a comprehensive literature review on modern revenue management models, price-based revenue management, and the group fairness in revenue management.}

\section{{Classical Network Revenue Management Model \& Preliminaries}}\label{sec:classicalnrm} 

{ 
For ease of exposition, we present the classical \emph{network revenue management model}\footnote{There is also a special case where the decision-maker only has one type of resource, and the special model is called as \emph{single-leg revenue management}.}. In the main text, we will operate under the assumption that all algorithms are applied to this classical model. } Next, we review the classical network revenue management model.

In classical network revenue management, there are $L$ types of indivisible resources, each with a capacity of $m_j, j \in [L]$. All capacities $m_j, \forall j \in [L]$, scale consistently, i.e., there exist constants $q_j$ for every $j \in [L]$ such that $m_j=q_j m$. There is a total of $T$ ---where $T$ is also scaling with $m$--- customers, arriving one by one. There are $n$ customer types. When a customer of type $i \in [n]$ arrives, they may request multiple units from several resource types $j$; we denote this as $A_{ij}$. Matrix $\A \in \bbR_{+}^{n \times L}$ is called the \emph{demand} matrix. Without loss of generality, we assume that $A_{ij} \leq \bar a$. When a type $i$ customer arrives in the system, the decision-maker has two choices: either accept the customer, allocate $A_{ij}$ units of resource $j$, and earn $r_i$ in revenue, where $r_i$ are known constants; or decline the customer's request. The primary objective is to \emph{maximize cumulative revenue}.

We study two variants of the above model based on the \emph{customers' arrival process}. For the first variant (Section~\ref{sec:stochatic}), we assume that customers arrive sequentially at random; i.e., each customer of type $i$ has an associated arrival rate $\lambda_i$. In every round, at most one customer will arrive, with the probability of a type $i$ customer arriving being $\frac{\lambda_i}{\sum_{i'=0}^{n}\lambda_{i'}}$\footnote{Typically, we normalize all arrival rates so that $\sum_{i'=0}^{n} \lambda_{i'}=1$.}. Here, $\lambda_0$ represents the rate for no arrivals. In Sections \ref{subsec:dlp}, and \ref{subsec:bid}, we study the case where the arrival rates $\lambda_i$ are known; Section~\ref{subsec:algdbpc} studies the case where the rates $\lambda_i$ are unknown. Under stochastic arrivals, \textit{regret} is commonly used for measuring the performance of an online algorithm. For the second variant (discussed in Section \ref{sec:adversarial}), we study the case where an adversary determines both the quantity and sequence of each type of customer, commonly referred to as the worst-case/adversarial scenario. Under adversarial arrivals, the \textit{competitive ratio} is the metric to study the performance of an online algorithm. 

\subsection{Regret and Competitive Ratio Analysis} \label{subsec:regret}

We first focus on the \emph{stochastic} arrival setting. Let $\Lambda_i(T)$ be the random variable corresponding to the total number of type $i$ customers arriving between $[0,T]$. The \textit{hindsight optimum} denotes the optimal revenue achievable if the exact number of arrivals for each customer type is known in advance. For the classic network RM, the hindsight optimum is defined as
\begin{align} \label{eq:defho}
    \HO = \max_{\mathbf{x}} \Big\{ \sum_{i\in [n]}r_ix_i \quad \text{ s.t. } \sum_{i\in [n]}\mathbf{A}_ix_i \leq \mathbf{m} \text{, } 0 \leq x_i \leq \Lambda_i(T), \forall i \in [n] \Big\},
\end{align}
where in Eq.~\eqref{eq:defho}, $\mathbf{x}$ is the decision variable vector representing the quantity of each type of customer to accept, $\mathbf{A}_i$ is the $i^{\text{th}}$ column of preference matrix $\mathbf{A}$, $\textbf{m}$ is the vector containing the capacity of each resource $m_j$, and $\Lambda_i(T)$ is a known value representing the exact number of type $i$ arrivals.

Given that the hindsight optimum is privy to the precise count of each customer type's arrivals, it is typically utilized as a benchmark to evaluate the online algorithms. The following definition introduces the concept of \emph{regret}.

\begin{definition} \label{def:regret}
Given $\Revp$ as the expected revenue generated by the online algorithm $\pi$, the regret of algorithm $\pi$ is defined as $\Reg(\pi) := \E[\HO]-\Revp$.
\end{definition}

Note that since the hindsight optimum is informed of the precise number of each type's arrivals, the resulting revenue is an upper bound for the revenue of any online algorithm. Hence, $\Reg \geq 0$, and the smaller the regret, the better.

Next, we consider the \emph{adversarial} arrival setting. Let $\mathcal{I}$ represent the set of all potential arrival instances. For any $I \in \mathcal{I}$, let $\HO(I)$ denote the maximum revenue achievable for arrival instance $I$. We denote $\text{Rev}({\pi}(I))$ as the expected revenue elicited by an online algorithm $\pi$ under the specific arrival instance $I$. 

\begin{definition} \label{def:cr}
The \emph{competitive ratio} of an algorithm $\pi$ is defined as $\CR(\pi) = \inf_{I \in \mathcal{I}}\frac{\text{Rev}({\pi}(I))}{\HO(I)}$.
\end{definition}

Note that competitive ratio is \emph{not} equivalent to regret. If the competitive ratio of an online algorithm $\pi$ equals $0$, this implies that a sequence exists in which the algorithm $\pi$ fails to generate \emph{any} revenue. However, generally speaking, it is possible that $\pi$ could perform efficiently on the majority of sequences, and if the arrival process is stochastic, $\pi$ could also incur diminishing regret.

In the next two subsections, we outline the most popular algorithmic structures for RM with \emph{stochastic} and \emph{adversarial} arrivals. The pseudo-code for all the algorithms introduced below can be found in the Appendix.

\subsection{Review: Popular Algorithm Structures under Stochastic Arrivals} \label{subsec:sto}

{
In this section, we provide a review of existing algorithmic frameworks designed for stochastic arrival processes. These algorithms can be broadly categorized into two main approaches:  
\begin{enumerate}
    \item \textit{Primal Methods (Section \ref{subsubsec:primal}):  } These algorithms rely on the primal solution of the linear programming formulation to guide decision-making.  
    \item \textit{Dual Methods (Section \ref{subsubsec:dual}): } These approaches leverage the dual solution of the linear programming formulation to inform allocation decisions.  
\end{enumerate}
}

\subsubsection{Primal Methods} \label{subsubsec:primal}

\text{ }


\xhdr{Algorithm \AlgDLP.} The Deterministic Linear Programming with Probabilistic Assignment algorithm (\AlgDLP) solves the deterministic linear programming at the beginning and uses the optimal primal solution to make probabilistic decisions, where the DLP formulation is derived by replacing all random variables with their expected values in the hindsight formulation Eq.~\eqref{eq:defho}. In the network RM model, since $\lambda_i$ is the likelihood of type $i$ arrival in each time epoch, the total number of type $i$ customers arriving between $[0,T]$ is $\Lambda_i(T) \sim \text{Bin}(T,\lambda_i)$. The DLP is defined as:
\begin{align} \label{eq:defdlp}
    \text{DLP}^{\star} = \max_{\mathbf{x}} \Big\{ T\sum_{i\in [n]}r_ix_i \quad \text{ s.t. } \sum_{i\in [n]}\mathbf{A}_ix_i \leq \frac{\mathbf{m}}{T} \text{, } 0 \leq x_i \leq \lambda_i, \forall i \in [n] \Big\},
\end{align}

We denote the optimal solution of Eq.~\eqref{eq:defdlp} as $\mathbf{x}^{*}$ (assuming it is unique). Algorithm \AlgDLP then accepts each type $i$ customer with a probability of $\frac{x_i^{*}}{\lambda_i}$, given sufficient capacity. The details can be found in Algorithm \ref{alg:DLP} in Appendix \ref{appendix:algs}. The asymptotic regret of Algorithm \AlgDLP for classical network RM is $O(\sqrt{T})$ \citep{gallego1994optimal,gallego1997multiproduct}.



\xhdr{Algorithm Resolving DLP (\AlgRDLP).} 
The \AlgRDLP algorithm addresses DLP's limitation of not accounting for the demand's randomness after $t = 0$. At a high level, \AlgRDLP solves the DLP and employs its optimal primal solution, denoted as $\mathbf{x}^{*}$, to make probabilistic assignments. However, at time point $t^{*}$, the applicability of the initial optimal solution $\mathbf{x}^{*}$ for making decisions may be compromised due to the random noise of the stochastic arrival process and probabilistic decision behavior.
Hence, we re-solve the DLP with the remaining capacity and remaining time horizon, and continue making probabilistic assignments by the new optimal solution $\mathbf{\tilde{x}^{*}}$ until the end. \cite{reiman2008asymptotically, jasin2012re, bumpensanti2020re} apply this approach with different resolving times in revenue management and achieve $o(\sqrt{T})$ regret.

{
An alternative version of \AlgRDLP, proposed by \cite{vera2019bayesian}, takes a different approach to make online decisions after resolving the DLP at every single time period. Instead of using probabilistic assignments, this version employs the compensated coupling technique, making deterministic accept/reject decisions based on the resolved optimal primal solutions. This method achieves a regret bound comparable to the standard \AlgRDLP. We will discuss in Section \ref{subsec:notfair} its inability to alleviate our fairness concerns.
} 

\subsubsection{Dual Methods} \label{subsubsec:dual}

\text{ }

Primal methods in revenue management use the dual solution of a linear program to guide allocation decisions. 

\xhdr{Algorithm Static Bid Price Control (\AlgSBPC).} In contrast to the \AlgDLP and \AlgRDLP algorithms, bid price control is a threshold-based policy that makes deterministic decisions for every arriving customer. 
A Static Bid Price Control (\AlgSBPC) uses a consistent threshold price for each leg, which remains unchanged over time. A request is accepted if the revenue it will provide the supplier exceeds the threshold price and is otherwise rejected.

To obtain the optimal fixed threshold price, \cite{williamson1992airline, talluri1998analysis, talluri2004theory} propose solving the dual problem of the DLP at the beginning, denoting the optimal dual variable as $\mathbf{\theta}^{*} \in \mathbb{R}^L$ (assuming uniqueness). When each customer of type $i$ arrives, if the revenue $r_i$ exceeds the aggregated bid price $\sum_{j=1}^L \theta_j^{*}A_{ij}$, the customer is accepted; if not, the customer is rejected. The regret of \AlgSBPC for network RM is $O(\sqrt{T})$.

{ 
\xhdr{Algorithm Dynamic Bid Price Control (\AlgDBPC).}
Similarly to primal methods, dual-based approaches can also dynamically update the dual variables at each time period. However, unlike primal methods, which often require resolving an LP, dual methods are well-known for their computational efficiency. Rather than recomputing the optimal dual solution at every step, these methods rely on first-order online optimization techniques to update the bid price dynamically. For example, \cite{li2020simple, sun2020near, balseiro2023best} propose a dynamic bid price control approach, leveraging online gradient descent or online mirror descent (\AlgDBPC). By continuously adjusting the bid prices, these algorithms set a dynamic threshold price for each resource, making deterministic accept/reject decisions for every customer.  

Compared to \AlgSBPC, \AlgDBPC achieves the same \( O(\sqrt{T}) \) regret bound under stochastic arrivals while offering two key advantages: (i) No prior knowledge of arrival rates: the algorithm remains effective even when the arrival rates \( \lambda_i \) are entirely unknown. (ii) No need to solve any primal or dual LP. However, as we discuss in Section \ref{subsec:notfair}, \AlgDBPC does not satisfy our fairness desiderata. 
}

\subsection{Review: Popular Algorithm Structures under Adversarial Arrivals}


\xhdr{Booking Limit (\AlgBL).} In the \AlgBL algorithm~\citep{williamson1992airline}, the supplier assigns a fixed quota/booking limit $\mathbf{b}=\{b_1,b_2,\ldots,b_n\}$ for each customer type. In the asymptotic setting, all booking limits scale with the resource quantity; hence, $b_i=\Theta(m)$ for $i \in [n]$. The algorithm accepts any type $i$ customers provided that there is sufficient resource capacity and the number of accepted type $i$ customers has not yet reached $b_i$. Beyond this point, all type $i$ customers are rejected. 

\xhdr{\AlgN.}
The \AlgN algorithm~\citep{williamson1992airline} is applicable for settings with a single type of resource and where customers can be categorized and ranked based on the revenue $r_i$ that they add to the supplier. Unlike the approach of assigning individual booking limits to each customer type (which is suboptimal) the \AlgN algorithm assigns booking limits to aggregated categories of customer types. For instance, in a situation with three customer types ranked as $r_1 > r_2 > r_3$, the algorithm sets a booking limit $b_1$ for type 1 customers and $b_2$ for the combined group of type 1 and type 2 customers. This method demonstrates superior performance compared to traditional booking limits and is considered optimal in the context of single-resource scenarios. However, a significant limitation of the \AlgN algorithm is its inapplicability in multi-resource environments, where ranking customers becomes impossible.


The \AlgN framework shares some similarities with the \AlgBL algorithm. Initially, $n$ quotas $\mathbf{b}=\{b_1,b_2,\ldots,b_n\}$ need to be established. Following this, $b_i$ is assigned as a quota for the total number of type $i$, $i+1$, $\ldots$, $n$ customers. 

{After introducing the prevalent algorithmic frameworks, it is imperative to underscore that, albeit our discourse primarily navigates through the classical network revenue management model for elucidatory convenience, our findings retain their applicability across a broad spectrum of more realistic revenue management models delineated in Section \ref{subsec:other}. This is attributed to our methodological refinement of the aforementioned algorithmic frameworks, thereby rendering these algorithms amenable to addressing a wide range of modern revenue management challenges. }

{
\section{Definition of Individual Fairness Metrics} \label{sec:def}

In Section \ref{sec:intro}, we conceptually introduced individual fairness in quantity-based revenue management from the customer’s perspective. Recall, a customer does not perceive unfairness simply from being denied a resource; rather, unfairness arises when some customers receive the resource while others, requesting the same resource under similar conditions, do not. To quantify this disparity in treatment, we use probability as a measure of fairness. Ensuring individual fairness requires minimizing the probability that two similar customers receive different treatment. Having established the classical quantity-based network revenue management model in Section \ref{sec:classicalnrm}, we now formally define the individual fairness metric.

\begin{definition}[Individual Fairness in Revenue Management] \label{def:ifmetrics}
Let $n$ be the number of customer types, with each type $i \in [n]$ having $n_i$ customers arriving. Let $[n_i]$ be the set which contains all type $i$ customers. An algorithm $\mathcal{A}$ is $(\alpha,\delta)$-fair for any choice of $\alpha, \delta \in (0,1)$ if for each customer of type $i$, with a probability at least $1-\delta$, for any $u, v \in [n_i]$, 
\begin{equation}\label{eq:tempfair} 
\mathbb{P}\left[i(v) \succ_{\calA} i(u) \right] \leq \alpha \cdot d(u,v),
\end{equation}
where $i(u), i(v)$ denote the $u$-th and $v$-th customer's of type $i$, and \( i(v) \succ_{\calA} i(u) \) indicates that algorithm \( \mathcal{A} \) does not accept \( i(u) \) given that it accepts customer \( i(v) \), namely, $\mathbb{P}\left[i(v) \succ_{\calA} i(u) \right] = \mathbb{P}\left[i(v) \text{ is accepted} |  i(u) \text{ is rejected}\right]$. \( d(u,v) \) is a general time-distance function that satisfies the following properties:  

\begin{enumerate}  
    \item \( d(u,v) \) is increasing with respect to \( |u-v| \).  
    \item \( d(u,v) = 1 \) if \( |u-v| = 1 \).
    \item For any $\alpha \in (0,1)$ chosen by the seller, there exist $u, v$ within the seller-buyer interaction horizon $T$ such that $d(u,v) = 1/\alpha$. Specifically, let $w_\alpha >1$ (independent of $T$) be the smallest positive integer such that $|u - v| = w_\alpha$ and $d(u,v)\geq 1/\alpha$,
\end{enumerate} 
\end{definition}

This definition ensures that fairness constraints hold with high probability, allowing for more flexible and practical revenue management policies while mitigating inefficiencies caused by extreme arrival sequences. At a high level, property 3 implies that customers do not really care about being treated differently with customers that are very far into the future. To better understand the implications and design choices behind this fairness metric, we include an extensive discussion in Section~\ref{sec:discussion}. 

\subsection{Unfairness in Classical RM Algorithms} \label{subsec:notfair}

After introducing the fairness metric in Definition \ref{def:ifmetrics}, we show that all classical revenue management (RM) algorithm structures can violate individual fairness for certain parameter choices of \( (\alpha, \delta) \). We classify the sources of unfairness into three categories:

\xhdr{(i) Unfairness due to resource capacity limits.  } This form of unfairness arises from the fundamental constraint that resources are limited. It affects all algorithms: for example, consider two type \( i \) customers who arrive close to each other. If the earlier customer arrives when the resource is already depleted, they will be rejected, while the later customer may have been accepted had they arrived earlier. Even though these customers arrive close in time and request the same resource, one may be accepted while the other is inevitably rejected simply because the resource becomes unavailable in between their arrivals.

\xhdr{(ii) Unfairness due to randomized individual decisions. }  
Algorithms such as \AlgDLP and \AlgRDLP fall into this category. These algorithms assign each type \( i \) customer an acceptance probability of \( \frac{x_i^*}{\lambda_i} \), independently of others. While this appears fair from the decision-maker’s perspective, it results in observable disparities. For instance, two consecutive type \( i \) customers may receive different outcomes—one accepted, the other rejected—with a non-negligible probability. If \( \alpha \) is smaller than this disparity probability ($1-\frac{x_i^*}{\lambda_i}$), the algorithm violates the fairness constraint.

\xhdr{(iii) Unfairness due to deterministic individual decisions. } This issue arises in algorithms that make accept/reject decisions for each individual based on deterministic rules. Notable examples include the alternative version of \AlgRDLP proposed in \cite{vera2019bayesian} and the \AlgDBPC algorithm. These algorithms resolve a linear program (either primal or dual) at each time step and deterministically decide whether to accept or reject a customer. As a result, consecutive customers of the same type may experience alternating treatment patterns—e.g., one accepted, the next rejected—even when they arrive in close succession. This deterministic oscillation leads to a violation of the fairness metric for any \( \alpha < 1 \).

\subsection{Objective: Enhancing Classical RM Algorithms to Satisfy Fairness Constraints with Minimal Efficiency Loss} \label{subsec:goal}

Having introduced the fairness metric in Definition~\ref{def:ifmetrics}, we now formalize the main goal of this paper: to enhance classical revenue management (RM) algorithms so that they satisfy the fairness constraint without significantly sacrificing asymptotic performance.

Specifically, we aim to enhance the algorithmic structures discussed in Section~\ref{sec:classicalnrm}—including primal-based, dual-based, and booking-limit-based methods. Roughly speaking, we define ``enhancement'' as the ability to maintain the same computational complexity as the original algorithm while also satisfying our fairness definition. We aim for our enhanced algorithms to maintain near-optimal regret for stochastic settings, and a similar-to-the-original competitive ratio.



}

\section{The ``Grace Period''} \label{sec:gpd}

In this section, we present the concept of the ``grace period'' and present the analysis of its incorporation into the First-Come-First-Serve algorithm (\FCFS). This will serve as a building block for the \textsf{GP-Enhanced} algorithms that we will follow. 
\FCFS works both in the case of \emph{stochastic} and \emph{adversarial} arrival sequences in network revenue management problems and is rather simple: it keeps accepting every arriving customer until the capacity of the requested resources is not enough. The primary sources of unfairness in RM settings (and hence also in the case of \FCFS) are: (i) limited resources, and (ii) the inherent structure of the algorithm. Note that when the available resources exceed the total demand, then \FCFS is always \emph{fair} since it accepts all customers; the problem arises when there are limited resources compared to the demand faced. 

At a high level, the \emph{grace period} is a time interval during which the algorithm (i.e., \FCFS in this case) decides when to stop accepting customers gracefully using randomness.

{
\begin{definition}(Decreasing Grace Period)
The \textit{decreasing grace period $[t_1(i),t_2(i)]$} for type $i$ customers is defined as: for \emph{any} $i(k)$ customer that arrives between round $t_1(i)$ and $t_2(i)$, $i(k)$ is accepted with a probability of $1-\beta$ if customer $i(k-1)$ was accepted, and $i(k)$ is rejected if $i(k-1)$ was rejected, where $\beta = 1 - (1-\alpha)^{1/w_{\alpha}}$ and $w_{\alpha}$ is defined in the third property of $d(u,v)$ in Definition \ref{def:ifmetrics}.\footnote{Note that $w_\alpha$ is a \emph{parameter} of the problem \emph{controlled} by the seller. Specifically, the seller originally chooses the target $\alpha$ together with the distance metric $d(u,v)$ and then just solves the equation $d(u,v) = 1/\alpha$ to identify $w_\alpha$.} In the decreasing grace period design, each customer's acceptance depends only on the outcome of the immediately preceding customer, and is otherwise probabilistically independent. 
\end{definition}
}
Note that the grace period is defined \emph{per customer type}, i.e., in general, each customer type will have a different grace period. The exact mathematical instantiation of the parameters of the grace period (i.e., the interval $[t_1(i), t_2(i)]$) as well as the fairness parameters $\alpha, \delta$ are \emph{algorithm} specific. To illustrate this, we mathematically instantiate the grace period for the \FCFS algorithm below.

\subsection{Decreasing Grace Period in Network Revenue Management under \FCFS} \label{subsec:classicalnrm}

To achieve $(\alpha,\delta)$-fairness in network revenue management under \FCFS, we can utilize the decreasing grace period as follows: for any $\alpha, \delta$, let $\gamma = \log_{1-\beta} \delta$, and assign a decreasing grace period $[t_1(i), t_2(i)]$, where $t_1 (i) = \inf\{t: \min_{j \in [L]}m_j(t) \leq \bar a n \gamma\}$ and $t_2(i) = T$ to each type $i$ customer within the range of $[n]$. Here, $m_j(t)$ denotes the remaining capacity of resource $j$ at time $t$. In other words, in order to address the unfairness arising from resource scarcity, it is sufficient for \FCFS to have a common decreasing grace period for all customer types. As a result, for \FCFS we will drop the dependence on the customer type and denote the grace period as $[t_1, t_2]$. The following theorem states that, by implementing the decreasing grace period as above, we satisfy $(\alpha,\delta)$-fairness, and the revenue loss in comparison to the \FCFS algorithm is bounded.

{
\begin{theorem} \label{thm:FCFSnetwork}
    Assigning a decreasing grace period $[t_1, t_2]$, where $t_1 = \inf\{t: \min_{j \in [L]}m_j(t) \leq \bar a n \gamma\}$ and $t_2 = T$ to each type $i$ customer, where $\gamma = \log_{1-\beta} \delta$, $\beta = 1 - (1-\alpha)^{1/w_{\alpha}}$ and $w_{\alpha}$ is defined in the third property of $d(u,v)$ in Definition \ref{def:ifmetrics}. \FCFS is $(\alpha, \delta)$-fair for network revenue management. Moreover, compared to the \FCFS algorithm without grace period, the revenue loss is bounded by $\frac{\bar a}{\underline a} n \gamma \bar r$, where $\bar{r}=\sup_{i \in [n]} r_i$, $\bar a = \sup_{i,j} A_{ij}$, and $\underline a = \inf_{i,j} A_{ij}$.
\end{theorem}

\begin{proof}{Proof of Theorem \ref{thm:FCFSnetwork}}
    We first show that without the resource capacity constraint, for any $u, v \in [n_i]$, the decreasing grace period can guarantee that
\begin{equation} \label{eq:tempcons1}
\mathbb{P}\left[i(v) \succ_{\calA} i(u) \right] \leq \alpha \cdot d(u,v).
\end{equation}
W.L.O.G, we assume that $u>v$ and denote $x = u-v$. If $x \geq w_{\alpha}$, Inequality \eqref{eq:tempcons1} holds since $\alpha d(u,v) \geq 1$; this is because $d(u,v)$ is increasing in $|u - v|$ and $d(u,v) = 1/\alpha$ when $|u - v| = w_\alpha$. If $x < w_{\alpha}$, recall the definition of the decreasing grace period: if the previous customer is accepted, the next is accepted with probability $1 - \beta$; otherwise, the next is rejected. Thus, if customer $i(v)$ is accepted and $i(u)$ is rejected, it must be that some customer in $\{i(v+1), \ldots, i(u)\}$ is the first to be rejected after a sequence of accepted customers. Specifically, for each $k \in \{v+1, \ldots, u\}$, we consider the case where $i(v+1), \ldots, i(k-1)$ are accepted and $i(k)$ is rejected. We then sum the probabilities of all such scenarios:

\begin{align*}
    \mathbb{P}\left[i(v) \succ_{\calA} i(u) \right] &= \mathbb{P} \left[ \text{one of the customers in } \{i(v+1), \ldots, i(u)\} \text{ is the first to be rejected}| i(v) \text{ is accepted}    \right]\\&\leq \sum_{j=0}^{x-1}\mathbb{P}\left[ i(v+1), i(v+2), \ldots, i(v+j)  \text{ is accepted and } i(v+j+1) \text{ is rejected} | i(v) \text{ is accepted} \right] \\&= \sum_{j=0}^{x-1}\prod_{\ell=1}^{j} \mathbb{P}\left[ i(v+\ell) \text{ is  accepted} \,\middle|\, i(v), i(v+1), \ldots, i(v+\ell-1) \text{ is accepted} \right] \\ &\quad\quad \cdot \mathbb{P}\left[ i(v+j+1) \text{ is rejected} \,\middle|\, i(v), i(v+1), \ldots, i(v+j) \text{ is accepted} \right] \\&=\sum_{j=0}^{x-1}\prod_{\ell=1}^{j} \mathbb{P}\left[ i(v+\ell) \text{ is accepted} \,\middle|\,  i(v+\ell-1) \text{ is accepted} \right] \\ &\quad\quad \cdot \mathbb{P}\left[ i(v+j+1) \text{ is rejected} \,\middle|\, i(v+j) \text{ is accepted} \right] \\&= \sum_{j=0}^{x-1} (1-\beta)^{j}\beta \\&= \beta \frac{1-(1-\beta)^x}{1-(1-\beta)} \\&= 1-(1-\beta)^x.
\end{align*}

Since $1-(1-\beta)^x$ is increasing in $x$, we have for $x < w_{\alpha}$,
\begin{equation} \label{eq:tempcons2}
1-(1-\beta)^x<1-(1-\beta)^{w_{\alpha}} = 1-\left(1-\left(1-(1-\alpha)^{1/w_{\alpha}}\right)\right)^{w_{\alpha}} = 1-(1-\alpha)=\alpha.
\end{equation}

As $\alpha d(u,v)$ is also an increasing function with $x=u-v$ and $\alpha d(u,v)=\alpha$ when $x=1$, by Equation \eqref{eq:tempcons2}, we have
\[
1-(1-\beta)^x<\alpha \leq \alpha d(u,v).
\]
Therefore, we have for any $u, v \in [n_i]$, the decreasing grace period satisfies Equation \eqref{eq:tempcons1}. If $u<v$, a symmetric argument also holds and we omit the proof.

Define event $E$ as $E:= \{\text{no resource is depleted in the time interval } [0, T] \}$. The above analysis shows that if $E$ happens, then the fairness constraint holds. Next, we show that $E$ happens with probability at least $1-\delta$ by showing that the complement of $E$ happens with probability at most $\delta$.
\begin{align}
    \mathbb{P}(E ^{\mathsf{C}}) \nonumber&\leq \mathbb{P}(\text{more than $\gamma n$ customers are accepted in the grace period}) \\ \label{eq:pfthm1eq1}& \leq \mathbb{P}(\exists i \in [n], \text{ s.t. the number of type } i \text{ customers accepted in the grace period } \geq \gamma) \\  \nonumber &=1 - \mathbb{P}(\exists i \in [n], \text{ s.t. the number of type } i \text{ customers accepted in the grace period } < \gamma) \\ \label{eq:pfthm1eq2}&= (1-\beta)^{\gamma} \\ \nonumber &= (1-\beta)^{\log_{1-\beta}\delta} = \delta,
\end{align}
where \eqref{eq:pfthm1eq1} is due to the pigeonhole principle, and \eqref{eq:pfthm1eq2} is because for each type $i$, the number of type $i$ customers accepted after the grace period starts is a geometric random variable with success probability $\beta$. Therefore, the cdf is $1-(1-\beta)^{\gamma}$.

To complete the proof of the theorem, we need to show that the revenue loss is bounded by $\frac{\bar a}{\underline a} n \gamma \bar r$. In the worst case for the revenue, the algorithm rejects all customers after $t_1$. By the definition of $t_1$, the remaining units for each $j$ are $\bar a n \gamma$. Since $\bar a n \gamma$ units of resource can serve at most $\frac{\bar a}{\underline a} n \gamma$ customers, we have that the revenue loss is at most $\frac{\bar a}{\underline a} n \gamma \bar r$.
\Halmos
\end{proof}
}

We highlight the intuition of why the revised \FCFS is $(\alpha, \delta)$-fair. In the revised algorithm, every customer arriving during $[1,t_1 -1]$ is accepted, but starting from the period $t_1$, the acceptance probability for each customer type progressively decreases. This decrease is conditioned on the preceding customer's realization. { This implies that if the resources are not depleted, the probability that customers of the same type that arrive one after the other get different treatment is bounded by $\beta=1-(1-\alpha)^{1/w_{\alpha}}$. The choice of $\beta$ ensures that, for any pair of customers $i(u)$ and $i(u + x)$ with $x \in [w_{\alpha}]$, the conditional probability that $i(u)$ is accepted and $i(u + x)$ is rejected is at most $\alpha$, thereby satisfying the fairness constraint. For $x > w_{\alpha}$, by the first and third properties of $d(u,v)$ in Definition \ref{def:ifmetrics}, we have $\alpha \cdot d(u, u+x) \geq 1$, so the fairness constraint is automatically satisfied since the probability is always upper bounded by $1$.} At round $t_1$, the remaining capacity of each resource $j \in [L]$ is at least $\bar a n \gamma$. Given that the stopping time (i.e., the time that the algorithm stops accepting them) for each customer type is a geometric random variable, we can apply standard concentration bounds to demonstrate that, with a high probability, all geometric random variables will stop before $\bar a n \gamma$ resource $j$ are exhausted for each $j \in [L]$. This concludes the proof that the revised algorithm is $(\alpha,\delta)$-fair.

To obtain a better understanding about the various parameters describing the revenue loss and the efficiency of our proposed algorithm, note that if $\delta=1/ m $, and given that $\underline a$,  $\bar a$, $n$ and $\bar r$ are constants, Theorem \ref{thm:FCFSnetwork} implies that the loss in revenue relative to the \FCFS algorithm is bounded by $O(\log m)$. The full proof of Theorem \ref{thm:FCFSnetwork} can be found in Appendix \ref{append:gpd}.


\subsection{Increasing Grace Period}

We conclude this section with a related concept, the \textit{increasing grace period}.

\begin{definition} \label{def:increasing_gp}
The \textit{increasing grace period $[t_1(i),t_2(i)]$} for type $i$ customers is defined as: for any customer $i(k)$ that arrives within the interval from $t_1(i)$ to $t_2(i)$, customer $i(k)$ is \emph{accepted} with probability $\beta$ if customer $i(k-1)$ was \emph{rejected}, and customer $i(k)$ is \emph{accepted} if customer $i(k-1)$ was \emph{accepted}.
\end{definition}

The intuition behind the increasing grace period design is that, after a sequence of customer rejections, when the algorithm is ready to start accepting customers again in order to maintain individual fairness, it should not do so abruptly. Instead, the acceptance probability should incrementally increase, with the degree of increase conditioned on the previous customer's outcome.

\section{Individually Fair RM Algorithms under Stochastic Arrivals} \label{sec:stochatic}


\begin{figure}[!tb]
\center
\includegraphics[width=\textwidth]{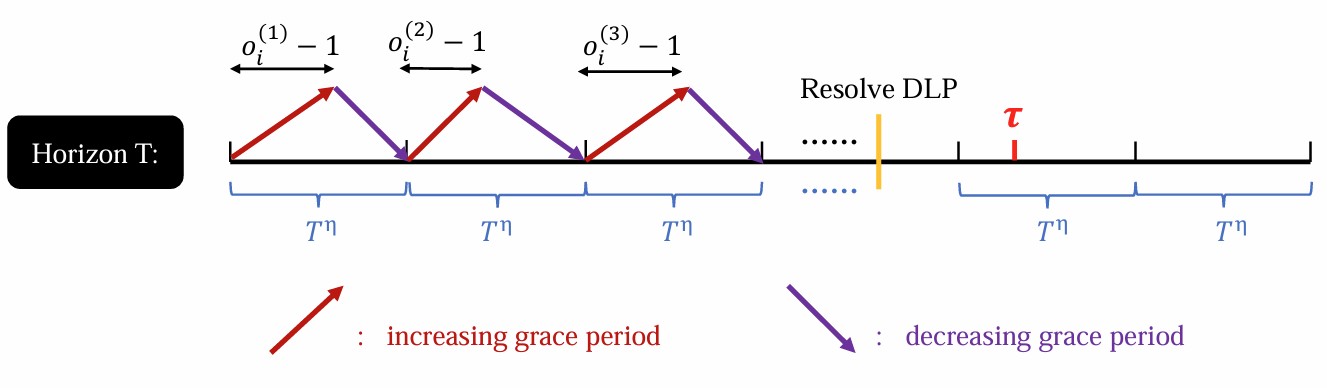}
\caption{\centering Decisions from Algorithm \ref{alg:RDLPrevise} on type $i$ customers} 
\label{fig:RDLP}
\end{figure}

{
In this section, we focus on making four algorithm structures mentioned in Section \ref{subsec:sto} (\AlgDLP, \AlgRDLP, \AlgSBPC, \AlgDBPC) individually fair without incurring significant additional regret. 
Further, as Algorithm \AlgDLP is a variant for \AlgRDLP (resolve 0 times), we analyze \AlgDLP and \AlgRDLP together. All primal and dual methods can be adapted to our revised algorithm structure since we focus only on the (re)solving of the primal/dual LP, rather than the specific decision-making process used by the original algorithms. That is, as long as different algorithms resolve the primal/dual LP at the same time, our revised approach remains consistent, regardless of their individual decision structures.
}


\subsection{Algorithm \AlgDLP and \AlgRDLP} \label{subsec:dlp}


Here, we introduce the Grace Period Enhanced \AlgRDLP algorithm (\GPRDLP). The regret of the original \AlgRDLP algorithm is $O(T^{\eta})$, where $\eta$ is a parameter that depends on the frequency of resolving the DLP. For example, $\eta=1/2$ corresponds to resolving zero times (i.e., Algorithm \AlgDLP), while $\eta = 1/3$ corresponds to resolving one or more times. 


We begin by uniformly partitioning the time horizon $[0,T]$ into $K$ segments. Each time segment has length $\max\{T^{\eta}, 2\bar a n \gamma\}$, where $\gamma = \log_{1-\beta}\delta$. $\bar a n \gamma$ is the length of the grace period. Then, we have $K=\min\{T^{1-\eta},T/(\bar a n \gamma)\}$. We split the discussion into two cases: (1) $T^{\eta}>2\bar a n \gamma$; (2) $T^{\eta} \leq 2\bar a n \gamma$.

\textit{Case 1: $T^{\eta}>\bar a n \gamma$. } In this case, the regret is greater than the grace period length, which means that we are not resolving for many times. We illustrate this case by discussing the resolving once algorithm. We solve the DLP with the initial capacity vector at $t=0$, and solve the DLP with the remaining capacity vector at $t=t^{*}$, where $t^*$ is the re-solving time point of the original \AlgRDLP algorithm. In the first time segment ($k=1$), for each type $i$ customer, we sample a random variable $y_i \sim \text{Bin}(\lambda_iT^{\eta},\frac{x_i^{*}}{\lambda_i})$. Our objective is to accept approximately $y_i$ customers of type $i$ in this interval, akin to the original \AlgRDLP algorithm's acceptance of $\text{Bin}(\lambda_iT^{\eta},\frac{x_i^{*}}{\lambda_i})$ customers of type $i$. However, different from \AlgRDLP's probabilistic allocation, we initially accept $y_i-\bar a n \gamma -1$ customers of type $i$, where $\gamma = \log_{1-\beta} \delta$. Subsequently, we implement a decreasing grace period for type $i$ customers until the end of the first interval. According to Theorem \ref{thm:FCFSnetwork}, this ensures the individual fairness metrics while accepting $z_i^{(1)} = \Theta(\log T)$ fewer customers of type $i$ compared to Algorithm \AlgRDLP.

Starting from the second time segment $(k \geq 2)$, due to the probable rejection of the last few customers in segment $k-1$, an immediate acceptance of customers in segment $k$ is unfair. Thus, we introduce an increasing grace period, progressively shifting from rejection to acceptance. Notably, in segment $k-1$, as $z_i^{(k-1)}$ fewer customers of type $i$ are accepted compared to Algorithm \AlgRDLP, we compensate by accepting $z_i^{(k-1)}$ additional type $i$ customers in segment $k$. { Here, $z_i^{(k)}$ are independently distributed for different $k$, but in segment $k$, we can use the realization of $z_i^{(k-1)}$ to draw} $y_i \sim z_i^{(k-1)}+\text{Bin}(\lambda_iT^{\eta},\frac{x_i^{*}}{\lambda_i})$ and implementing an increasing grace period until the arrival of the $(y_i-\bar a n \gamma)^{\text{th}}$ customer of type $i$. Subsequently, a decreasing grace period is applied until the conclusion of the time segment. 

Lastly, following the methodology introduced in Section \ref{subsec:classicalnrm}, as a resource approaches depletion, an additional decreasing grace period is introduced to conclude the algorithm. The details can be found in Algorithm \ref{alg:RDLPrevise}, and the intuition is illustrated in Figure \ref{fig:RDLP}.

{
\textit{Case 2: $T^{\eta} \leq \bar a n \gamma$.} Suppose the resolving algorithm achieves a regret smaller than the grace period length, such as the \( O(1) \) regret bound shown in \cite{vera2019bayesian, bumpensanti2020re}. In this case, we partition the time horizon into segments of length at least twice the grace period length. However, if the algorithm resolves too frequently—e.g., \cite{vera2019bayesian} resolves the DLP at every time step—this frequency is too high for a  grace period to be maintained.

To address this, we accept a modest sacrifice in revenue in order to maintain the grace period design and enforce the fairness constraint. Consider, for example, the algorithm proposed by \cite{vera2019bayesian}, which resolves the deterministic linear program (DLP) at every time step. Our enhanced version only resolves the DLP at periodic time points \( \eta \bar{a} n \gamma \), where \( \eta \in \left[ T / (\bar{a} n \gamma) \right] \). This modification reduces the computational complexity and ensures fairness through the grace period design, albeit with the possibility of slightly reduced performance compared to the original algorithm. Nevertheless, similar to the analysis in Case 1, the regret of the enhanced algorithm remains bounded by the grace period length \( O(\bar{a} n \gamma) = O(\log_{1-\beta} \delta) \).

}


\begin{algorithm}[h]
\caption{\GPRDLP}
\label{alg:RDLPrevise}

\DontPrintSemicolon 
\KwIn{Preference matrix $\textbf{A}$. Arriving rate vector $\mathbf{\lambda}$. Capacity vector $\textbf{m}$. Resolving time point $t^{*}$. Regret bound of input R-DLP algorithm $O(T^{\eta})$. Fairness parameters $(\alpha,\delta)$. $\gamma=\log_{1-\beta}\delta$.}
\KwOut{Optimal solution vectors $\mathbf{x}^{*}$ and $\tilde{\mathbf{x}}^{*}$.}
\textbf{Initialize:} Vector $\mathbf{z}^{(0)}=0$. $K=T/(\max\{T^{\eta},2\log_{1-\beta}\delta\})$.\;
Solve DLP based on the initial capacity, and denote the optimal primal solution as $\mathbf{x}^{*}$.\;
\For{$k \in \{1,2,...,\lfloor t^{*}/(\max\{T^{\eta},2\log_{1-\beta}\delta\})\rfloor \}$}{
    Run Algorithm \ref{alg:RDLPreviseaux} with $t_1=(k-1)(\max\{T^{\eta},2\log_{1-\beta}\delta\})$, $t_2=k(\max\{T^{\eta},2\log_{1-\beta}\delta\})$, $\mathbf{z}=\mathbf{z}^{(k-1)}$, $\mathbf{x}^{*}$. Get the output $\mathbf{z}^{(k)} \gets \mathbf{w}$.\;
}
Solve DLP based on the remaining capacity and time, and denote the optimal solution as $\tilde{\mathbf{x}}^{*}$.\;
\For{$k \in \{\lceil t^{*}/(\max\{T^{\eta},2\log_{1-\beta}\delta\}) \rceil,\lceil t^{*}/(\max\{T^{\eta},2\log_{1-\beta}\delta\}) \rceil +1,\ldots, K  \}$}{
    Run Algorithm \ref{alg:RDLPreviseaux} with $t_1=(k-1)(\max\{T^{\eta},2\log_{1-\beta}\delta\})$, $t_2=k(\max\{T^{\eta},2\log_{1-\beta}\delta\})$, $\mathbf{z}=\mathbf{z}^{(k-1)}$, $\tilde{\mathbf{x}}^{*}$. Get the output $\mathbf{z}^{(k)} \gets \mathbf{w}$.\;
}
\end{algorithm}

\begin{algorithm}[h]
\caption{Auxiliary Algorithm for Algorithm \ref{alg:RDLPrevise}}
\label{alg:RDLPreviseaux}

\DontPrintSemicolon 
\KwIn{Arriving rate vector $\mathbf{\lambda}$. Capacity vector $\textbf{m}$. Time horizon $(t_1,t_2)$. Fair parameters $(\alpha,\delta)$. Vector $\mathbf{z}$. Optimal primal solution $\mathbf{x}^{*}$.}
\KwOut{Vector $\mathbf{w}$, indicating the sum of rejected type $i$ customers.}
Sample $y_i \sim z_i+\text{Bin}((t_2-t_1)\lambda_i,\frac{x^{*}_i}{\lambda_i})$.\;
\If{$t_1 \neq 1$}{
    Give an increasing grace period $[t_1,o_i-1]$ to type $i$ customer, where $o_i$ is the time when the $(y_i-\bar a n\gamma)^{\text{th}}$ type $i$ customer arrives.\;
}
Give a decreasing grace period $[o_i,t_2]$ to type $i$ customer.\;
Return vector $\mathbf{w}$, whose $i^{\text{th}}$ element is the sum of number of rejected type $i$ customers in the increasing grace period $[t_1,o_i-1]$ and in the decreasing grace period $[o_i,\tilde{o_i}]$, where $\tilde{o_i}$ is the time when the $y_i^{\text{th}}$ type $i$ customer arrives.\;
\For{$t \in \{t_1,t_1+1,\ldots,t_2\}$}{
    \If{$\min_{j \in [L]}m_j(t) - \bar a n \gamma \leq 0$}{
        Give a decreasing grace period $[t,T]$ to all types of customers.\;
    }
}
\end{algorithm}

{
\begin{theorem} \label{thm:RDLPrevise}
{Given any $\alpha, \delta \in (0,1)$, and any input \AlgRDLP algorithm that attains a regret of $O(T^{\eta})$, Algorithm \GPRDLP: (i) is $(\alpha,\delta)$-fair, and (ii) incurs regret $O(\max\{T^{\eta}, \log_{1-\beta} \delta\})$.} 
\end{theorem}

\begin{remark}
    Specifically, when setting \( \delta = 1/T \), if \AlgRDLP incurs regret \( T^{\eta} = \omega(\log T) \), then \GPRDLP maintains the same regret bound of \( T^{\eta} \). Conversely, if \AlgRDLP has regret \( T^{\eta} = O(\log T) \), then \GPRDLP achieves a regret bound of \( \Theta(\log T) \). This result shows that the revenue loss introduced by \GPRDLP remains minimal.
\end{remark}
}


\paragraph{Proof Sketch.}
Here, we provide a proof sketch for the case where \RDLP \text{ }re-solves the LP once. Algorithm \GPRDLP is $(\alpha, \delta)$-fair since it introduces a grace period every time that a decision transitions from acceptance to rejection and vice versa. 
To confirm that the regret of Algorithm \GPRDLP is $O(T^{\eta})$, we have
\begin{align*}
\Reg &= \HO - \Revp = \HO - \RDLP + \RDLP -\Revp \\&= O(T^{\eta}) + \RDLP -\Revp,
\end{align*}
where $\RDLP$, $\Revp$ are revenue generated by Algorithm \AlgRDLP and Algorithm \GPRDLP respectively. It suffices to show that 
\[
\RDLP -\Revp = O(\max\{T^{\eta}, \log 1/\delta\})
\]

Let $\tau$ be a random variable indicating the stopping time at which some resource becomes exhausted (triggering the termination of the algorithm), with $\tau$ belonging to time segment $k(\tau)$. Ignoring the grace period and accepting $y_i$ type $i$ customers in each segment, the expected revenue of Algorithm \GPRDLP almost matches that of Algorithm \AlgRDLP for $t \in [0,(k(\tau)-1)T^{\eta}]$.

Furthermore, although the application of both increasing and decreasing grace periods in each time segment may lead to the additional loss of up to $\Theta(\log 1/\delta)$ customers of each type, we accept an equivalent number of customers of each type in the following time segment. By doing this, before $t^{*}$, the expected difference in the number of accepted customers of each type between Algorithm \AlgRDLP and Algorithm \GPRDLP stems only from the last time segment before $t^{*}$, amounting to $\Theta(\log 1/\delta)$. Therefore, the expected difference in remaining capacity between Algorithm \AlgRDLP and Algorithm \GPRDLP is bounded by $\Theta(\log 1\delta)$ at time $t^{*}$. As such, the disparity in the re-solved optimal solutions at $t^{*}$ between Algorithm \AlgRDLP and Algorithm \GPRDLP is negligible. Consequently, we eliminate all revenue loss due to the grace period for all $k(\tau)-1$ time segments. For the time segment $k(\tau)$ of size $T^{\eta}$, we may lose a maximum of $O(T^{\eta})$ in revenue compared to the \AlgRDLP 
algorithm. Hence, we have
\[
\RDLP-\Revp = O(T^{\eta})+O(\log 1/\delta),
\]
The full proof can be found in Appendix \ref{append:Sec5}. 
\qed

\subsection{Algorithm \AlgSBPC} \label{subsec:bid}

In the algorithm \AlgSBPC, an \FCFS approach is employed for each customer type. 
Drawing intuition from our treatment of \FCFS, to guarantee fairness, we only need to offer a decreasing grace period to those types whose bid prices are greater than the aggregated bid prices. From Theorem \ref{thm:FCFSnetwork}, the extra revenue loss of the fair variant is $O(\log 1/\delta)$ for any $\alpha, \delta \in (0,1)$. 

\begin{theorem} \label{thm:FBPCrevise}
Given any $\alpha, \delta \in (0,1)$, and any input \AlgSBPC algorithm with a regret of $O(\sqrt{T})$, by giving a decreasing grace period $[\min_{j \in [L]}m_j - \bar a n \gamma,T]$ to each type $i$ customer (for $i$ such that $r_i > \sum_{j=1}^L \theta_j^{*} A_{ij}$): (i) the algorithm is $(\alpha,\delta)$-fair, (ii) incurs regret $O(\max\{\sqrt{T},\log_{1-\beta} \delta\})$. 
\end{theorem}

\begin{remark}
    Specifically, for $\delta = O\left( (1-\alpha)^{\sqrt{T}}\right)$, the regret is $O(\sqrt{T})$.
\end{remark}

{
\subsection{Algorithm \AlgDBPC} \label{subsec:algdbpc}

In the Algorithm \AlgDBPC, instead of using a fixed bid price, it dynamically updates the bid price with some first order optimization methods, such as online gradient descent or online mirror descent. The details of \AlgDBPC can be found in Algorithm \ref{alg:ogd}. \AlgDBPC can be markedly unfair. To see this, let us simplify the setting and assume that there is only one type of customer for $t = 1,2,\ldots,T$. Also assume that there is only one type of resource and each arriving customer requests a single unit of this resource. Subsequently, if we follow steps 15 and 16 in Algorithm \ref{alg:ogd}, we find that $\theta^{(t)}$ is a deterministic process that can be described using the following formula:
\begin{align} \label{eq:deterministicogd}
\theta^{(t+1)} &= \left\{\begin{matrix}
\; \theta^{(t)}-\frac{D}{G\sqrt{T}}\frac{m(t)}{T}+\frac{D}{G\sqrt{T}}, & \text{ if $\theta^{(t)} < r$ }, &\\  
\; \theta^{(t)}-\frac{D}{G\sqrt{T}}\frac{m(t)}{T}, & \text{ if $\theta^{(t)} \geq r$ }, & \\ 
\end{matrix}\right. 
\end{align}
Here, $r \in (0,1)$, $m(t)$ denotes the number of resources remaining at time $t$, with $m(0)=m$, and
\begin{align*} 
m(t+1) &= \left\{\begin{matrix}
\; m(t)-1, & \text{ if $\theta^{(t)} < r$ }, &\\  
\; m(t), & \text{ if $\theta^{(t)} \geq r$ }, & \\ 
\end{matrix}\right. 
\end{align*}

The deterministic process $\theta^{(t)}$, as defined in Equation \eqref{eq:deterministicogd}, allows for the precise computation of how frequently decisions switch between accepting and rejecting adjacent customers. Regardless of the constants $D$, $G$, and $r$, this decision-making flips $\Theta(T)$ times between acceptance and rejection. This behavior highlights the algorithm's significant unfairness, as it implies that almost every customer is treated differently compared to their immediate predecessor or successor.

To modify \AlgDBPC to comply with $(\alpha,\delta)$-fairness while minimally impacting revenue, we first reject all arrivals in the period $t \in [0,T^{2/3}]$. Meanwhile, we operate \AlgDBPC as an auxiliary algorithm, documenting its decisions. We designate $u_i$ as the number of accepted type $i$ customers by \AlgDBPC within the $[0,T^{2/3}]$ time horizon, and $\Lambda_i$ as the count of arriving type $i$ customers within $[0,T^{2/3}]$. Then, we obtain that the Euclidean distance between the rate of accepted type $i$ customers by \AlgDBPC ($\frac{u_i}{\Lambda_i}$), and that by \AlgDLP ($\frac{x_i^{*}}{\lambda_i}$), is bounded with high probability. Thus, even without access to the $\lambda_i$ arrival rate and with the prohibition of linear programming, we can utilize the initial $T^{2/3}$ periods to approximate the optimal DLP solution $\textbf{x}^{*}$. Finally, for $t \in [T^{2/3},T]$, we execute the \GPRDLP Algorithm with a $T^{2/3}$ time segment length, setting $x_i^{*}=\frac{u_i}{T^{2/3}}$ and $\lambda_i=\frac{\Lambda_i}{T^{2/3}}$, culminating in a $(\alpha,\delta)$-fair algorithm with regret $O(T^{2/3} \log T)$ if $T^{2/3} \log T = \Omega(\log_{1/(1-\beta)}1/\delta)$. Otherwise, the regret is the length of the grace period, which is $O(\log_{1-\beta}\delta)$. Full details can be found in Algorithm \ref{alg:revisedogd}, and the following theorem summarizes the results. The proof can be found in Appendix \ref{append:Sec5}.

\begin{algorithm}[t!]
\caption{\GPDBPC} \label{alg:revisedogd}

\DontPrintSemicolon 
\KwIn{Preference matrix $\textbf{A}$. Reward vector $\textbf{r}$. Capacity vector $\textbf{m}$. OGD Parameter: $G$, $D$, $\bar \theta$. Fairness parameters $(\alpha,\delta)$.}
\KwOut{Update of dual variable $\mathbf{\theta}$ and acceptance decisions over time.}
\textbf{Initialize:} Dual variable $\mathbf{\theta}^{(0)} \gets \mathbf{0}$.\;
Reject all customers arriving between $t \in [0,\max\{T^{2/3},2\log_{1-\beta}\delta\}]$.\;
Denote $\Lambda_i$ as the number of arriving customers of type $i$ between $t \in [0,\max\{T^{2/3},2\log_{1-\beta}\delta\}]$, $i \in [n]$.\;
Run Algorithm \AlgDBPC as an auxiliary algorithm for $t \in [0,\max\{T^{2/3},2\log_{1-\beta}\delta\}]$.\;
Denote $u_i$ as the number of accepted customers of type $i$ by \AlgDBPC, $i \in [n]$.\;
Set $K=(T/\max\{T^{2/3},2\log_{1-\beta}\delta\})-1$, $\lambda_i=\frac{\Lambda_i}{\max\{T^{2/3},2\log_{1-\beta}\delta\}}$ and $x_i^{*}=\frac{u_i}{\max\{T^{2/3},2\log_{1-\beta}\delta\}}$, $\textbf{z}^{(0)}=\textbf{0}$.\;
\For{$k \in \{1,2,...,K \}$}{
    Run Algorithm \ref{alg:RDLPreviseaux} with $t_1=k\max\{T^{2/3},2\log_{1-\beta}\delta\}$, $t_2=(k+1)\max\{T^{2/3},2\log_{1-\beta}\delta\}$, $\mathbf{z}=\mathbf{z}^{(k-1)}$, optimal primal solution $\mathbf{x}^{*}$. Denote the output vector $\mathbf{w}$ as $\mathbf{z}^{(k)}$.\;
}
\end{algorithm}

\begin{theorem} \label{thm:reviseogd}
Given any $\alpha, \delta \in (0,1)$, and any input \AlgDBPC algorithm that attains a regret of $O(\sqrt{T})$, under a mild non-degenerate assumption (assumption 3 in \cite{agrawal2014dynamic}), Algorithm \GPDBPC: (i) is $(\alpha,\delta)$-fair, (ii) incurs regret $O(\max\{T^{2/3} \log T, \log_{1-\beta}\delta\})$.
\end{theorem}

\begin{remark}
    For $\delta = O\left((1-\beta)^{T^{2/3} \log T}\right)$, the regret of \GPDBPC is $O(T^{2/3} \log T)$.
\end{remark}
}

\section{Individually Fair RM Algorithms under Adversarial Arrivals} \label{sec:adversarial}

In this section, we focus on RM problems with adversarial arrivals, where an adversary dictates both the \emph{quantity} and the \emph{sequence} of arriving customers of each type. The two prevalent algorithm structures used in this setting are Booking Limit(\AlgBL) and \AlgN. Next, we introduce how to make these algorithms individually fair.

\subsection{Algorithm \AlgBL}

Algorithm \AlgBL sets quotas for each customer $i$, and then makes decisions in an \FCFS manner. Consequently, our method of ensuring individual fairness mirrors the modification of \FCFS in Section \ref{sec:gpd}. Algorithm \GPBL provides a decreasing grace period as the number of accepted type $i$ customers approaches the quota $b_i$, and extends this decreasing grace period to all customer types if the remaining capacity is almost depleted. Details can be found in Algorithm \ref{alg:bookingrevised}.

\begin{algorithm}[t!]
\caption{\GPBL}
\label{alg:bookingrevised}

\DontPrintSemicolon 
\KwIn{Preference matrix $\mathbf{A}$. Capacity vector $\textbf{m}$. Booking limit $\mathbf{b}$. $\gamma=\log_{1-\beta}\delta$.}
\KwOut{Decision on accepting customers over time.}
\textbf{Initialize:} Number of accepted customers $s_i=0$, $i \in [n]$.\;
\For{$t \in \{1,2,...,T\}$}{
    Observe customer of type $i \in [n]$.\;
    \uIf{$\min_{j \in [L]}m_j(t) - \bar a n \gamma \leq 0$}{
        Give a decreasing period $[t,T]$ to all types of customers.\;
    }
    \Else{
        \uIf{$s_i < b_i - \gamma$}{
            Accept the arriving customer, and set $s_i \gets s_i + 1$.\;
        }
        \Else{
            Give a decreasing period $[t,T]$ to type $i$ customers.\;
        }
    }
}
\end{algorithm}

\begin{theorem} \label{thm:revisebooking}
Given any $\alpha, \delta \in (0,1)$,resource capacity $\mathbf{m}$ where $m_j=\Theta(m)$ for $j \in [L]$, and any input \AlgBL algorithm that attains a competitive ratio of $C$, \GPBL achieves the following: (i) is $(\alpha,\delta)$-fair, (ii) reaches a competitive ratio of $C-O\left(\frac{\log_{1/(1-\beta)} 1/\delta}{m}\right)$.
\end{theorem}

Theorem \ref{thm:revisebooking} implies that \GPBL incurs negligible additional loss asymptotically, even under the worst case instance. The detailed proof can be found in Appendix \ref{append:Sec6}. To execute a competitive ratio analysis, we first observe that, given $b_i=\Theta(m)$, if the aggregate number of arrivals is $o(m)$, both the original and the revised booking limit algorithms yield a competitive ratio of $1$ since all customers are accepted.

In the event that the total number of arrivals is $\Omega(m)$, the offline optimal revenue is $\Theta(m)$ as the resource capacity scales with $\Theta(m)$. Additionally, any extra loss incurred by \GPBL in comparison to \AlgBL is from the decreasing grace period. In the worst scenario, we may lose at most $(\bar a+1) n \gamma = \Theta(\log 1/\delta)$ customers. Hence, 
\[
\inf_{I \in \mathcal{I}}\frac{\text{Rev}(\pi(I))}{\HO(I)} \geq \inf_{I \in \mathcal{I}}\frac{\text{Rev}(\textsf{BL}(I))-(\bar a+1) n  \gamma \bar r}{\HO(I)} = C-O\left(\frac{\log_{1/(1-\beta)} 1/\delta }{m}\right).
\]

\subsection{Algorithm \AlgN}

\AlgN ranks the customer types based on the revenue they bring to the supplier. Compared to \AlgBL, \AlgN also sets quotas and accepts customers in a \FCFS manner. Differently, each quota $b_i$ is the upper bound of number of customers can be accepted with type $i,i+1,\ldots,n$. As such, we still provide a decreasing grace period as the cumulative number of $i$, $i+1$, $\ldots$, $n$ customers approaches the quota $b_i$. The rationale behind this modification aligns precisely with that used in \AlgBL. Details can be found in Algorithm \ref{alg:revisenesting}.

\begin{algorithm}[t!]
\caption{\GPN}
\label{alg:revisenesting}
\DontPrintSemicolon 
\KwIn{Capacity $m$. Nesting quota $\mathbf{b}$.}
\KwOut{Decision on accepting customers over time.}
\textbf{Initialize:} Number of accepted customers $s_i=0$, $i \in [n]$. $\gamma=\log_{1-\beta}\delta$.\;
\For{$t \in \{1,2,...,T\}$}{
    Observe customer of type $i \in [n]$.\;
    \uIf{$m(t) -  n \gamma \leq 0$}{
        Give a decreasing period $[t,T]$ to all types of customers.\;
    }
    \Else{
        \uIf{$\sum_{j=i}^{n}s_j < b_{i} - n\gamma$}{
            Accept the arriving customer, and set $s_i \gets s_i+1$.\;
        }
        \Else{
            Give a decreasing period $[t,T]$ to type $i$ customers.\;
        }
    }
}
\end{algorithm}

\begin{theorem} \label{thm:revisenesting}
Given any $\alpha, \delta \in (0,1)$, and any input \AlgN algorithm that attains a competitive ratio of $C$, \GPN: (i) is $(\alpha,\delta)$-fair, (ii) reaches a competitive ratio of $C-O\left(\frac{\log_{1/(1-\beta)} 1/\delta}{m}\right)$.
\end{theorem}

The proof is almost identical to that of Theorem \ref{thm:revisebooking} and hence we omit it.

\section{Numerical Studies}

In this section, we conduct a numerical study to evaluate the efficacy of our grace period-enhanced algorithms, using \GPDLP as a representative example, in boosting revenue compared to baseline algorithm \AlgDLP. As discussed in Section \ref{sec:stochatic}, \AlgDLP is characterized as unfair while \GPDLP is considered fair. In this section, we investigate: 
\begin{enumerate}
    \item the expected cumulative revenue for both algorithms without accounting for the negative impacts of unfairness;
    \item the expected cumulative revenue when considering the negative impacts of unfairness.
\end{enumerate}

\xhdr{Data Description} To build our semi-synthetic dataset, we use a transnational dataset encompassing all transactions recorded from December 1, 2010, to December 9, 2011, for a UK-based, online retail business operating outside traditional store formats. The company specializes in unique gifts suitable for various occasions and primarily serves a wholesale customer base. The dataset is publicly available at \url{https://www.kaggle.com/datasets/tunguz/online-retail}. For our analysis, we extracted the following attributes from the dataset:
\begin{itemize}
    \item \textit{Description} (nominal): The name of the product (item) involved in the transaction.
    \item \textit{Quantity} (numeric): The number of units of each product (item) involved in a transaction. 
    \item \textit{UnitPrice} (numeric): The price per unit of the product in GBP.
    \item \textit{CustomerID} (nominal): A unique 5-digit identifying integer assigned to each customer. 
\end{itemize}
The dataset comprises over 500,000 transactions, featuring transactions from approximately 5,000 distinct customers, identified by unique \textit{CustomerID}s. Customers are classified into 20 distinct categories based on their purchasing patterns by K-means clustering, where $K=20$ and we cluster according to the \textit{Description} and \textit{Quantity} of products purchased. As the \textit{Description} is a nominal variable, we apply one-hot encoding to convert it into a numeric variable. Additionally, the dataset contains over 4,000 unique items listed under the \textit{Description} column. To clarify, we perform K-modes clustering, where $K=40$, to group all inventories into 40 distinct product types.

\xhdr{Simulation Setup} In this section, we focus on comparing the performances of the \AlgDLP Algorithm with the \GPDLP Algorithm (\GPRDLP only solves the DLP once). Initially, we construct the preference matrix $\textbf{A}$, where each element $A_{ij}$ represents the average quantity of product $j$ purchased by customer type $i$, for $i \in [20]$ (customer types) and $j \in [40]$ (product types). Furthermore, we introduce a reward vector $\textbf{r}$, with each component $r_i$ reflecting the average profit generated by customer type $i$, for $i \in [20]$. We estimate the arrival rate $\lambda_i$ for each customer type $i$ by dividing the total count of type $i$ customers by the overall number of customers recorded in the dataset. As the original dataset omits inventory details for each product, we deduce the available units for each product type $j$ by calculating the total consumed units and adjusting this figure with a random scaling factor $q_j = \text{Unif}(0.5,0.9)$ for each $j \in [40]$.
Utilizing these parameters, we proceed to solve the DLP as defined in equation \eqref{eq:defdlp}, and we can obtain the optimal primal solution $\textbf{x}^{*}$. We assess the model over a time horizon $T$, which varies from 50,000 to 500,000 in increments of 50,000, to examine the dynamics over different operational scales.

We examine two distinct scenarios in our analysis. In the first scenario, individual fairness is disregarded in terms of its impact on revenue. This means that perceived unfairness towards a customer does not influence the revenue generated. Through this scenario, we aim to evaluate the revenue differences between the two algorithms, supporting our claim that the revenue loss associated with \GPDLP is minimal, as asymptotically suggested in Theorem \ref{thm:RDLPrevise}. 

{
To implement \GPDLP, we evaluate the following parameter settings: \( (\alpha, \delta) = \{(0.05, 1/T), (0.05, 1/\sqrt{T}), (0.3, 1/T), (0.3, 1/\sqrt{T}) \} \). A smaller \( \alpha \) corresponds to a more balanced allocation, enhancing fairness. Since (based on our definition of fairness) a fair algorithm ensures fairness with at least probability \( 1 - \delta \), a smaller \( \delta \) increases confidence in the fairness guarantee.}

In the second scenario, we consider a situation where individual fairness directly affects revenue. According to our survey detailed in Section \ref{sec:experiment}, 56\% of participants reported experiencing perceived individual unfairness in real. Thus, we model a scenario where a customer who is rejected, while a same type customer immediately before or after them is accepted, perceives this as unfair with a 56\% likelihood. Reflecting the consequences of perceived unfairness, we integrate a revenue penalty based on our survey findings: 62\% of such customers would stop purchasing from the store, and 69\% would not recommend the store to others, signifying potential revenue loss. In this model, we reduce the total revenue by an amount $cr_i$ for each type $i$ customer who perceives unfairness, where $c$ represents the projected impact of the loss, essentially the expected frequency of visits by a type $i$ customer. For the purposes of our simulations, we set $c$ to values of $\{0.5, 1, 2\}$, to cover a range of possible loss intensities. We also fix $(\alpha,\delta)=(0.05,1/T)$.

\xhdr{Results} { In the first scenario, we disregard individual fairness to evaluate its impact on revenue outcomes. As shown in Figure \ref{fig:revenue1}, the difference in total revenue between the two algorithms across all parameter choices \( (\alpha, \delta) \) remains marginal, aligning with the \( O(\sqrt{T}) \) bound from Theorem \ref{thm:RDLPrevise}.  

Furthermore, the near overlap of curves corresponding to different parameter settings suggests that \GPDLP is robust to parameter choice. The intuition behind this low sensitivity is that in \GPDLP, \( (\alpha, \delta) \) only influences the length of the grace period, given by \( \log_{1-\alpha} \delta \). Even if this grace period is long—leading to a higher rejection rate of customers within a time segment \( [(k-1)\sqrt{T}, k\sqrt{T}] \)—the algorithm compensates for this loss in the subsequent segment \( [k\sqrt{T}, (k+1)\sqrt{T}] \) by accepting an equivalent number of additional customers. As a result, any shortfall in one segment is recovered in the next. The only potential revenue loss occurs in the final segment or in the period when resources are fully depleted. However, this loss in one segment is negligible in comparison to the total revenue, confirming that \GPDLP maintains strong performance regardless of parameter selection.
}
\begin{figure}[!tb]
\center
\includegraphics[width=0.48\textwidth]{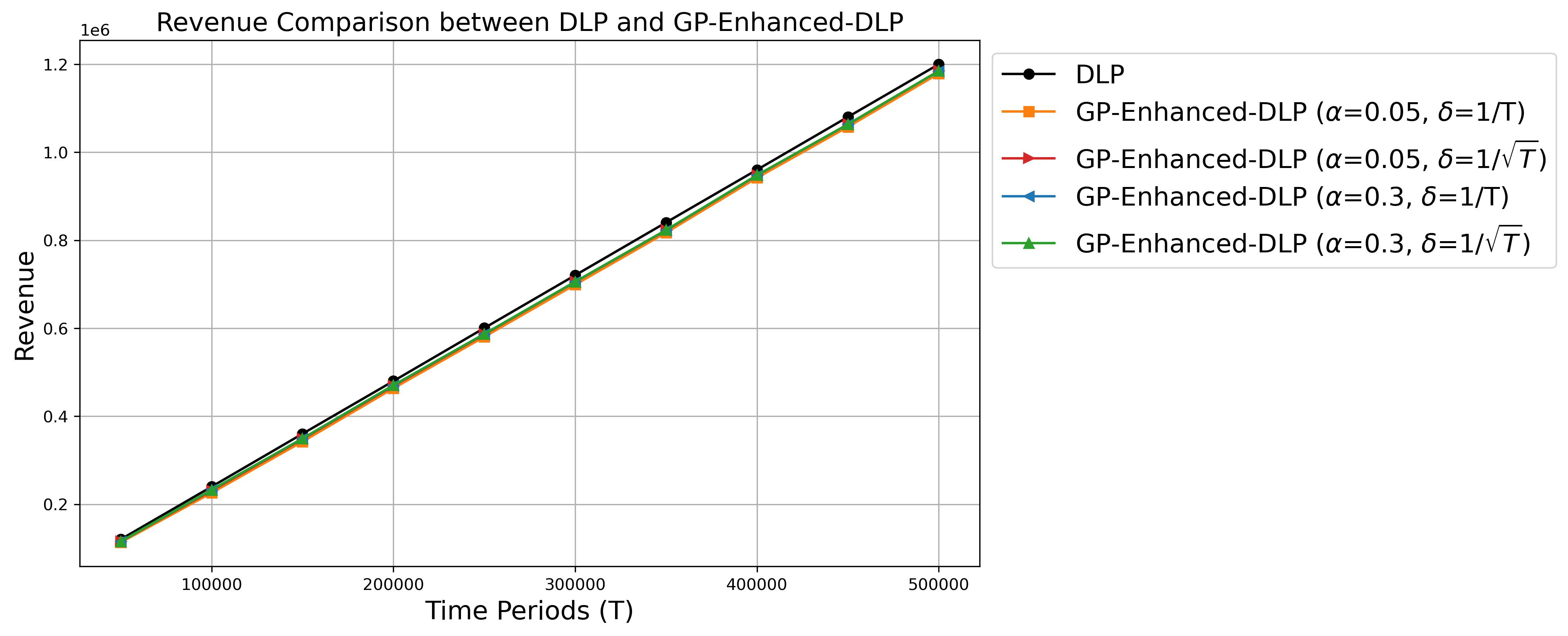}
\includegraphics[width=0.48\textwidth]{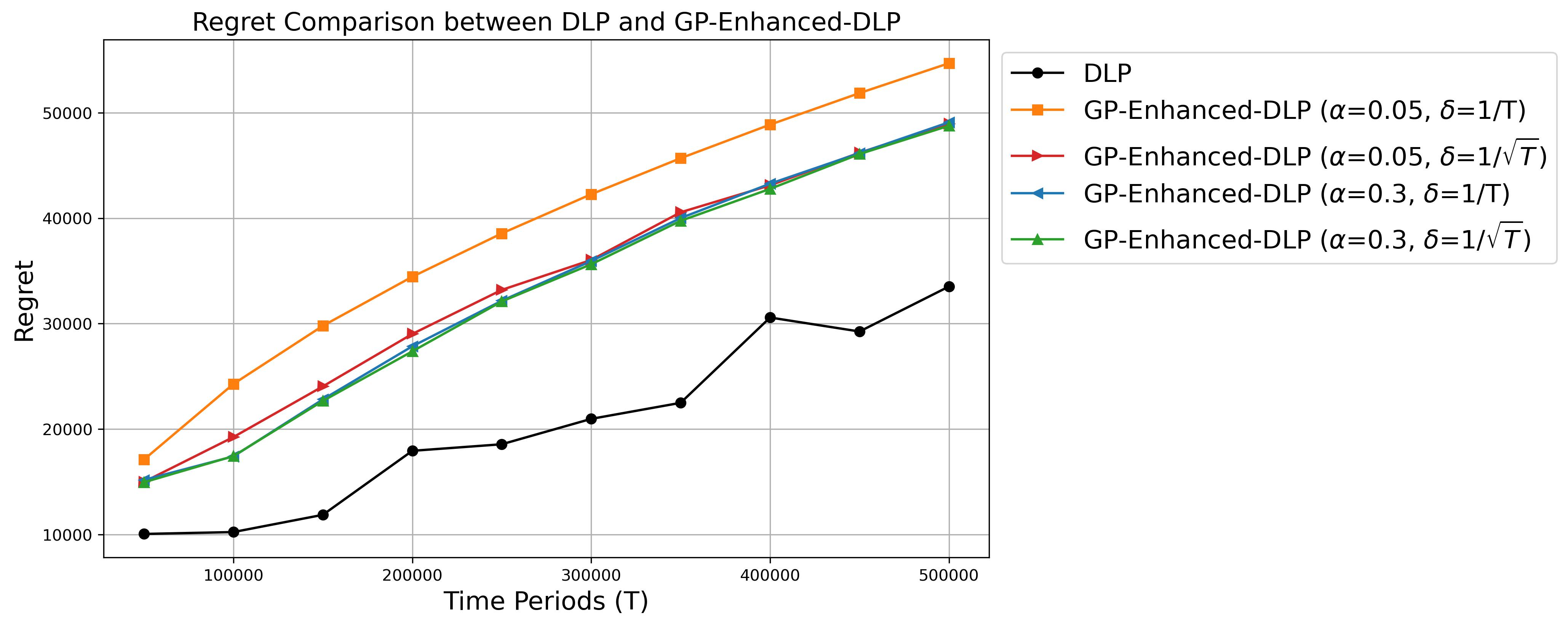}
\caption{\centering Comparison between \AlgDLP and \GPDLP (\textbf{left:} total revenue; \textbf{right:} regret).} 
\label{fig:revenue1}
\end{figure}

Conversely, in the second scenario, revenue loss is incurred whenever a type $i$ customer perceives unfairness, quantified as $cr_i$. The results, illustrated in Figure \ref{fig:revenue2}, indicate that \GPDLP consistently outperforms \AlgDLP under these conditions. Notably, the disparity in total revenue between the algorithms widens as $c$ (the scale of perceived unfairness) and $T$ (the time horizon) increase. This pattern underscores the significance of individual fairness in retail environments characterized by repeated customer transactions. Specifically, if a customer is expected to make three purchases (illustrated here when $c=2$), the revenue generated by \GPDLP can be nearly double that of \AlgDLP when T=500,000, highlighting the critical importance of addressing individual fairness to sustain and enhance revenue in settings where repeat business is common. Besides the numerical results, we can even derive some theory in this setting:

\begin{figure}[!tb]
\center
\includegraphics[width=0.8\textwidth]{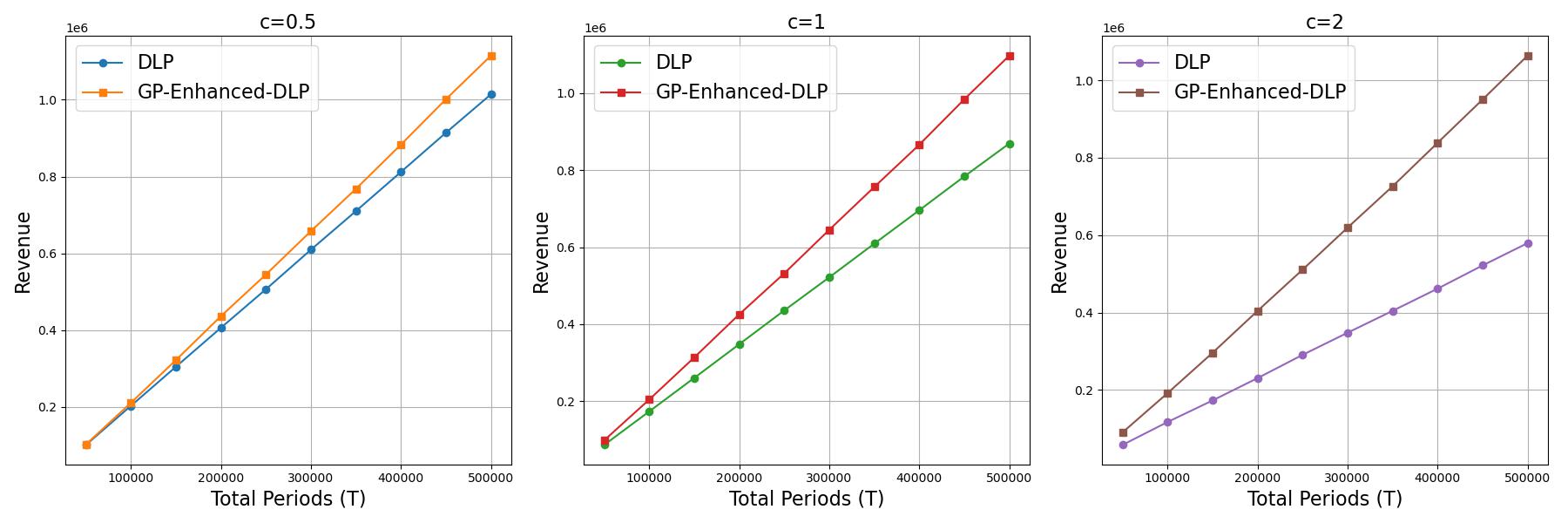}
\caption{\centering Comparison between \AlgDLP and \GPDLP (\textbf{left:} $c=0.5$; \textbf{middle:} $c=1$; \textbf{right:} $c=2$).} 
\label{fig:revenue2}
\end{figure}

\begin{proposition} \label{thm:numerical}
    Consider a time horizon $[1, T]$, during which a customer arrives randomly at each time $t \in [1, T]$. For each customer type $i$, where $i \in [1, n]$, if a customer of type $i$ is rejected while either the preceding or succeeding customer of the same type $i$ is accepted, this customer has a probability $q$ of detecting the discrepancy. Upon detection, this results in a total revenue loss of $cr_i$. Let $\tilde{R}(\AlgDLP)$ and $\tilde{R}(\pi)$ represent the expected cumulative revenue generated by Algorithm \AlgDLP and Algorithm \GPDLP, respectively. Then, the difference in expected cumulative revenue between the two algorithms is bounded by:
    \[
    \tilde{R}(\pi)-\tilde{R}(\AlgDLP)=O(T).
    \]
\end{proposition}

\proof{Proof of Proposition \ref{thm:numerical}}
Recall that $R(\AlgDLP)$ and $R(\pi)$ denote the expected cumulative revenue generated by Algorithm \AlgDLP and Algorithm \GPDLP respectively without considering the unfair effect. By Theorem \ref{thm:RDLPrevise}, we have
\[
R(\AlgDLP) - R(\pi) = O(\sqrt{T}).
\]

Since Algorithm \AlgDLP accepts each type $i$ customer with probability $\frac{x_i^{\star}}{\lambda_i}$, for any type $i$ customer, the probability that the customer is rejected but one of the preceding or succeeding type $i$ customer is accepted is $2\frac{x_i^{\star}}{\lambda_i}\left(1-\frac{x_i^{\star}}{\lambda_i}\right)$. Therefore, in expectation, each type $i$ customer will cause a revenue loss of $2qcr_i\frac{x_i^{\star}}{\lambda_i}\left(1-\frac{x_i^{\star}}{\lambda_i}\right)$, which is a positive constant. Therefore, we can obtain
\[
R(\AlgDLP) - \tilde{R}(\AlgDLP) = O(T).
\]

For Algorithm \GPDLP, in each $\sqrt{T}$ time segments, at most two customers can feel unfair, where one is in the decreasing grace period and the other is in the increasing one. Therefore, the expected revenue loss is at most $2\sqrt{T}qc\bar r$. Therefore,
\[
R(\pi)-\tilde{R}(\pi)=O(\sqrt{T}),
\]
which implies that 
\[
\tilde{R}(\pi) - \tilde{R}(\AlgDLP) = O(T).
\]
\Halmos
\endproof

\section{Survey} \label{sec:experiment}

In this section, we delineate the primary findings derived from our survey which was conducted between October and November 2023 on Prolific\footnote{\url{https://www.prolific.com/}}. We collected responses from 150 respondents. The primary goal was to evaluate three hypotheses:

\begin{enumerate}
\item[{\bf H1.}] When buyers observe ``unfair'' (defined according to our Definition \ref{def:ifmetrics}) treatment in a store prior to shopping there, their purchasing intent diminishes.
\item[{\bf H2.}] A customer who has previously made a purchase from a store and subsequently perceives unfair treatment is likely to exhibit adverse future buying behaviors.
\item[{\bf H3.}] Individuals perceive disparate treatments occurring in the near term as more unfair compared to those in the distant past.
\end{enumerate}

Comprehensive demographic data is available in Appendix \ref{subapp:demo}. The respondents were all based in the United States, fluent in English, and mostly gender-balanced (about a half of them are identifying as male and about half as female). The majority identified as white. To ensure the quality of the responses, we incorporated three attention checks. Only the data from respondents who correctly answered at least two of these questions were considered valid. From the initial 150 participants, 140 responses met this criterion.

The goal of our survey was to assess individuals' buying behavior when they observed other customers receiving discounts of roughly 10\% but they themselves did not. We recognize the significance of the specific discount value in this context; people's valuation of money varies, and a discount perceived as negligible by some might not induce feelings of unfairness to them. To address this, we conducted a comparative survey, wherein the context incorporated a more substantial discount (specifically, 30\%). This supplementary survey, which also comprised of 150 participants, is documented in Appendix \ref{subapp:30}. In all the scenarios we considered, we posit that there are two stores \emph{TrendThread} and \emph{Stitch\&Style} which the survey respondent is considering. The products sold by the two places are \emph{identical} in terms of quality and overall characteristics. The only difference is with regards to the pricing and the discounts that each stores offers. Additionally, at the survey's conclusion, we inquired if participants had encountered or noticed such unfairness in the real world, and 56\% of the respondents affirmed they had.

\subsection{Buying Behavior and Observed Pre-Purchased Unfairness}

In this section, we seek to validate hypothesis {\bf H1}. We begin by presenting two distinct scenarios designed to evoke feelings of unfairness \emph{before} making a purchase among consumers.

The first scenario revolves around a consumer's interaction with social media posts describing others' purchasing experiences. Imagine you come across a social media post wherein a user claims to have a $10\%$ discount on a hoodie from the store \textit{TrendThread}. You promptly visit the online shopping website, but no such discount is extended to you from \textit{TrendThread}. \textit{Stitch\&Style} is offering the identical hoodie at exactly the same price as \textit{TrendThread}. We aim to capture whether consumers, upon encountering the absence of a promised discount at \textit{TrendThread}, opt to transfer their patronage to \textit{Stitch\&Style} for their purchase\footnote{A concise version is presented here for clarity, with the detailed questions available in Appendix \ref{subapp:question}}.

\begin{figure}[!tb]
\center
\includegraphics[width=0.7\textwidth]{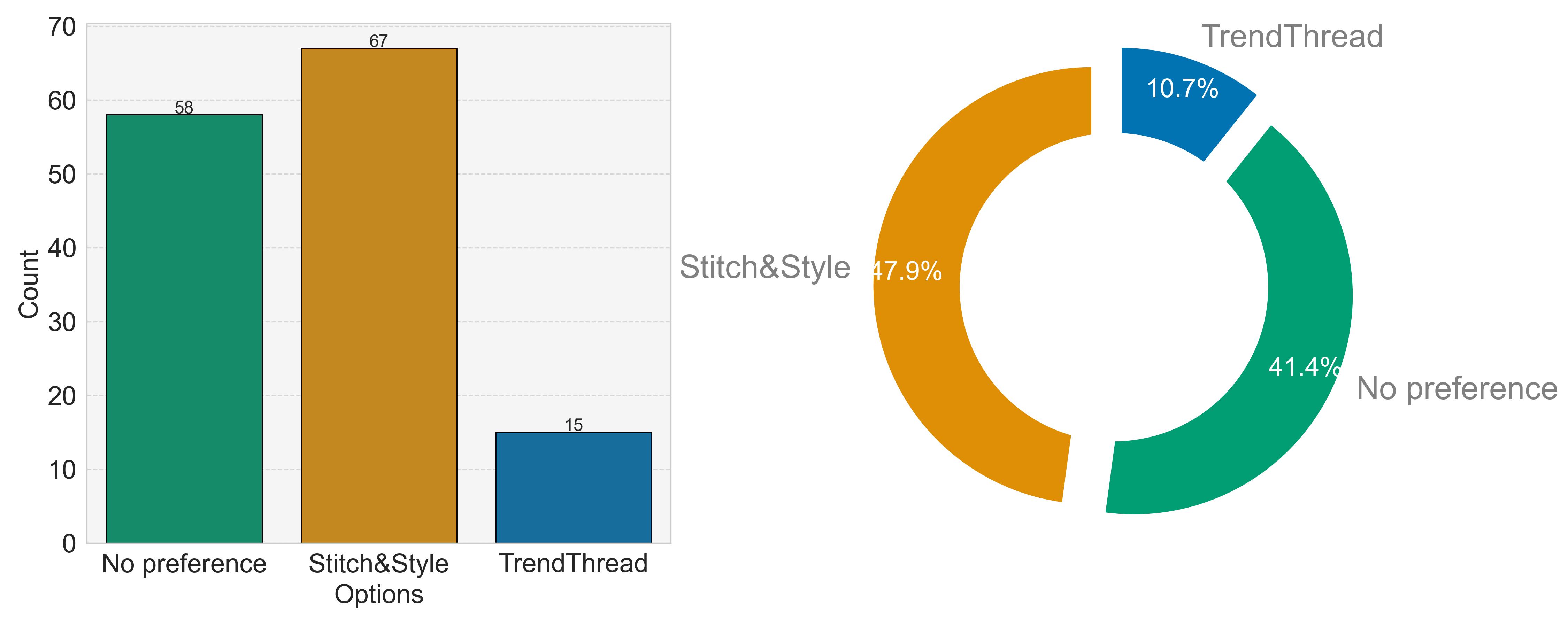}
\caption{\centering Distribution of preferences in response to exposure to discount-related posts on social media} 
\label{fig:Q1}
\end{figure}

As illustrated in Figure \ref{fig:Q1}, $47.9\%$ of the respondents perceive the different treatment as a significant factor influencing their purchasing decisions. Within the $41.5\%$ who opted for \textit{No Preference}, there emerges a question: is their indifference attributed to the perceived unfairness or the magnitude of the discount ($10\%$)? In Figure \ref{fig:Q1(2)}, of those participants who initially selected \textit{No Preference}, more than $20\%$ would divert their purchase to an alternative store if faced with a more substantial discount disparity. From this observation, we infer that approximately $30\%$ of the participants remain indifferent to differential treatment. In contrast, close to $60\%$ would reconsider their store preference upon perceiving some form of unfairness prior to purchase.

\begin{figure}[!tb]
\center
\includegraphics[width=0.7\textwidth]{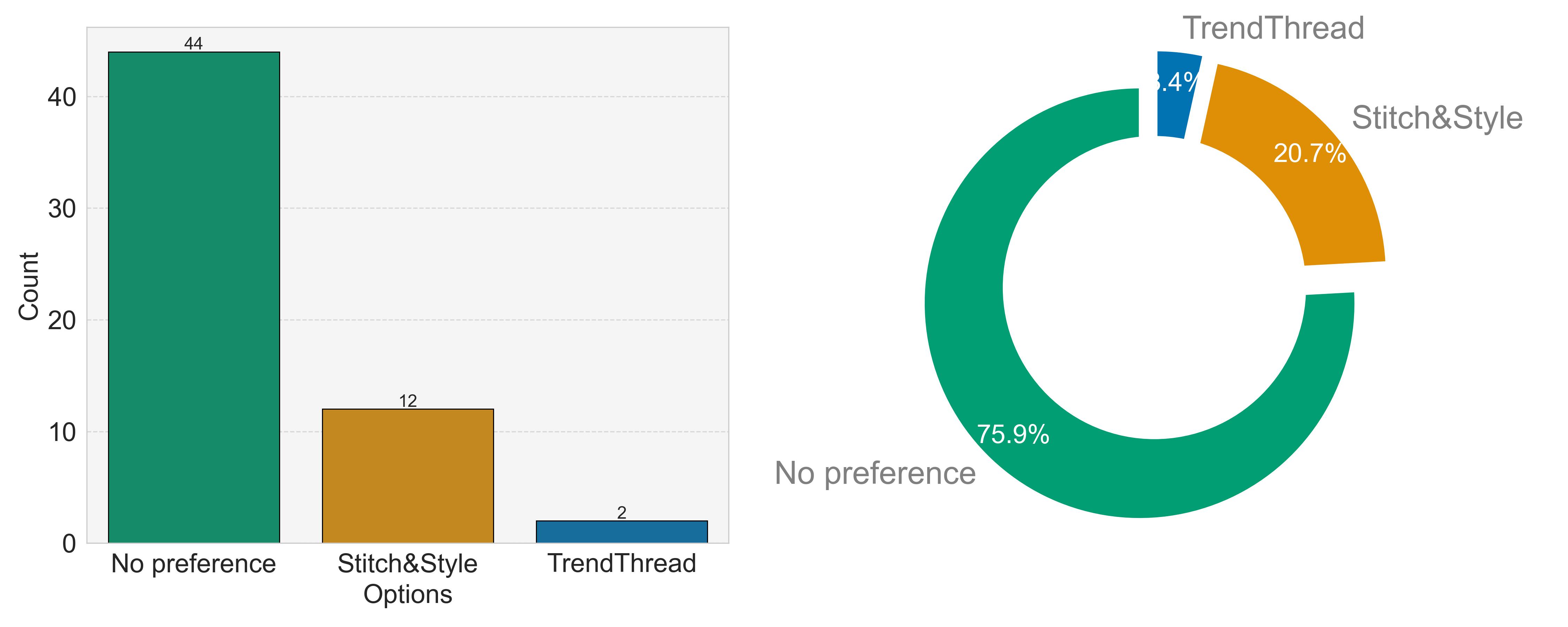}
\caption{\centering Impact of $30\%$ discount on initially indifferent consumers (i.e., ``No preference'' from Fig~\ref{fig:Q1})} 
\label{fig:Q1(2)}
\end{figure}

The second scenario tests whether perceptions of pre-purchase unfairness can also come from customer reviews associated with a store. Consider two stores offering an identical hoodie at the same price and with equivalent ratings. The first store, \textit{TrendThread}, displays one negative comment addressing discount unfairness among five positive comments\footnote{An auxiliary analysis was conducted to discern whether the placement of this negative comment has any consequential impact. Refer to Appendix \ref{subapp:question} for details.}. Conversely, the second store, \textit{Stitch\&Style}, presents three negative comments to discount unfairness. The focus is to analyze how the volume of comments related to unfair treatments influences consumer shopping behavior.

\begin{figure}[!tb]
\center
\includegraphics[width=0.8\textwidth]{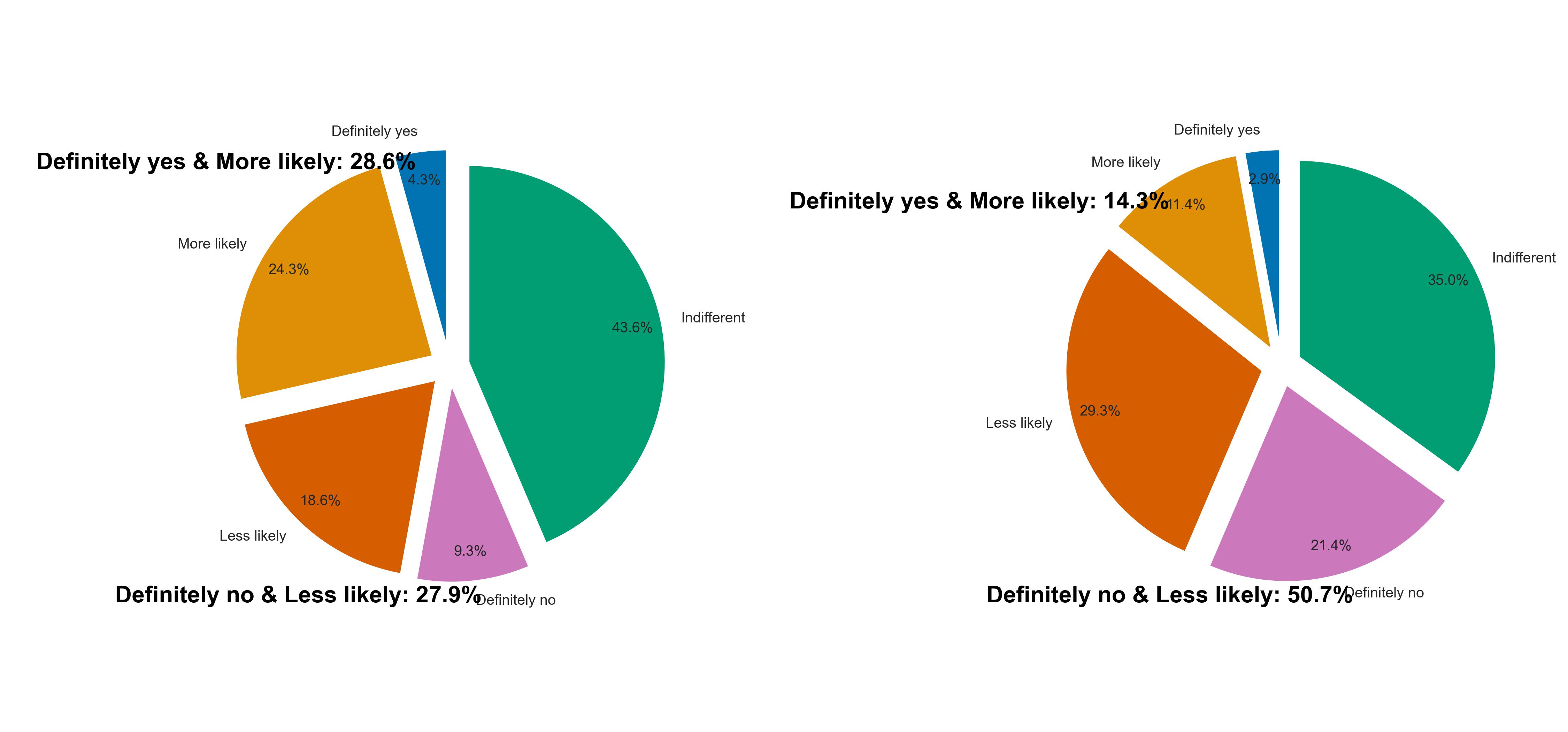}
\caption{\centering \textbf{Left, Right} Distribution of preferences for store \textit{TrendThread} and \textit{Stitch\&Style} respectively.} 
\label{fig:Q3}
\end{figure}

As depicted in Figure \ref{fig:Q3}, upon encountering few comments indicating unfair treatment, $28.6\%$ of respondents express a positive inclination towards making a purchase, whereas $27.9\%$ demonstrate a negative inclination. However, in the presence of a 
\emph{trend} of unfairness in the comments, the proportion of positive respondents diminishes to $14.3\%$, while those with a negative predisposition drastically surge to over $50\%$. 

The findings from both scenarios underscore the observation that a pronounced perception of pre-purchase unfairness can significantly influence consumer behavior, leading many consumers to avoid shopping at the store they see as unfair. This verifies our hypothesis {\bf H1}.

\subsection{Buying Behavior and Observed Post-purchased Unfairness}

Next, we test hypothesis {\bf H2}. We propose a scenario where a customer, having made a purchase at the full price, subsequently learns of another individual who got a discount for the same item shortly thereafter, potentially triggering feelings of unfairness. To assess the repercussions of such a scenario on future consumer behavior, we initially investigate whether respondents perceive unfairness in this situation. Subsequently, we probe their likelihood of repeating purchases at the same store and their propensity to recommend the store to others.

\begin{figure}[!tb]
\center
\includegraphics[width=\textwidth]{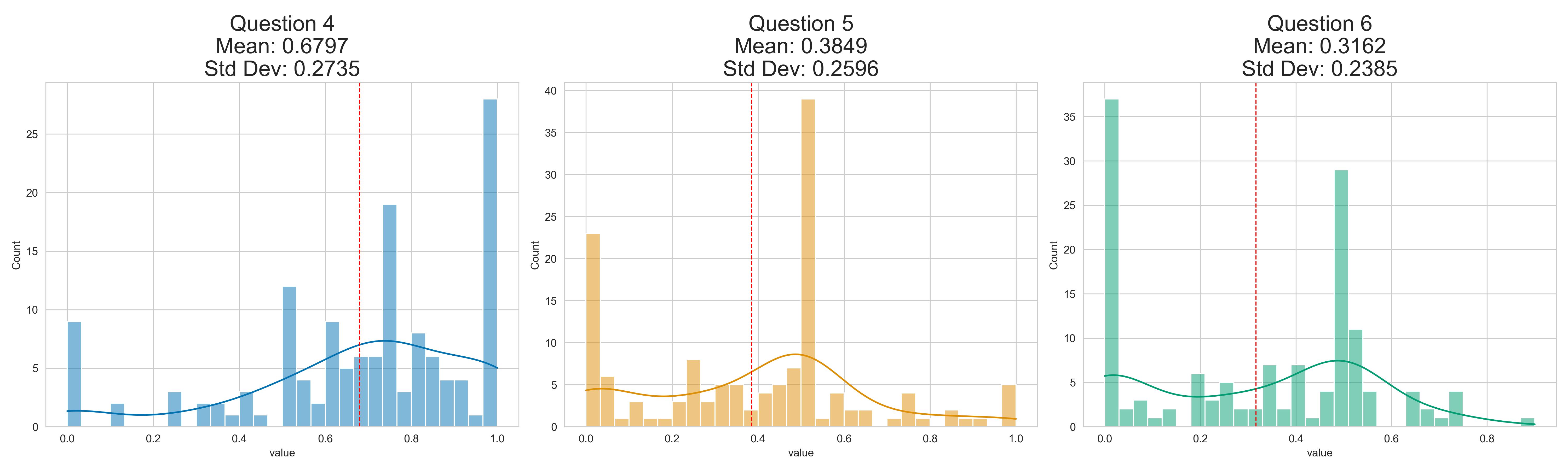}
\caption{\textbf{Left: } Perceived unfairness level (0 = Not at all, 1 = Extremely); \textbf{Middle: } Customer repurchase likelihood (0 = Definitely no, 1 = Definitely yes);  \textbf{Right: } Customer recommendation likelihood (0 = Definitely no, 1 = Definitely yes)} 
\label{fig:Q4}
\end{figure}

In Figure \ref{fig:Q4}, the left histogram depicts the perceived unfairness distribution, with $0$ indicating no perceived unfairness and $1$ denoting extreme unfairness. Note that in order to obtain this information, we asked respondents to choose how unfairly treated they feel, by moving a needle. We ran the experiment with different initial placements of the needle, to abstract away from biases. For more details, we refer the reader to Appendix~\ref{append:survey}.

A significant number of respondents perceive high levels of unfairness, evidenced by an average response of approximately $0.68$. This indicates that post-purchase differential treatment significantly influences perceptions of unfairness. The middle histogram illustrates respondents' willingness to repurchase from the store, where $0$ represents a definite refusal and $1$ a definite intention to repurchase. The mean response, around $0.38$, suggests a reduced likelihood of repeat patronage due to observed disparities in post-purchase treatment. The right histogram, portrays the willingness to recommend the store with $0$ indicating an unwillingness to recommend and $1$ the opposite; it reveals a predominant reluctance to recommend among participants. This trend towards non-recommendation may adversely affect the store's potential revenue. These findings underline the importance of maintaining post-purchase individual fairness as a key factor influencing customer retention and advocacy, with direct implications for revenue streams.

\subsection{Continuous Time Metrics of Individual Fairness}

In this part, our objective is to validate hypothesis {\bf H3}. This investigation also serves to corroborate the mathematical definition of individual fairness metrics presented in Definition \ref{def:ifmetrics}. To assess whether individuals are more sensitive to recent disparities in treatment as opposed to those in the more distant past, we propose the following scenario: Imagine encountering two identical posts, each stating that a user got a $10\%$ discount from \textit{TrendThread}. The only distinction between these posts is their post time: one was created 5 minutes ago, while the other dates back to a month ago.

\begin{figure}[!tb]
\center
\includegraphics[width=0.5\textwidth]{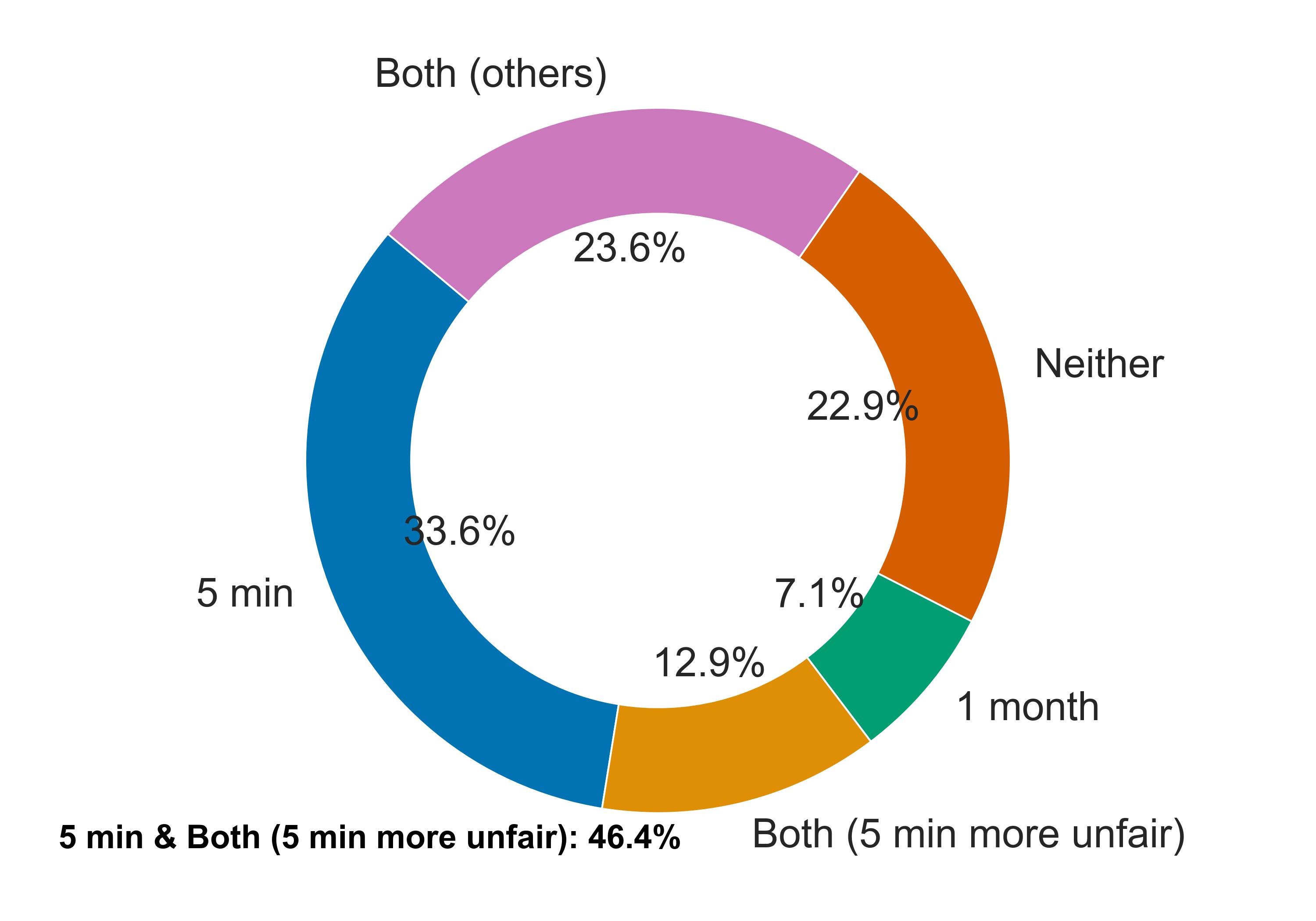}
\caption{Distribution of opinions regarding which post makes the respondents feel more unfairly treated} 
\label{fig:Q7}
\end{figure}

As illustrated in Figure \ref{fig:Q7}, the predominant initial choice among respondents is the \textit{5-minute} post, accounting for $33.6\%$ of the selections. Incorporating those respondents who initially chose \textit{Both} and subsequently indicated that the \textit{5-minute} post evoked stronger feelings of unfairness, we find that $46.4\%$ of participants perceive recent different treatments as more unfair than ones set further back in time. With this $46.4\%$ surpassing the proportions of other available choices, our continuous metrics hypothesis is empirically validated.

\section{Extension: Price-based Revenue Management}
\label{sec:extension}


Price-based RM represents another significant class of RM problems~\cite{gallego1994optimal, gallego1997multiproduct, talluri2004revenue}. In this paradigm, a seller markets one or multiple products over a finite horizon with a given initial resource inventory. Unlike the quantity-based RM, which determines whether to accept or reject an incoming customer, the decision in price-based RM involves setting a price for each arriving customer. Typically, we assume the existence of a bijective mapping between the price and the purchasing probability of each type of customer, whereby a higher price correlates with a lower purchasing probability. The objective is to maximize the total revenue over the given finite time horizon. We first define individual fairness for price-based RM, drawing upon a similar intuition to Definition \ref{def:ifmetrics}, designed for quantity-based RM.

\begin{definition} \label{def:ifmetricsprice}
Assume the algorithm $\mathcal{A}$ assigns a price $p_u^{(i)}$ to the $u^{\text{th}}$ customer of type $i$. An Algorithm $\mathcal{A}$ is $(\alpha,\delta)$-fair if for all $i \in [n]$, $u,v \in [n_i]$, where $n$ is the number of customer type and $n_i$ is the number of type $i$ customers who will arrive, with probability at least $1-\delta$, we have
\[
\mathbb{P}(p_u^{(i)} < p_v^{(i)} ) \leq d(u,v).
\]
\end{definition}

Traditional methods for solving price-based revenue management, for instance, those proposed by \citet{gallego1994optimal, gallego1997multiproduct}, use the \AlgDLP algorithm structure, while \citet{jasin2012re,wang2022constant} employs the \AlgRDLP algorithm structure. These papers offer static or dynamic pricing strategies aimed at minimizing regret under a stochastic arrival setting. It is evident that neither static nor dynamic pricing strategies satisfy the individual fairness metrics as introduced in Definition \ref{def:ifmetricsprice}. In the case of static pricing, as the capacity depletes, we assign a price of $+\infty$ to arriving customers. For dynamic programming, its inherent structure is entirely not individually fair. 

First, we employ the concept of a decreasing grace period, adjusting it to suit the price-based context. Specifically, a decreasing grace period to type $i$ customer within the interval $[t_1,t_2]$ means that for every $u^{\text{th}}$ type $i$ customer arriving within the interval $[t_1,t_2]$, we assign them a price $p_{u-1}^{(i)}$ with a probability of $1-\beta$, and a price of $+\infty$ with a probability of $\beta$. This is analogous to the decreasing grace period for quantity-based RM. 

Subsequently, we use the modification of the static pricing algorithm as an example to illustrate how to apply a decreasing grace period to guarantee individual fairness. A static pricing algorithm assigns a price $p^{(i)}$ to all type $i$ customers as long as sufficient resources exist, and assigns a price of $+\infty$ to all customers once resources are exhausted. The following theorem summarizes the results and the proof can be found in Appendix \ref{append:Sec8}:

\begin{theorem} \label{thm:pricebased}
    For the price-based revenue management problem where the price $p^{(i)}$ is given to type $i$ customers, by assigning a decreasing grace period $[t_1(i), t_2(i)]$, where $t_1 (i) = \inf\{t: \min_{j \in [L]}m_j(t) \leq \bar a n \gamma\}$ and $t_2(i) = T$ to each type $i$ customer,  where $\gamma=\log_{1-\beta}\delta$, the $(\alpha,\delta)$-fair metrics is satisfied. Moreover, compared to the original static pricing algorithm, the revenue loss is bounded by $\frac{\bar a}{\underline a} n\gamma \bar{p}$, where $\bar{p}=\sup_{i \in [n]} p^{(i)}$.
\end{theorem}

{
\section{Discussion} \label{sec:discussion}

In this paper, we addressed fairness from the customer's side and we provided a way for standard RM algorithms to satisfy our notion with almost no revenue loss. Defining the ``correct'' notion of fairness is a tall task, one which cannot be definitively solved by a single paper. That said, in an effort to move the discussion on fairness definitions in RM forward, we provide some more details around the properties of \emph{our} fairness definition ---and how we came to it--- in this section.

\subsection{Fairness Within and Across Customer Types} \label{subsec:sametype}

In Section \ref{sec:classicalnrm}, we established that in classical revenue management problems, customers are categorized based on the type and quantity of resources they request, rather than demographic characteristics or other personal attributes. In this setting, customers within the same type are homogeneous—they request the same resource(s) in the same amount and generate the same revenue—and our definition of individual fairness naturally applies within each type.

In real-world settings, however, it is common for different customer subgroups to request the same resource but generate different revenues. For example, airlines often sell the same economy-class seat at varying prices depending on the customer's profile, and hotels may charge corporate clients lower rates than tourists for identical rooms. To model such cases, we extend the classical setting by allowing each type \( i \in [n] \) to consist of \( K \) subgroups, where each subgroup \( k \in [K] \) requests the same resource in identical quantity but yields a different reward \( r_{ik} \).

To maintain a unified framework, we show that this extended model can be equivalently transformed back into the classical network revenue management form. Specifically, let \( \lambda_{ik} \) denote the arrival rate of subgroup \( k \) within type \( i \), satisfying \( \sum_{k \in [K]} \lambda_{ik} = \lambda_i \). We define the weighted reward for type \( i \) as

\[
\hat{r_i} = \frac{\sum_{k \in [K]} \lambda_{ik} r_{ik}}{\lambda_i}.
\]

By aggregating all subgroups into a single type using this weighted reward, we can treat them uniformly for the purposes of revenue optimization and fairness. As a result, the fairness mechanisms and algorithms developed in the main body of the paper can be directly applied by replacing the rewards of all subgroups within type $i$ with the aggregated weighted reward $\hat{r_i}$.

\subsection{Time-Dependent Fairness: Ensuring Stronger Guarantees for Close Arrivals} \label{subsec:time}

Through our survey in Section~\ref{sec:experiment}, we examined how humans perceive fairness over time. Our findings indicate that people view disparate treatment occurring in close temporal proximity as significantly more unfair than similar disparities occurring in the distant past. This is intuitive: for instance, a customer is unlikely to care whether another customer receives preferential treatment 100 or 101 days later, whereas they would perceive a fairness concern if the disparity occurred today versus tomorrow.  

To reflect this intuition, our fairness definition is \emph{time-dependent}. Specifically, fairness should be locally strong—ensuring a low probability of observing disparate treatment for customers who arrive close in time—while allowing for weaker guarantees as the time gap increases. The three properties of $d(u,v)$ from Definition~\ref{def:ifmetrics} make sure that we satisfy these desiderata. Property 1, which lets $d(u,v)$ be an increasing function with respect to $|u-v|$, ensures local fairness by imposing stricter constraints on customers who arrive closer together than those farther apart. Property 2 sets $d(u,v) = 1$ when $|u - v| = 1$, providing a consistent benchmark for evaluating fairness at the smallest time separation. Property 3 introduces a threshold $w_{\alpha}$ such that when $|u - v| = w_{\alpha}$, the fairness constraint becomes $\mathbb{P}\left[i(v) \succ_{\calA} i(u) \right] \leq \alpha \cdot d(u,v) = \alpha \cdot 1/\alpha=1$, effectively placing no restriction on the algorithm. This reflects the idea that customers may not expect or require fairness comparisons with others who arrive significantly earlier or later. In practice, this property is reasonable—due to factors such as limited visibility into past transactions, changes in context or inventory, or simply diminishing concern over time, customers often do not compare their treatment to that of others far removed in time.




\subsection{Incorporating 
Arrival Uncertainty: Fairness with High Probability} \label{subsec:probability}

Our fairness definition is probabilistic. This is because directly enforcing fairness constraints in revenue management poses challenges due to the uncertainty in the arrival process. In particular, we will demonstrate next how such uncertainty can lead to highly inefficient allocation when fairness is imposed rigidly (i.e., with probability 1) by Equation \eqref{eq:tempfair}. To illustrate this, we present a simple impossibility result: no algorithm that satisfies our fairness definition with probability 1 (i.e., $\delta = 0$) can achieve a sub-linear regret even if there is only one type of customer!

\xhdr{Impossibility Result. }  Consider the simplest revenue management (RM) setting: a single type of customer arrives sequentially with rate $\lambda$, and a single type of resource is available with a total capacity of \( M \) units. Each customer requests 1 unit of the resource and yields a revenue of 1 if accepted. Without fairness constraints, the optimal policy is trivial: accept every customer until the resource is fully depleted. There is no tradeoff—maximizing revenue simply means admitting as many customers as possible.

However, under the fairness constraint of Equation~\eqref{eq:tempfair}, the situation changes dramatically. \textbf{No algorithm} satisfying the fairness constraint can achieve a sub-linear regret (see Appendix~\ref{append:def}). To see this, consider the case where only one unit of resource remains at some time \( t \). If the algorithm accepts the arriving customer at time \( t \), the next customer must be rejected due to depleted capacity. This immediate rejection creates a sharp disparity between two neighboring customers, violating the fairness constraint. As a result, any feasible algorithm must avoid fully depleting resources and must behave conservatively even when capacity is available.

To see the impact of this constraint on regret, consider the first customer whose acceptance probability is strictly positive. To satisfy the fairness constraint in Equation~\eqref{eq:tempfair}, the next customer must be accepted with probability at least \( (1 - \alpha) \) times that of the previous customer. This cascading requirement implies that acceptance probabilities cannot be reduced quickly. In fact, it takes infinite time steps to decay the acceptance probability from nearly 1 to 0, which means that the resource depletion must be earlier than the acceptance probability becoming $0$. Hence, any algorithm satisfying the fairness constraint must set the acceptance probability of the first customer to be $0$, which implies that the total revenue is $0$.



This observation highlights a fundamental issue: enforcing fairness deterministically for every possible arrival sequence leads to highly inefficient allocation policies. In effect, the algorithm must behave as though the arrival process is adversarial—even when it is in fact stochastic—simply to avoid violating fairness in rare, extreme cases.

To illustrate this, consider again the basic revenue management setting with a single type of customer and one type of resource. Suppose the arrival rate is \( \lambda = \frac{M}{2T} \), meaning that the expected number of customers over the horizon is only half the resource capacity. In most realizations of this process, the total number of arrivals will be well below \( M \), and thus the algorithm could safely accept every customer while satisfying fairness. However, there is still a non-zero probability that more than \( M \) customers arrive. Under the original fairness definition—which requires the constraint to hold for every realization—the algorithm cannot afford to accept all early customers, as doing so might unfairly exclude future ones in this rare event. So it must preemptively reject some customers even when the resource limit is not reached, simply to hedge against low-probability events. 

To address this, we introduce an additional parameter \( \delta \), which allows the fairness constraint to be enforced with high probability, rather than absolute certainty. Specifically, we require that Equation~\eqref{eq:tempfair} holds with probability at least \( 1 - \delta \). This formulation preserves fairness in the overwhelming majority of cases, while avoiding pathological behavior in rare scenarios. The parameter \( \delta \) thus provides a principled way to balance fairness with efficiency in stochastic environments. 
}

\newpage

\bibliographystyle{ACM-Reference-Format}
\bibliography{reference}

\ECSwitch

\ECHead{Online Appendix}

\section{Survey Appendix} \label{append:survey}

\subsection{Demographics} \label{subapp:demo}
We conducted our survey through Prolific, ensuring gender balance—with an equal representation of male and female respondents—and fluency in English. Detailed demographic statistics are presented in Figure \ref{fig:demo}. The majority of our participants, over 50\%, belong to the age group of 18-35 years. Approximately 30\% are between 35 and 54 years old. Notably, more than 80\% of the respondents identify as white or Caucasian. Our respondent population is representative of the Prolific respondent base (see e.g., \cite{douglas2023data}).

The respondents of our survey come from diverse educational backgrounds. About 60\% have a bachelor's degree or have undertaken some form of graduate study. The rest have completed high school, earned an associate degree, or have some college education without obtaining a degree.

From an economic perspective, our sample spans various income brackets. Roughly 30\% report an income exceeding \$100K annually, while nearly 60\% fall within the \$25K-\$100K income range. Additionally, over 85\% of our participants identify as the primary shopper within their households. 

\begin{figure}[!tb]
\center
\includegraphics[width=\textwidth]{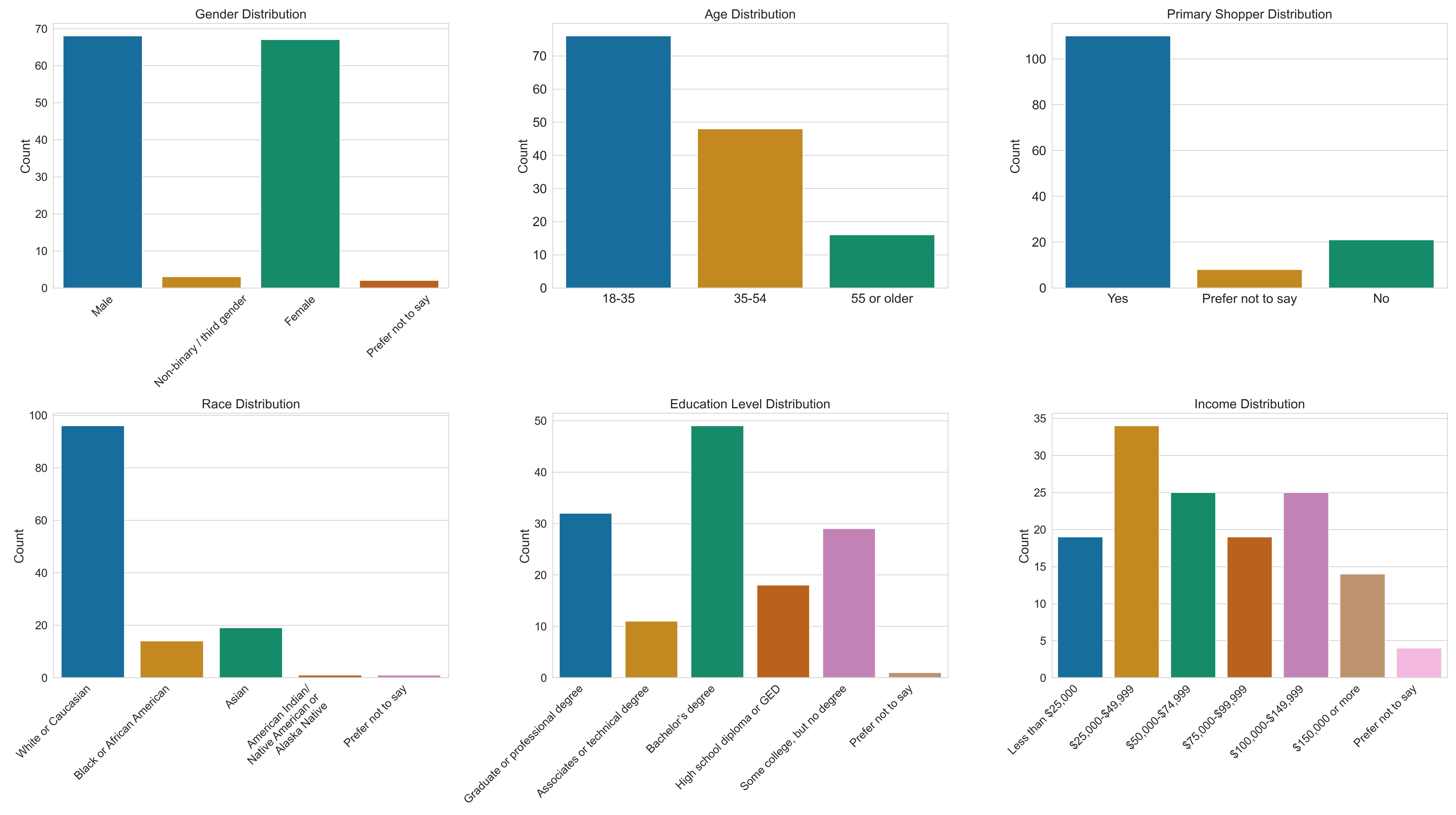}
\caption{Demographics of survey participants. The x axis corresponds to the categories for each
plot and the y axis reports the number of participants per category} 
\label{fig:demo}
\end{figure}

\subsection{Survey Questions} \label{subapp:question}

In this section, we display the questions and their respective distributions that were not introduced in Section \ref{sec:experiment}. \textbf{`Question 1'}, aimed at verifying the authenticity of participants' shopping behavior, seeks to determine if respondents would, all other factors being equal, opt for a cheaper option. The related distribution, depicted in Figure \ref{fig:Que1}, reveals that over 94\% of participants chose the less expensive item, thereby indicating a general propensity for rational economic behavior.

\begin{figure}[!tb]
\center
\includegraphics[width=0.4\textwidth]{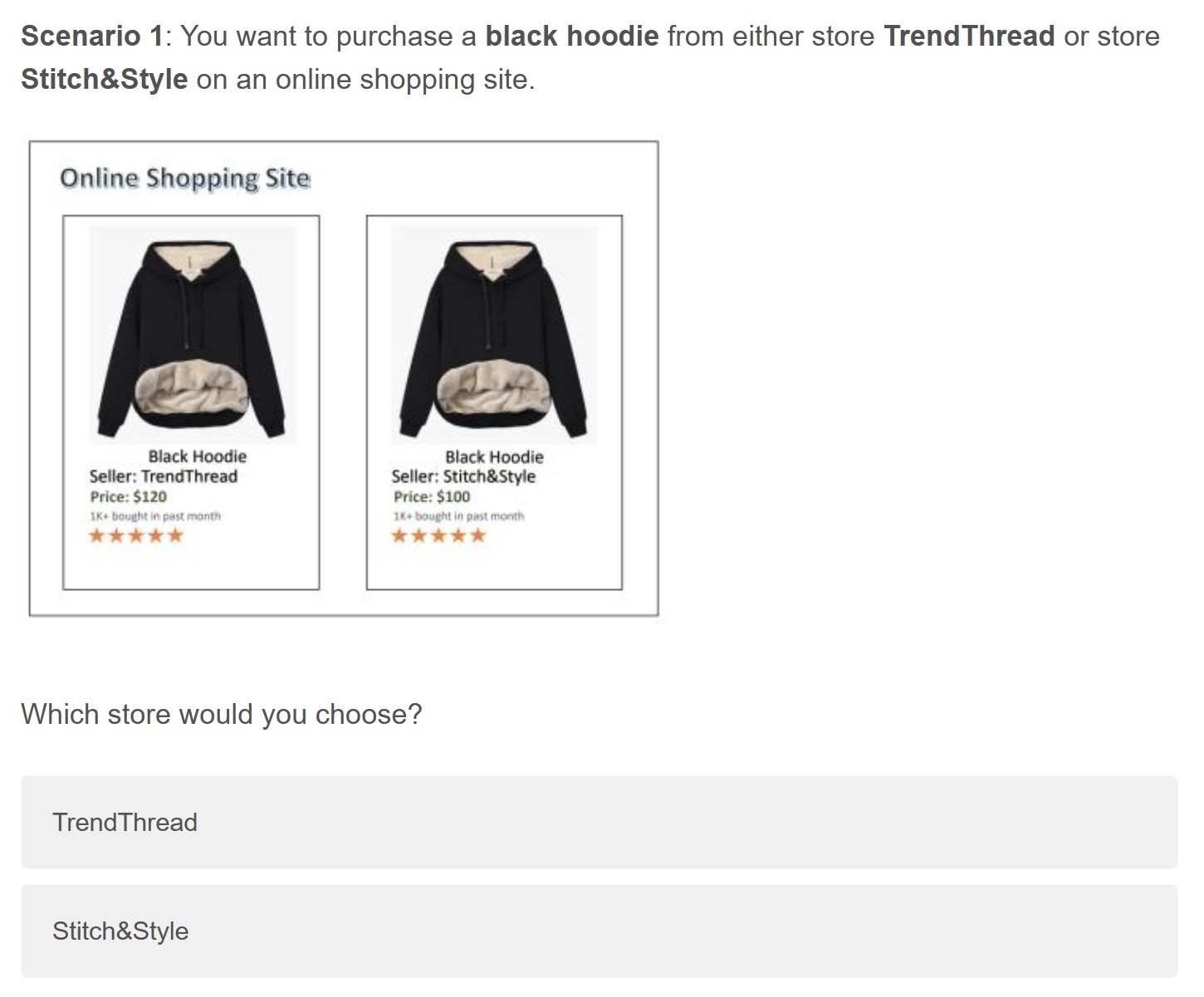}
\includegraphics[width=0.55\textwidth]{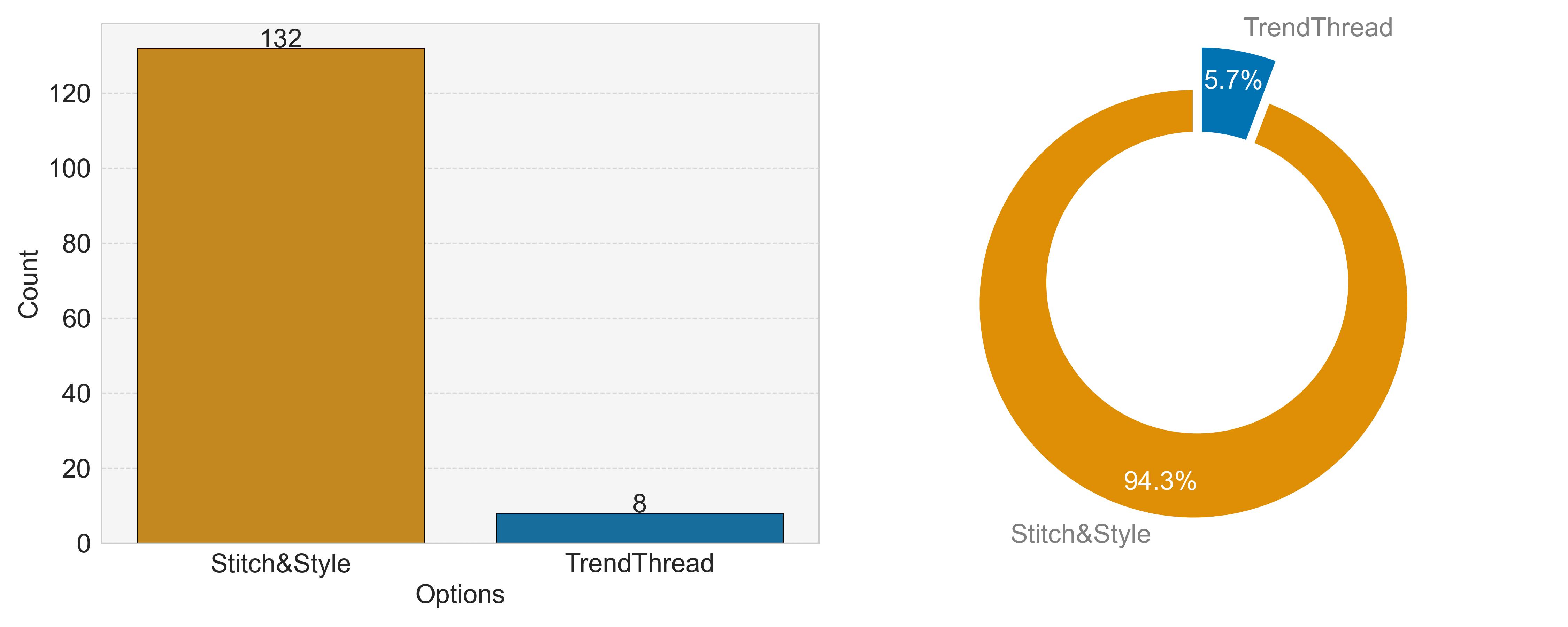}
\caption{Question 1 and its Answer Distribution} 
\label{fig:Que1}
\end{figure}

\textbf{`Question 2'} and \textbf{`Question 3'} are aiming to test hypothesis {\bf H1} mentioned in Section \ref{sec:experiment}. In these questions, we examine a situation focusing on an individual's engagement with social media content that depicts others' experiences with purchases. The detailed question can be found in Figure \ref{fig:Que23}, and the distribution can be found in Figures \ref{fig:Q1} and \ref{fig:Q1(2)}.

\begin{figure}[!tb]
\center
\includegraphics[width=\textwidth]{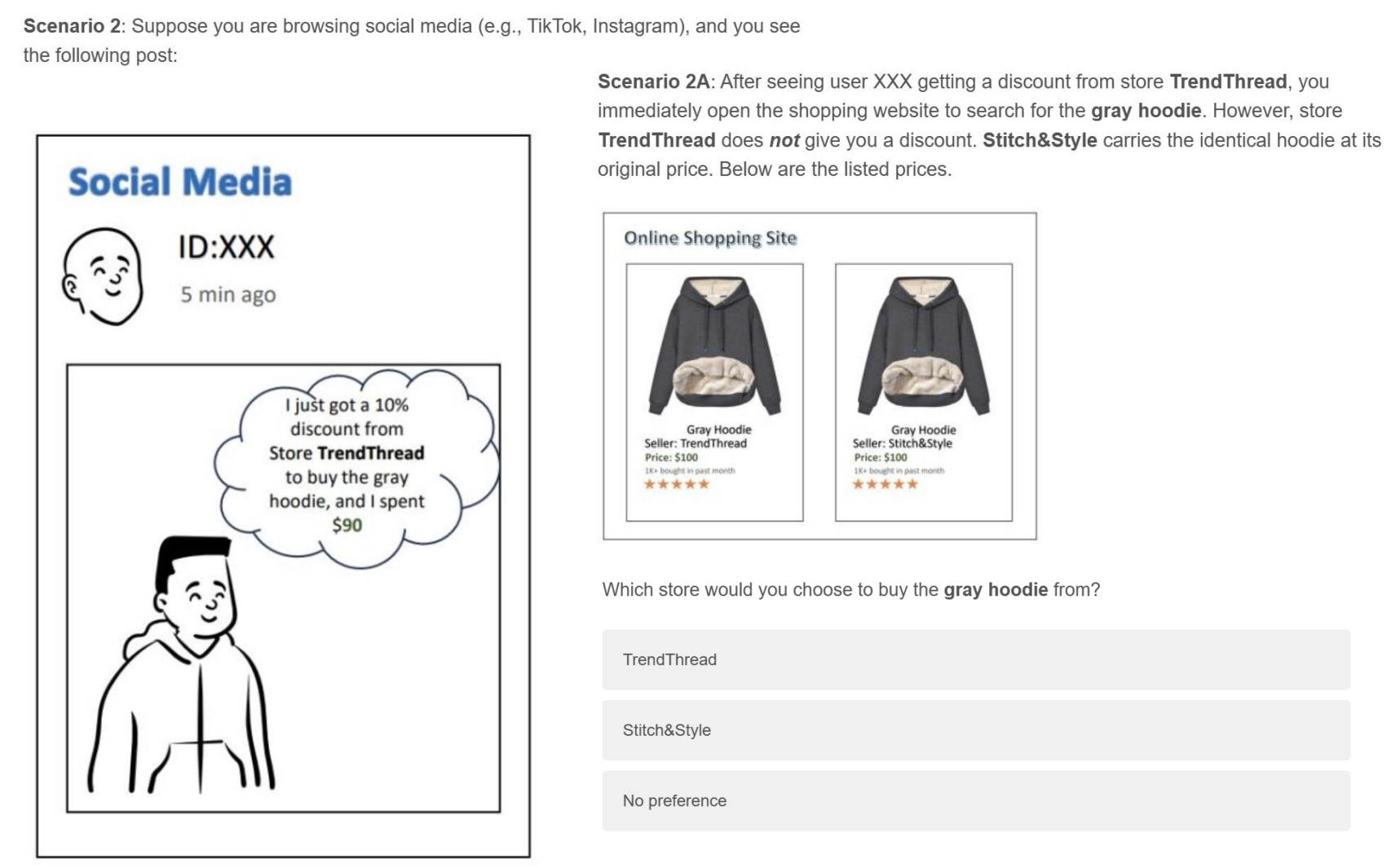}
\includegraphics[width=\textwidth]{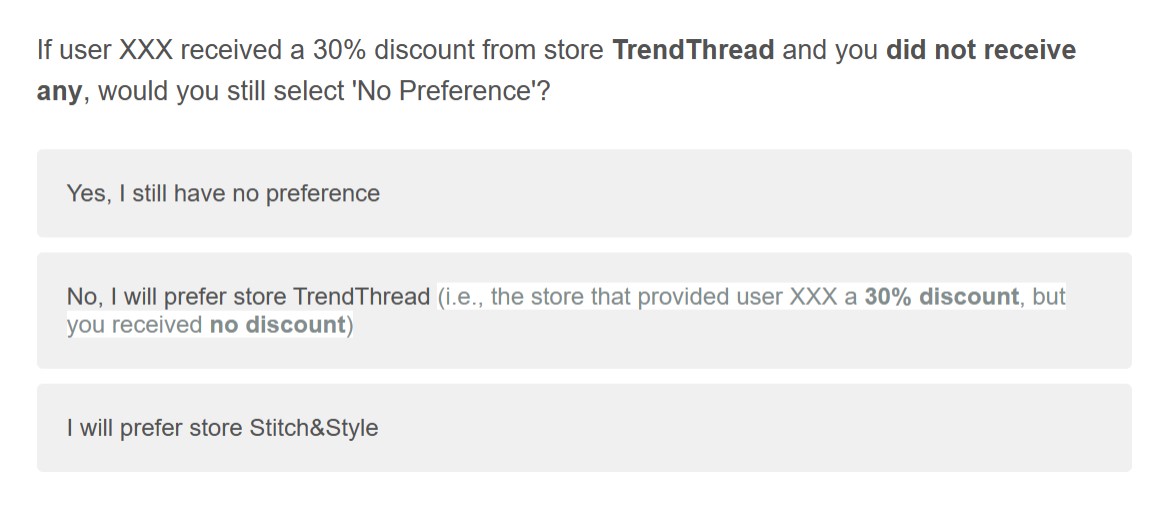}
\caption{Question 2 and Question 3} 
\label{fig:Que23}
\end{figure}

\textbf{`Question 4'} still tests hypothesis {\bf H1}. Contrasting the second and third questions, we highlight that pre-purchase unfairness can arise from customer reviews. We explore a scenario where two shops sell identical hoodies, each accompanied by varying review comments. The question can be found in Figure \ref{fig:Que4}. Additionally, we sought to investigate whether the positioning of such comments influences consumer reactions and potentially introduces bias. To this end, we randomized the placement of comments from User 5, ensuring they appeared with equal frequency at any point from the first to the fifth position. The aggregated distribution is presented in Section \ref{sec:extension}, Figure \ref{fig:Q3}. For a more detailed analysis on the order of comments, the distribution corresponding to each of the five cases is depicted in Figure \ref{fig:question3}.

\begin{figure}[!tb]
\center
\includegraphics[width=\textwidth]{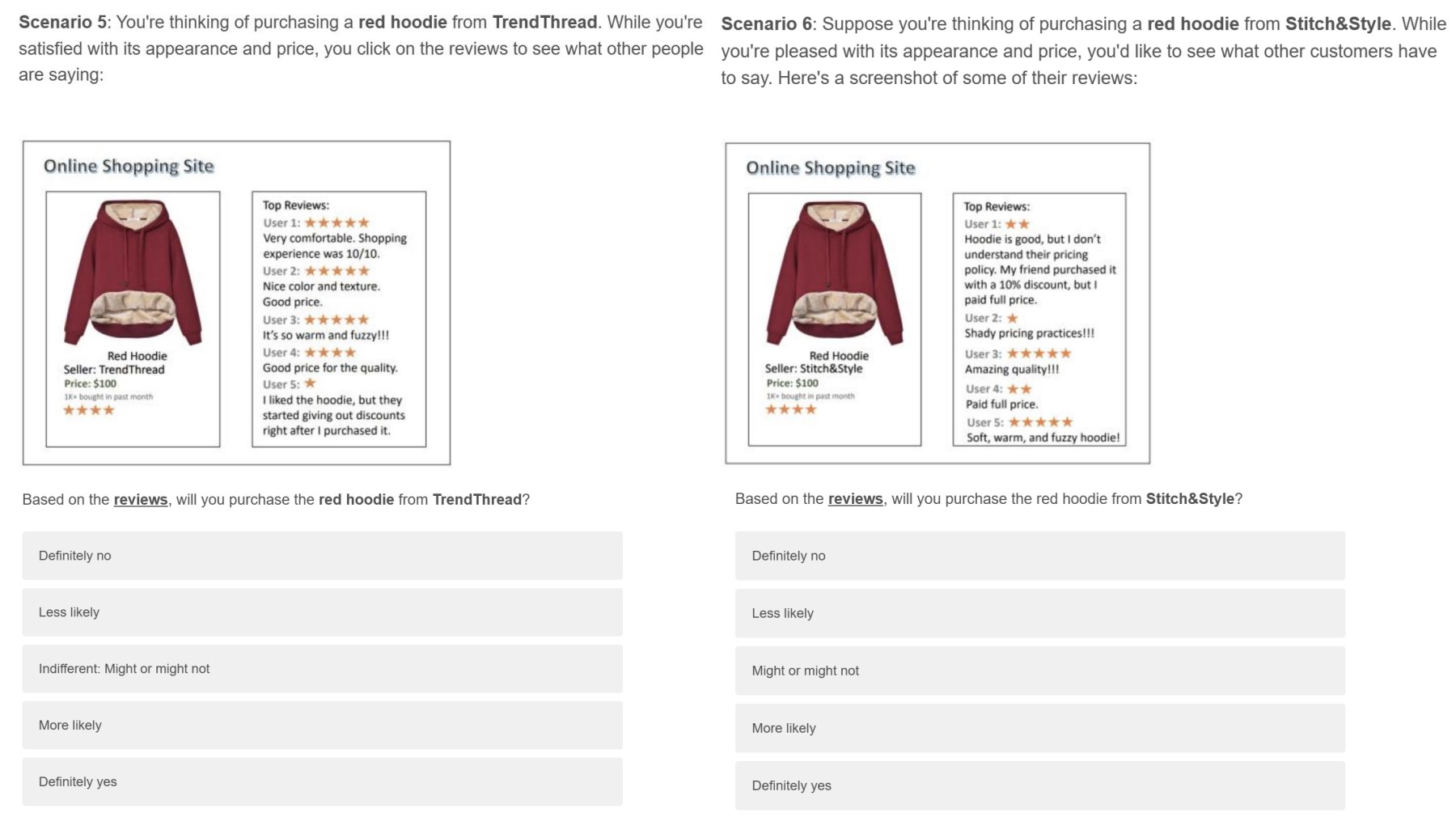}
\caption{Question 4} 
\label{fig:Que4}
\end{figure}

\begin{figure}[!tb]
\center
\includegraphics[width=\textwidth]{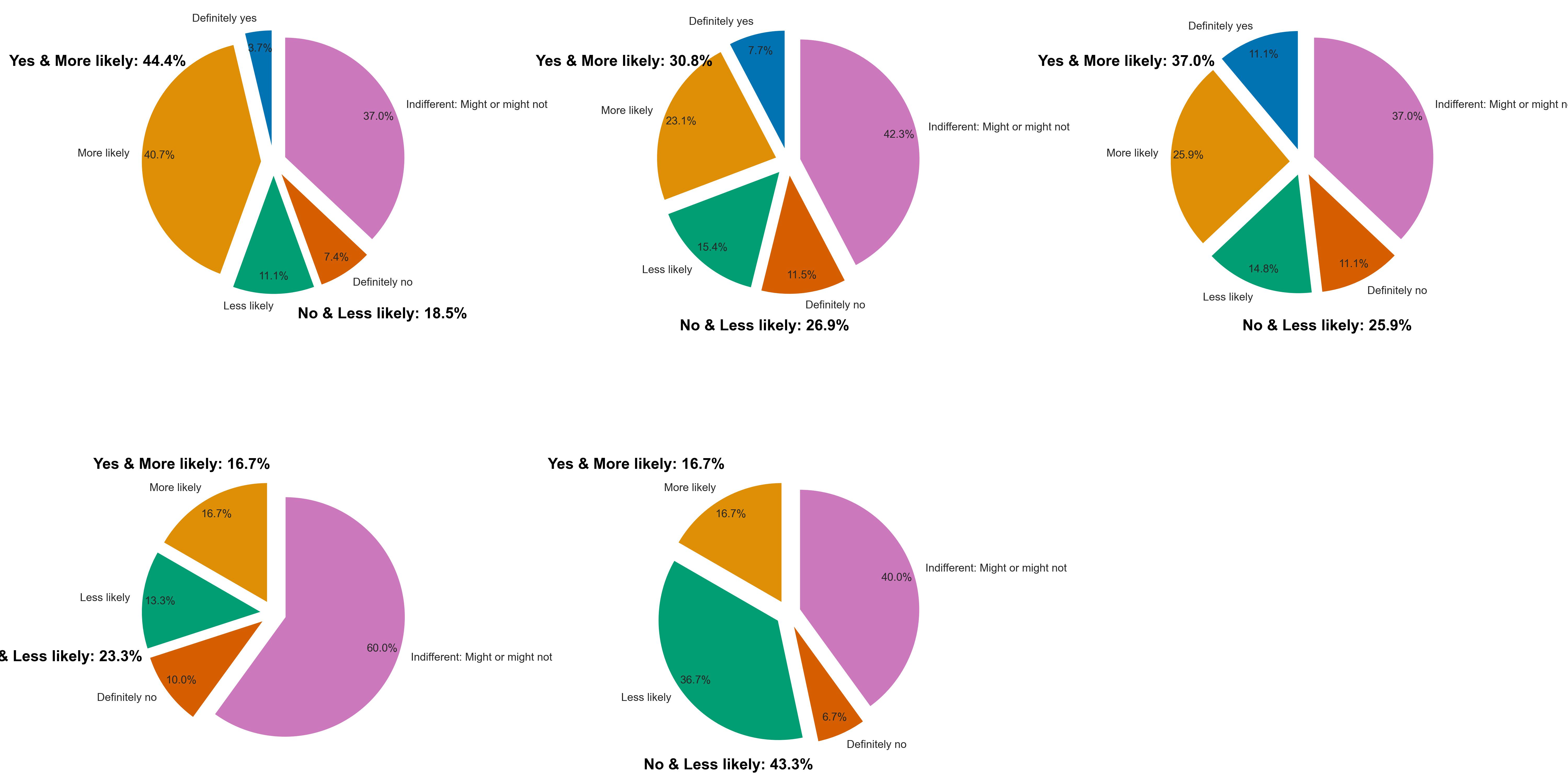}
\caption{Answer Distribution under Random Order of Comments} 
\label{fig:question3}
\end{figure}

In Figure \ref{fig:question3}, a significant bias is evident in the order of the comments. When comparing the same five comments, where the sole negative comment is placed last, the first diagram reveals a positive purchase incentive exceeding $44\%$ and a negative purchase incentive of approximately $18\%$. However, if the negative comment is positioned at the beginning, the positive responses drop to merely $16\%$, while the negative responses surge to over $43\%$. This intriguing observation may open the door to numerous compelling research questions.

\textbf{`Question 5'}, \textbf{`Question 6'}, and \textbf{`Question 7'} are used to test hypothesis {\bf H2} mentioned in Section \ref{sec:experiment}. We suggest a situation in which a customer, who has bought an item at its full price, later discovers that another person obtained the same product at a discounted rate, possibly leading to a sense of perceived unfairness. The detailed questions can be found in Figure \ref{fig:Que5}, and the distribution is in Figure \ref{fig:Q4}.

\begin{figure}[!tb]
\center
\includegraphics[width=\textwidth]{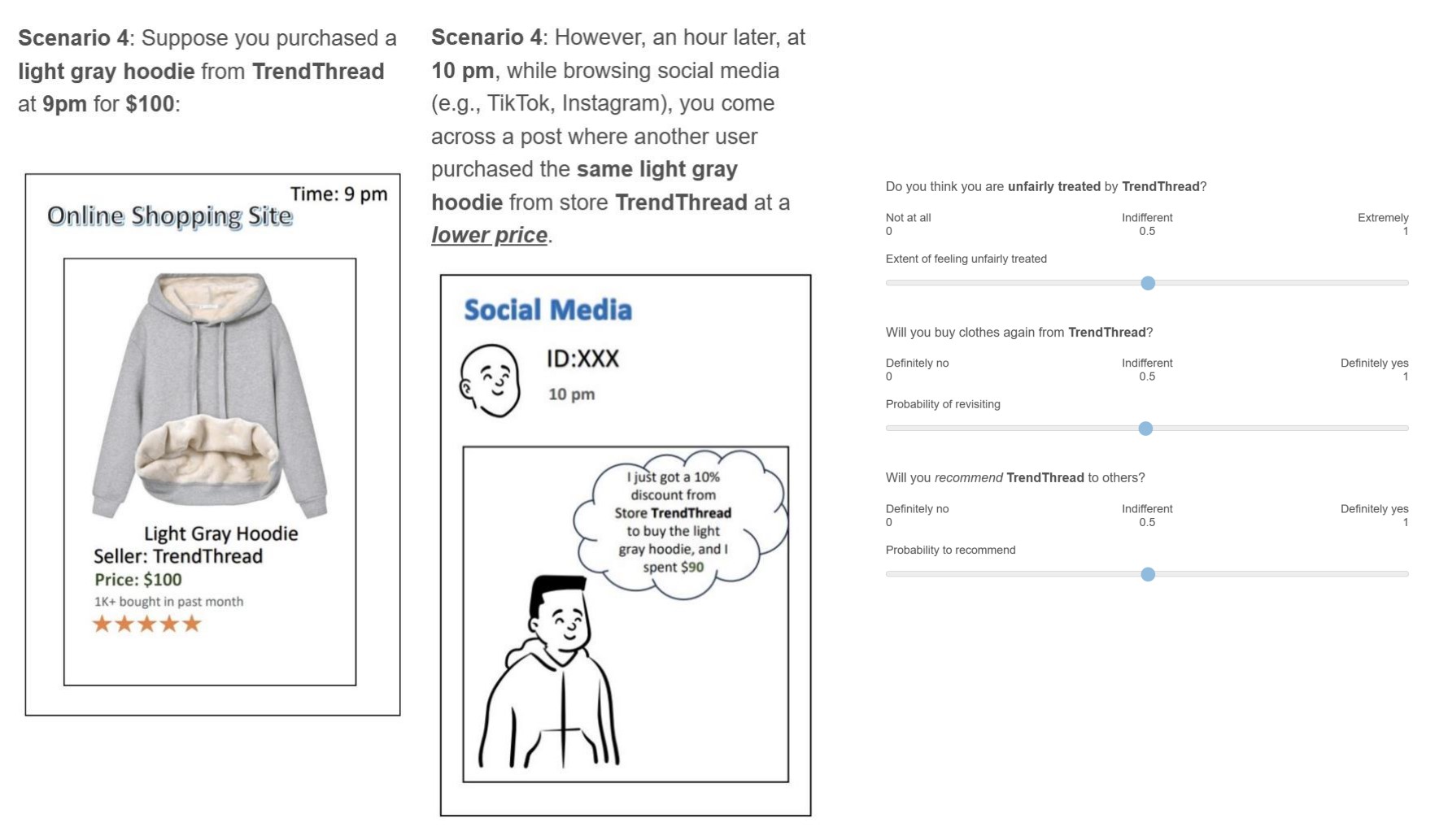}
\caption{Details of Questions 5-7} 
\label{fig:Que5}
\end{figure}

An aspect of our investigation focused on the potential bias introduced by the initial positioning of the slider in three questionnaire items. To assess this, the sliders were initialized in four distinct scenarios, each distributed equally among participants:

\begin{itemize}
\item All sliders beginning at the midpoint (0.5).
\item All sliders beginning at the maximum value (1).
\item All sliders beginning at the minimum value (0).
\item All sliders beginning at random positions.
\end{itemize}

The response distributions for each initialization scenario are illustrated in Figure \ref{fig:randomq4}, which suggests that there is no significant bias resulting from the initial location of the slider.

\begin{figure}[!tb]
\center
\includegraphics[width=\textwidth]{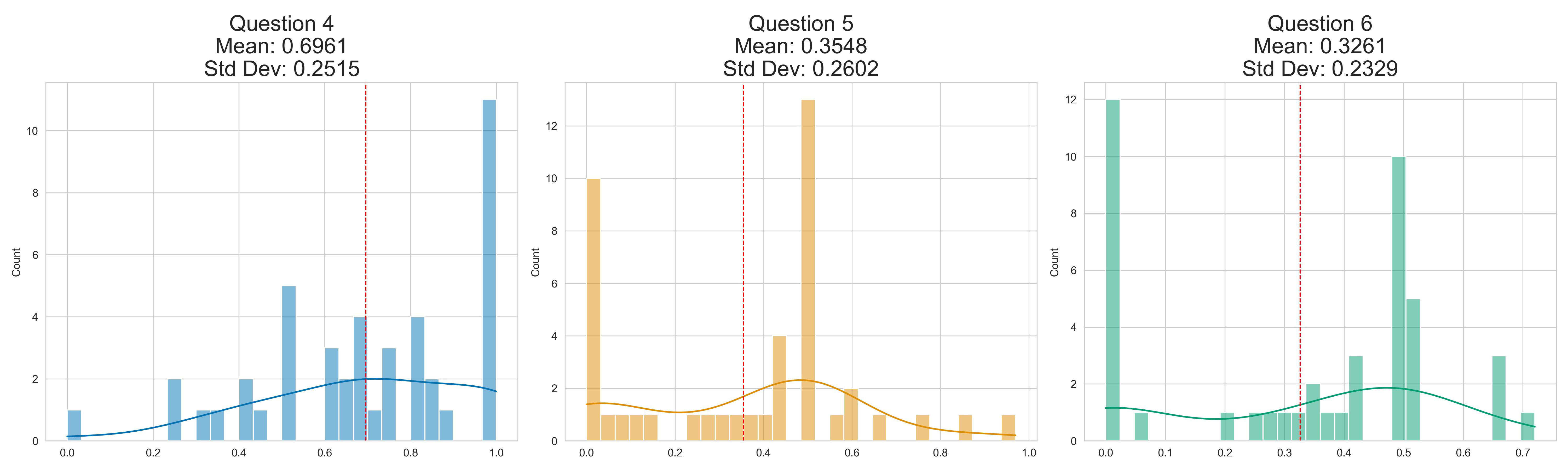}
\includegraphics[width=\textwidth]{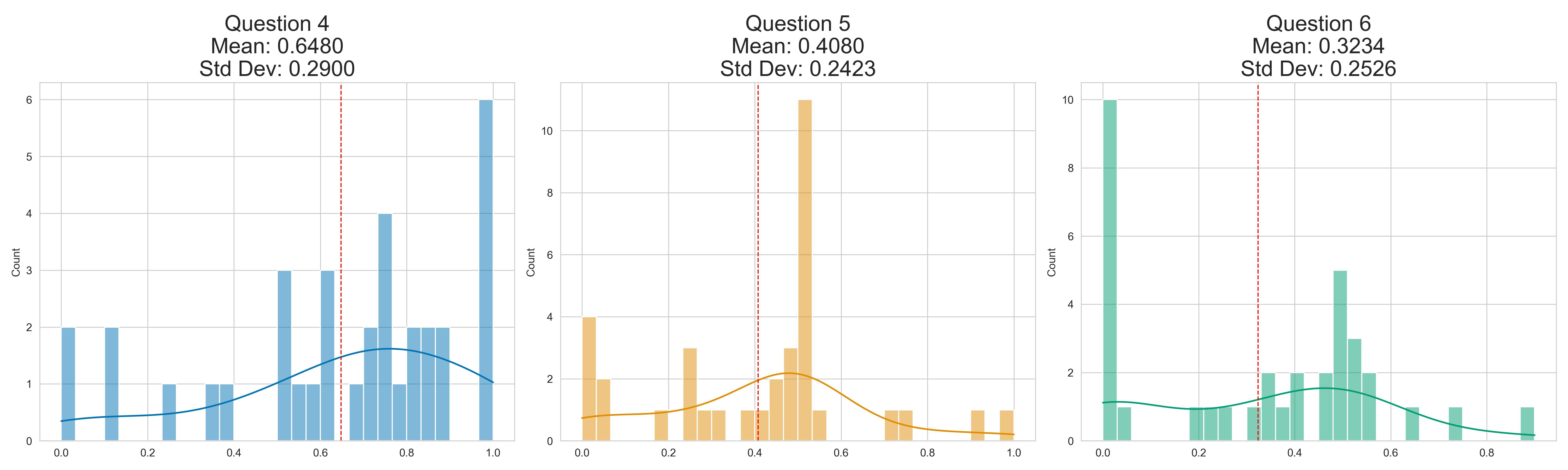}
\includegraphics[width=\textwidth]{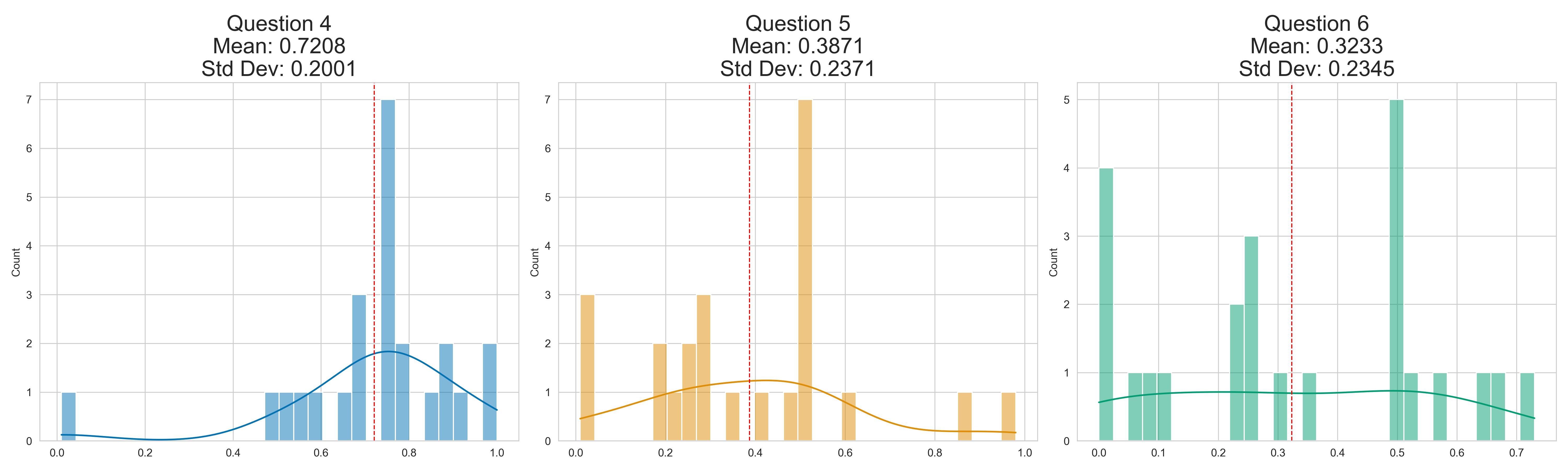}
\includegraphics[width=\textwidth]{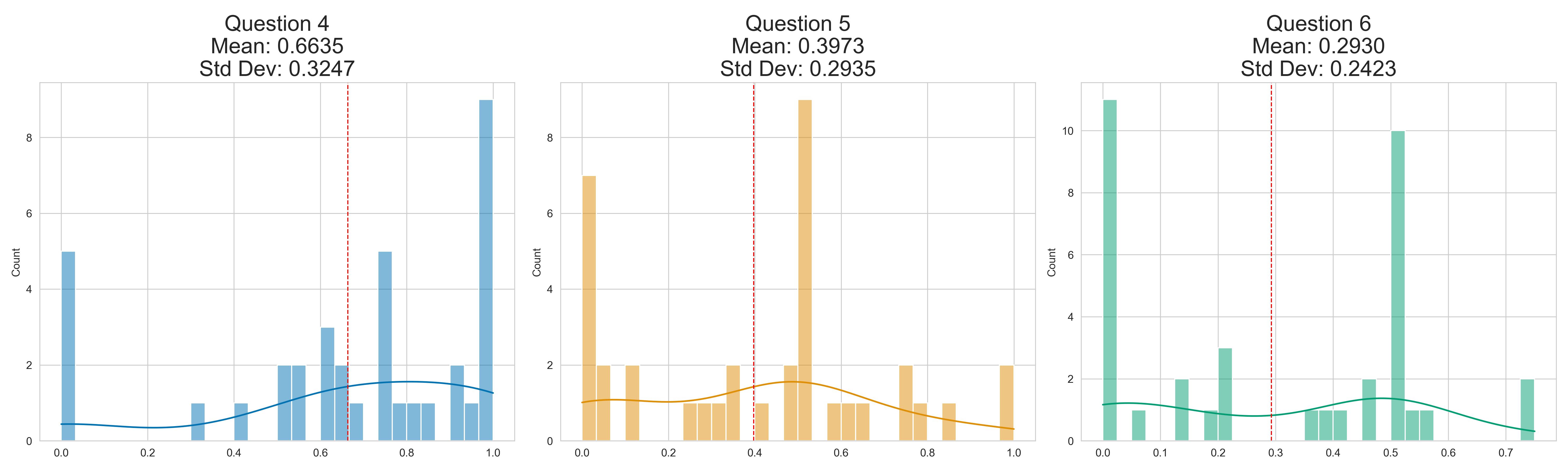}
\caption{Answer Distribution of Questions 5-7 with Different Starting Location of Sliders} 
\label{fig:randomq4}
\end{figure}

Finally, \textbf{`Question 8'} verifies hypothesis {\bf H3}, which is to evaluate if people are more responsive to recent differences in treatment compared to those that occurred further back in time. The detailed question can be found in Figure \ref{fig:Que6}, and the distribution has been mentioned in Figure \ref{fig:Q7}.

\begin{figure}[!tb]
\center
\includegraphics[width=\textwidth]{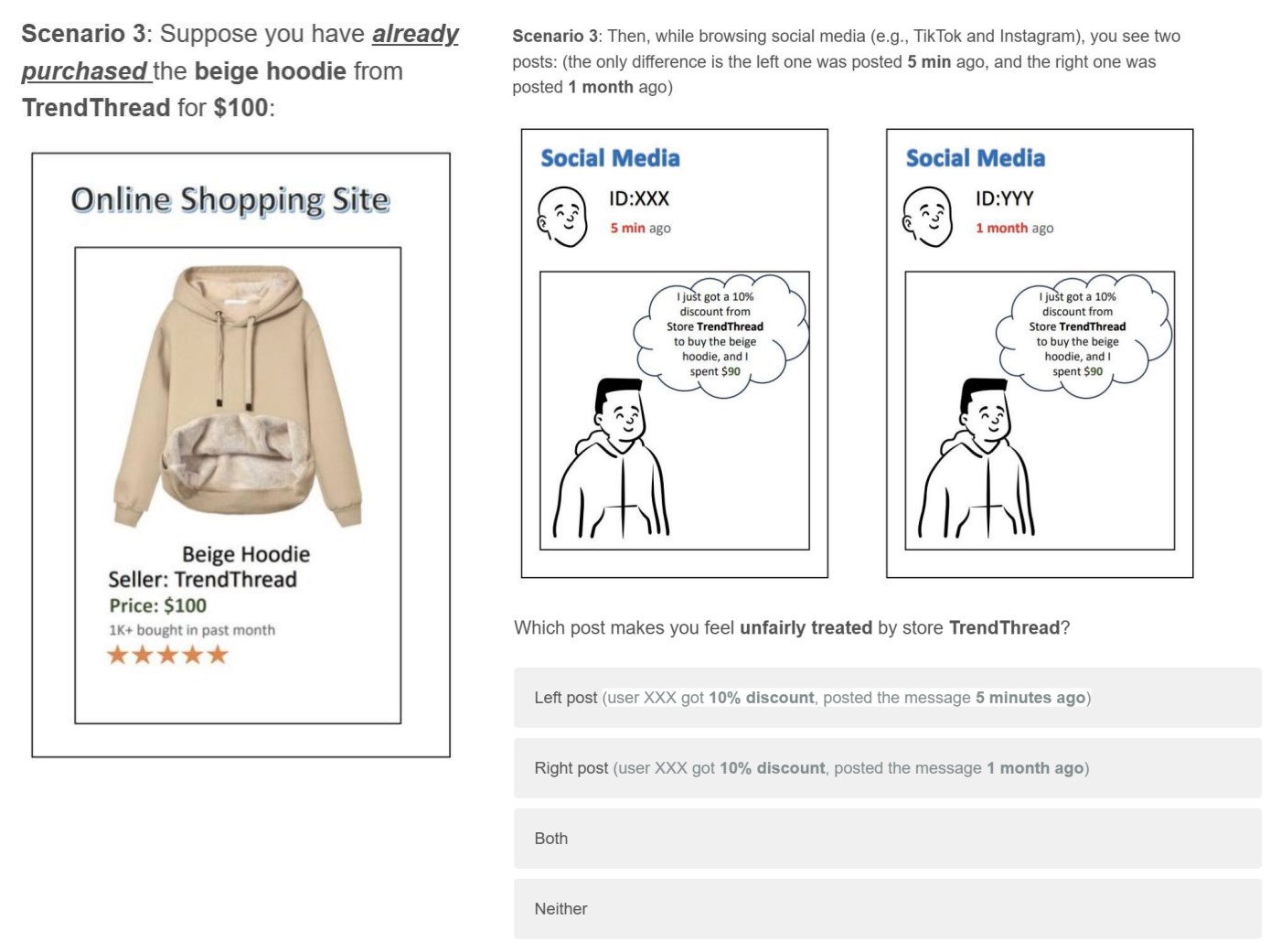}
\caption{Details of the Eighth Question} 
\label{fig:Que6}
\end{figure}

\subsection{$30\%$ Discount Survey} \label{subapp:30}

We conducted a comparative survey, wherein the $10\%$ discount is replaced by a $30\%$ discount. In this section, we will show the distribution of the answers of Question 1 - 8. 

Specifically, the distribution for \textbf{`Question 1'} can be found in Figure \ref{fig:30Qq1}. \textbf{`Question 2'} is depicted in Figure \ref{fig:30Q1}, while \textbf{`Question 3'} is represented in Figure \ref{fig:30Q1(2)}. The responses to \textbf{`Question 4'} are presented in Figure \ref{fig:30Q3}, and the distributions for \textbf{`Questions 5'} to \textbf{`Question 7'} are consolidated in Figure \ref{fig:30Q4}. Finally, the distribution for \textbf{`Question 8'} is available in Figure \ref{fig:30Q7}. Upon comparison with the data in Section \ref{sec:experiment}, it appears that the magnitude of the discount offered does not significantly influence the distribution of responses to these questions.

\begin{figure}[!tb]
\center
\includegraphics[width=0.9\textwidth]{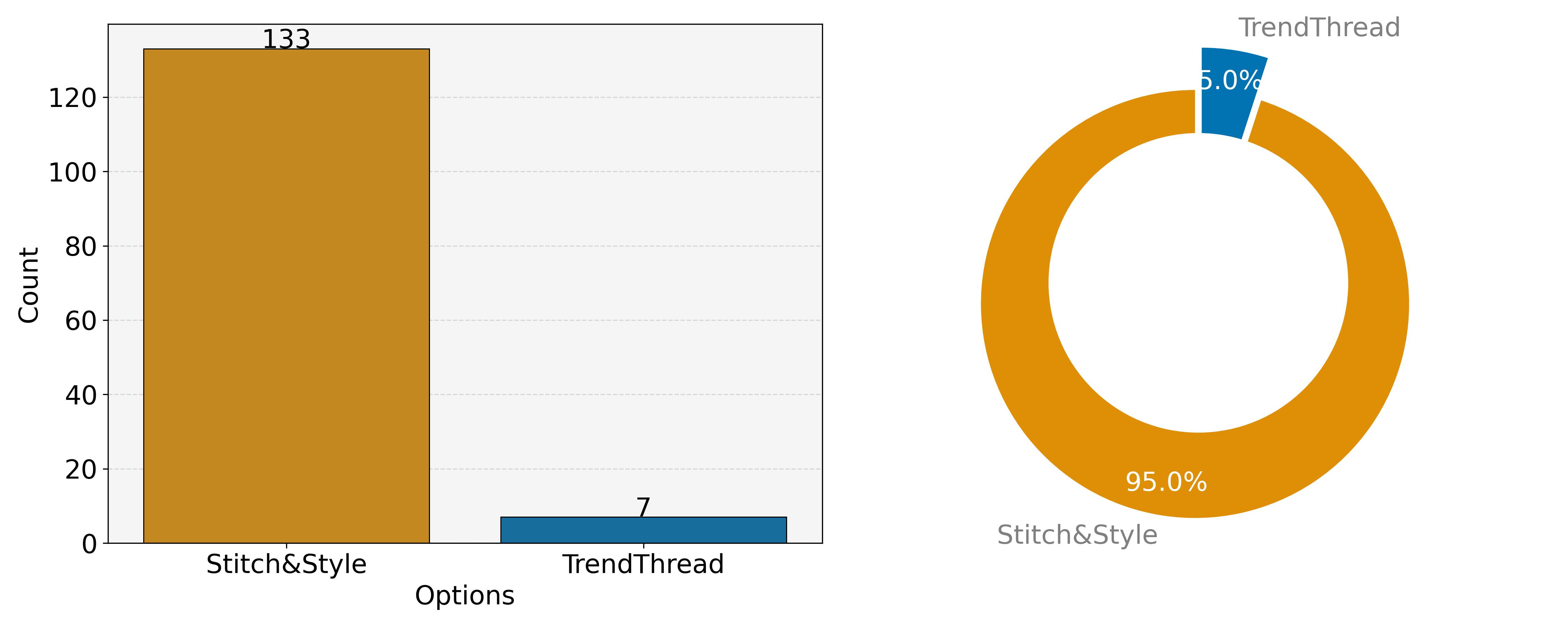}
\caption{Distribution of preferences among respondents of Question 1} 
\label{fig:30Qq1}
\end{figure}

\begin{figure}[!tb]
\center
\includegraphics[width=0.9\textwidth]{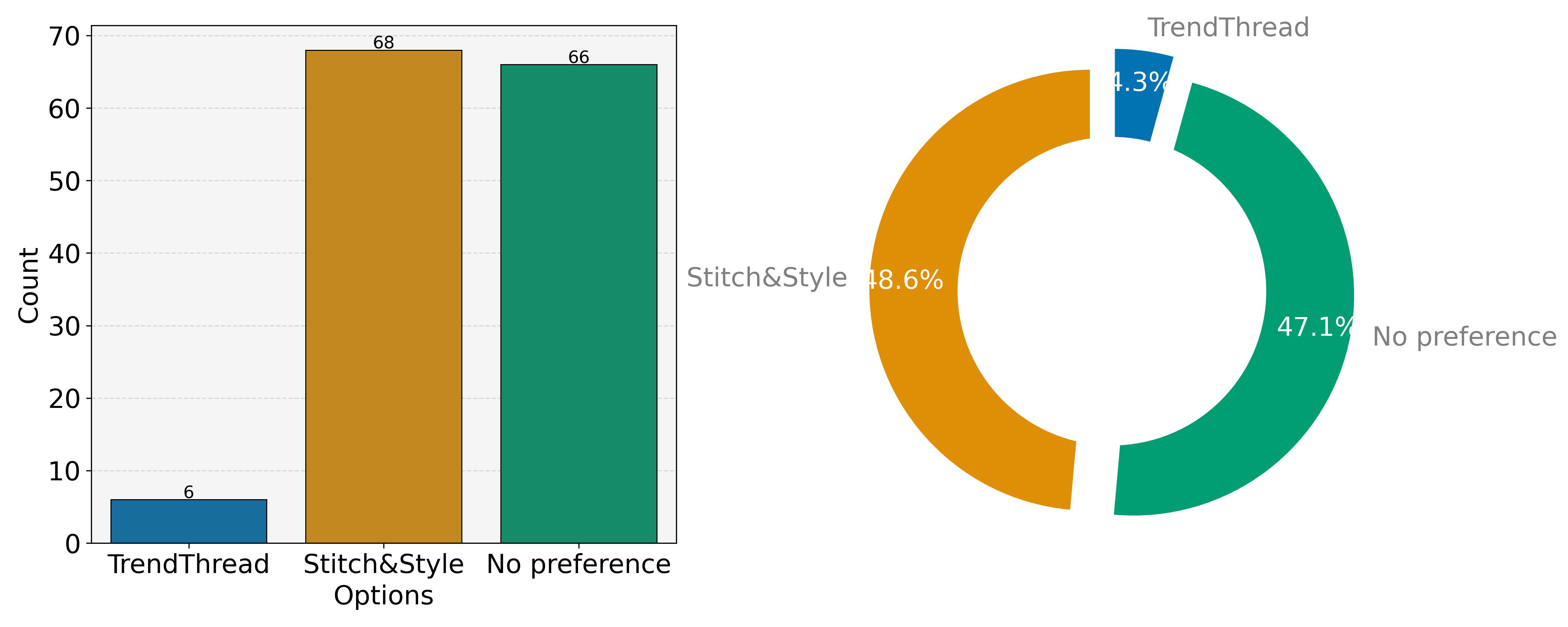}
\caption{Distribution of preferences among respondents of Question 2} 
\label{fig:30Q1}
\end{figure}

\begin{figure}[!tb]
\center
\includegraphics[width=0.9\textwidth]{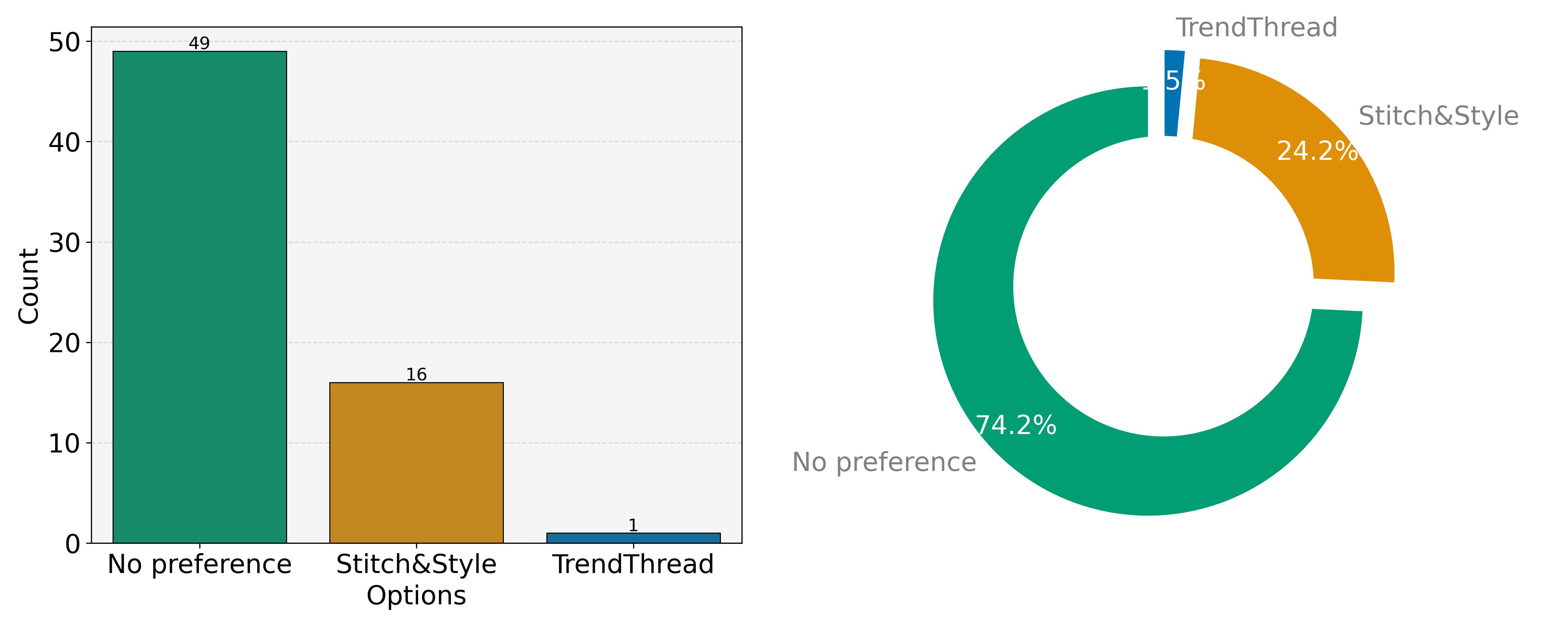}
\caption{Distribution of preferences among respondents of Question 3} 
\label{fig:30Q1(2)}
\end{figure}

\begin{figure}[!tb]
\center
\includegraphics[width=\textwidth]{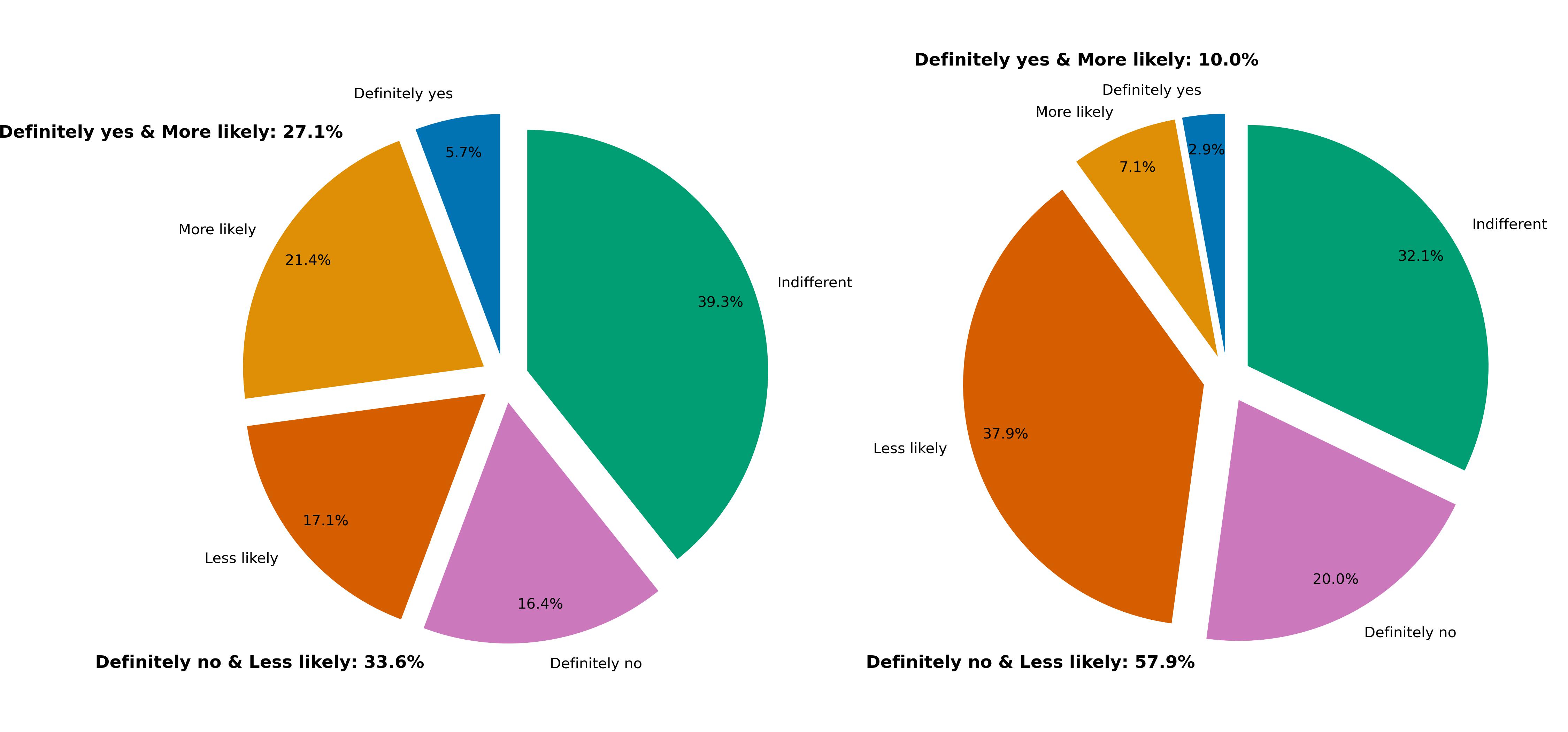}
\caption{\textbf{Left: } Distribution of preferences for store, \textit{TrendThread}; \textbf{Right: } Distribution of preferences for store, \textit{Stitch\&Style} } 
\label{fig:30Q3}
\end{figure}

\begin{figure}[!tb]
\center
\includegraphics[width=\textwidth]{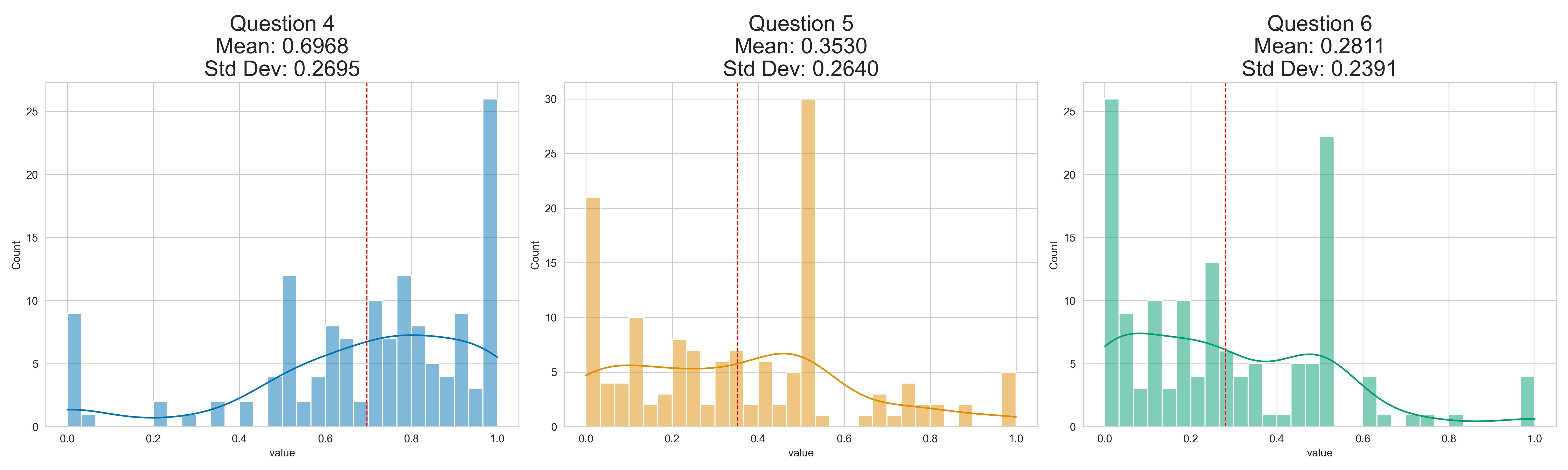}
\caption{\textbf{Left: } Distribution of preferences of Question 5; \textbf{Middle: } Distribution of preferences of Question 6;  \textbf{Right: } Distribution of preferences of Question 7} 
\label{fig:30Q4}
\end{figure}

\begin{figure}[!tb]
\center
\includegraphics[width=0.5\textwidth]{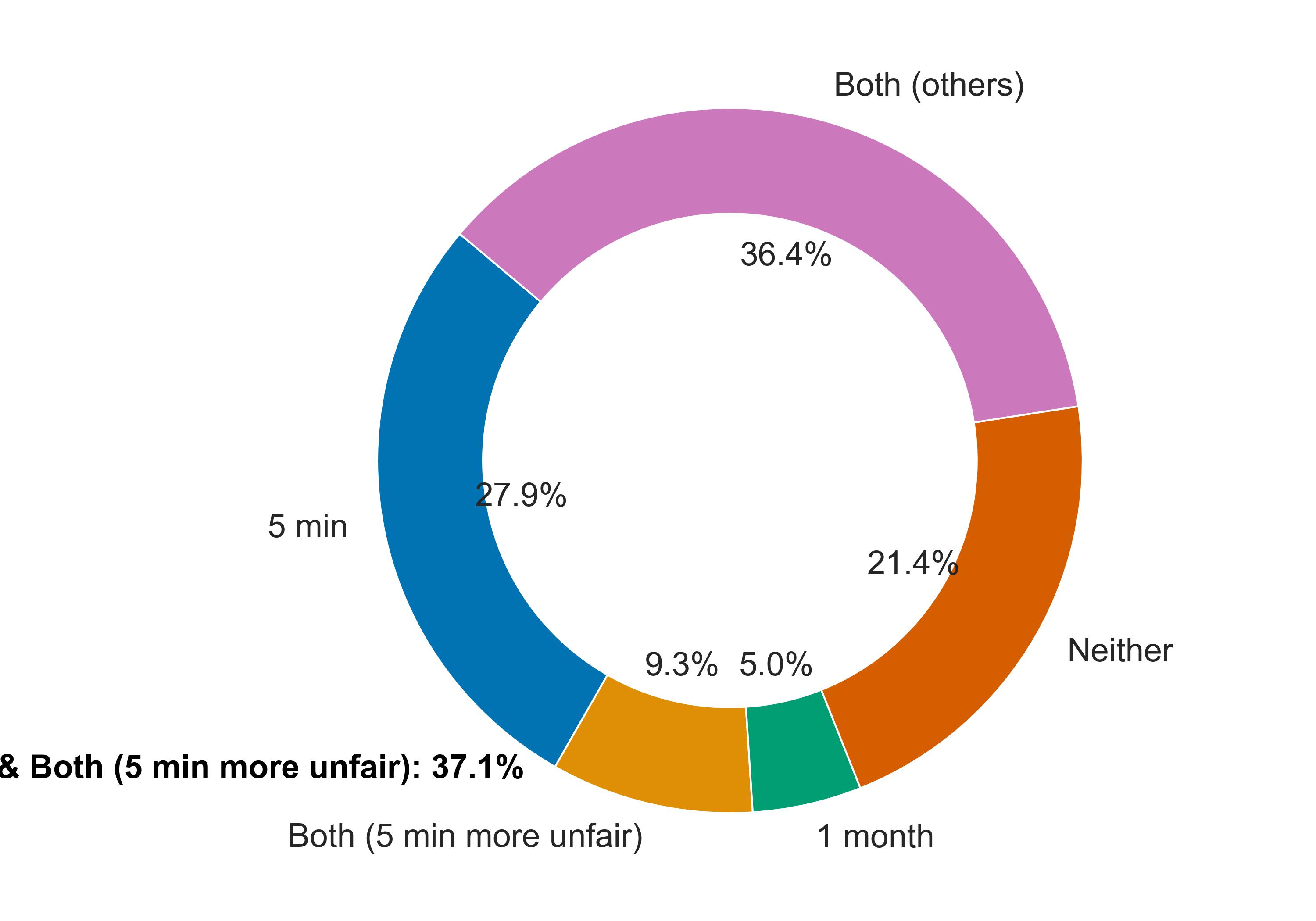}
\caption{Distribution of preferences among respondents of Question 8} 
\label{fig:30Q7}
\end{figure}

\newpage

\section{Supplementary Material for Section \ref{sec:intro}} \label{append:intro}

\subsection{Other Related Work} \label{subsec:other}

{\textbf{Time Variant Individual Fairness. } In the fields of Computer Science and Operations Research, individual fairness, often referenced as envy-freeness, is a critical metric in scenarios where equitable treatment of entities—such as customers, users, or resources—is desired. The time-variant aspect of individual fairness adds another layer of complexity, as it demands that customers arriving at similar times should receive similar treatment. Different papers adapt the concept of fairness to suit their specific context, which can lead to varying definitions and formulations. We present some general axioms and mathematical formulations that capture the essence of time-variant individual fairness as it has been discussed in the literature thus far:


\paragraph{{\bf Axiom 1: Lipschitz Continuity in Time}} \citep{dwork2012fairness,gupta2021individual,mukherjee2020two} This formulation states that the change in treatment should be limited by the elapsed time between arrivals, ensuring a smooth transition:
   \[
   D(f(t), f(s)) \leq L d(t, s)
   \]
   Here, \(t\) and \(ts\) represent the arrival time, \(f(x)\) denotes the treatment function, \(D\) is a distance metric measuring similarity between treatments, \(d\) is a distance metric measuring the time difference, and \(L\) is a Lipschitz constant that bounds the rate of change in treatment over time. This ensures that the treatment function \(f\) does not vary too sharply with small changes in time.

\paragraph{{\bf Axiom 2: Probabilistic Fairness}  (Our paper)} \citep{donahue2020fairness,herlihy2023planning}. This approach introduces randomness in treatment while maintaining fairness statistically:
   \[
   P(D\left(f(t),f(s)\right) \leq \delta | t, s) \geq 1 - \epsilon \text{ for } d(t, s) < \gamma
   \]
   where \(\delta, \gamma, \epsilon\) are thresholds dictating the acceptable level of treatment difference, time difference, and probability of unequal treatment, respectively.

\paragraph{{\bf Axiom 3: Time-Consistent Treatment}} \citep{cohen2021dynamic,gupta2023temporal,heidari2018preventing} This axiom ensures that entities arriving within a short time window receive comparable treatment. Mathematically, it can be expressed as:
   \[
   \forall t,s \text{ with } d(t, s) < \epsilon, \quad D(f(t), f(s)) \leq \delta
   \]
   Here, \(\epsilon, \delta\) are small constants for the permissible time window and treatment disparity, respectively.


Mathematically speaking, Axiom 1 is stronger than Axiom 2. The intuition behind is that in Axiom 1, the treatment function must not change more rapidly than a constant rate over time. This means that the rate at which treatment changes is tightly controlled, providing a uniform smoothness in how treatments are administered over time. As such, every algorithm satisfying Axiom 1, by default satisfies Axiom 2 as well.

Axiom 2 is stronger than Axiom 3. The intuition is that Axiom 2 generally specifies that for entities arriving within a close timeframe, the probability of receiving similar treatment must be high—usually quantified by a certain threshold. It requires a statistical guarantee that can be measured and verified. Axiom 3, while ensuring that entities arriving at similar times receive similar treatment, doesn’t specify how similar the treatments need to be or provide a statistical probability of similarity. Axiom 3 is more qualitative and lacks the explicit statistical guarantees required by Axiom 2, making it a less stringent and thus weaker notion of fairness.
}


We remark here that our paper (and our definition of fairness) satisfies Axiom 2. As such, it lies somewhere in the middle in terms of strength with regards to other fairness models/desiderata considered in the literature. Concretely, all the algorithms that satisfy Axiom 2 fairness must satisfy Axiom 3 fairness, thus making our contributions significant for a wider range of problems.

\xhdr{Modern Revenue Management Models.} Due to the limitations of the traditional network revenue management model, researchers have diversified its approach for greater relevance and depth. Despite these advancements, the resolution of these contemporary models often relies on the foundational algorithmic frameworks discussed in Section \ref{sec:classicalnrm}. Our contributions aim to refine these foundational structures, thereby promoting individual fairness across all advanced revenue management models that employ decision-making algorithms outlined in Section \ref{sec:classicalnrm}. For instance, to ensure group fairness, a \AlgDLP algorithm is proposed by \cite{ma2020group}, while a \AlgDBPC algorithm is suggested by \cite{balseiro2021regularized}. In the context of choice-based revenue management, a choice-based \AlgDLP algorithm is recommended by \cite{liu2008choice, gallego2015online}. To tackle revenue management issues involving offline prediction in demand, \cite{balseiro2022single, golrezaei2023online} advocate for a \AlgBL algorithm. Revenue management with correlated arrivals is addressed by \cite{jiang2023constant} with an \AlgSBPC method, and \cite{jiang2020online} proposes a \AlgDBPC algorithm for non-stationary arrivals. In scenarios involving revenue management with reusable resources, \cite{huo2022online} employs a \AlgN algorithm. For revenue management with horizontal uncertainty, a \AlgDBPC algorithm is suggested by \cite{balseiro2023online}. Finally, \cite{vera2019bayesian} extends the revenue management framework to a more general online decision-making setting and applies a variant of the \AlgRDLP algorithm to broader problems, such as online packing.

\xhdr{Price-based Revenue Management.} Contrasting with the quantity-based model in this paper, price-based revenue management represents a distinct category. Here, pricing decisions inversely affect demand, with higher prices leading to lower demand and vice versa. \cite{maglaras2006dynamic} suggests that, under certain conditions, price-based and quantity-based models align in function. Furthermore, price-based strategies also utilize the algorithmic frameworks mentioned in Section \ref{sec:classicalnrm}. Works by \cite{gallego1994optimal,gallego1997multiproduct} implement the \AlgDLP algorithm for decision-making, whereas \cite{jasin2012re,jasin2014reoptimization,wang2022constant} employ the \AlgRDLP algorithm. In Section \ref{sec:extension}, we explore the concept of individual fairness within the price-based framework and discuss the integration of a grace period design in this context.

\xhdr{Group Fairness.} While a considerable volume of literature exists concerning fairness in RM algorithms, much of it primarily aims to achieve group fairness \citep{cohen2022price,aminian2023fair,balseiro2021regularized,cayci2020group,donahue2020fairness,ma2020group,manshadi2021fair}. Group fairness in decision-making entails ensuring that every group (or type) of customer is treated equitably. However, group fairness does not guarantee individual fairness. This is because, firstly, individuals might not be aware of the groups to which they belong. Secondly, individuals within each group can be treated inequitably based on their specific contexts and features.

\newpage

\section{Supplementary Material for Section \ref{sec:classicalnrm}} \label{appendix:algs}

\begin{algorithm}[h]
\caption{\AlgDLP}
\label{alg:DLP}

\DontPrintSemicolon 
\KwIn{Optimal primal variable $\textbf{x}^{*}$. Preference matrix $\textbf{A}$.  Arriving rate vector $\mathbf{\lambda}$. Capacity vector $\textbf{m}$. Time horizon $[0,T]$.}
\KwOut{Acceptance decision for customers over time and update of capacity vector $\textbf{m}$.}
\For{$t \in \{1,2,...,T\}$}{
    Observe customer of type $i \in [n]$.\;
    \uIf{$\textbf{A}_i \leq \textbf{m}(t)$ for all $i\in [n]$}{
        Accept customer with probability $\frac{x^{*}_{i}}{\lambda_{i}}$.\;
        \If{customer is accepted}{
            $\textbf{m}(t+1) \gets \textbf{m}(t) -  \textbf{A}_{i}$.\;
        }
    }
    \Else{
        Reject the arriving customer and break.\;
    }
}
\end{algorithm}

\begin{algorithm}[!tb]
\caption{\AlgRDLP}
\label{alg:RDLP}

\DontPrintSemicolon 
\KwIn{Preference matrix $\textbf{A}$. Arriving rate vector $\mathbf{\lambda}$. Capacity vector $\textbf{m}$. Time horizon $[0,T]$. Resolving time point $t^{*}$.}
\KwOut{Updated optimal primal solutions $\mathbf{x}^{*}$ and $\tilde{\mathbf{x}}^{*}$ based on capacity adjustments at $t^{*}$.}
At $t=0$, solve DLP based on the initial capacity. Denote the optimal primal solution as $\mathbf{x}^{*}$.\;
Run Algorithm \ref{alg:DLP} with initial capacity $\textbf{m}$, horizon $[0,t^{*}]$, and optimal primal solution $\mathbf{x}^{*}$.\;
At $t=t^{*}$, solve DLP based on the remaining capacity. Let $\tilde{\mathbf{x}}^{*}$ be the optimal primal solution.\;
Run Algorithm \ref{alg:DLP} with remaining capacity, horizon $[t^{*},T]$, and optimal primal solution $\tilde{\mathbf{x}}^{*}$.\;
\end{algorithm}

\begin{algorithm}[!tb]
\caption{\AlgSBPC}
\label{alg:fixedbid}

\DontPrintSemicolon 
\KwIn{Optimal dual variable $\mathbf{\theta}^{*}$. Preference matrix $\textbf{A}$. Reward vector $\textbf{r}$. Capacity vector $\textbf{m}$.}
\KwOut{Acceptance decision for customers over time and update of capacity vector $\textbf{m}$.}
\For{$t \in \{1,2,...,T\}$}{
    Observe customer of type $i$. Set $y_t \gets \mathbf{1} \big( r_i > \sum_{j=1}^L \theta_j^{*} A_{ij} \big)$.\;
    \uIf{$\textbf{A}_i \leq \textbf{m}(t)$ for all $i\in [n]$}{
        \uIf{$y_t$ equals to $1$}{
            Accept the arriving customer.\;
        }
        \Else{
            Reject the arriving customer.\;
        }
        Set $\textbf{m}(t+1) \gets \textbf{m}(t) - y_t A_{i}$.\;
    }
    \Else{
        Reject the arriving customer and break.\;
    }
}
\end{algorithm}

\begin{algorithm}[h]
\caption{\AlgDBPC} \label{alg:ogd}

\DontPrintSemicolon 
\KwIn{Preference matrix $\textbf{A}$. Reward vector $\textbf{r}$. Capacity vector $\textbf{m}$. OGD Parameter: $G$, $D$, $\bar \theta$.}
\KwOut{Update of dual variables $\mathbf{\theta}$ over time.}
\textbf{Initialize:} Dual variable $\mathbf{\theta}^{(0)} \gets \mathbf{0}$.\;
\For{$t =1,2,\ldots,T$}{
    Observe customer of type $i$. Set $y_t \gets \mathbf{1} \big( r_i > \sum_{j=1}^L \theta_j^{(t)} A_{ij} \big)$.\;
    \uIf{$\textbf{A}_i \leq \textbf{m}(t)$ for all $i\in [n]$}{
        \uIf{$y_t$ equals to $1$}{
            Accept the arriving customer.\;
        }
        \Else{
            Reject the arriving customer.\;
        }
        Set $\textbf{m}(t+1) \gets \textbf{m}(t) - y_t A_{i}$.\;
    }
    \Else{
        Reject the arriving customer and break.\;
    }
    Construct function $g_t(\mathbf{\theta}) = \sum_{j=1}^L \theta_j \left( \frac{m_j(t)}{T} - y_t A_{ij} \right)$.\;
    Update the dual variables using the OGD procedure:\;
    \Indp$\eta_t \gets \frac{D}{ G\sqrt{T} }$,\;
    $\mathbf{\theta}^{(t+1)} \gets
    \mathbf{\theta}^{(t)} - \eta_t \nabla_{\mathbf{\theta}} g_t(\mathbf{\theta}^{(t)})$,\;
    $\theta_j^{(t+1)}  \gets \min \bigg(\max \big(0, \theta_j^{(t+1)} \big) , \; \bar{\theta} \bigg)
    \text{ for all } i \in [m]$.\;
    \Indm
}
\end{algorithm}

\begin{algorithm}[h]
\caption{\AlgBL}
\label{alg:bookinglimit}

\DontPrintSemicolon 
\KwIn{Preference matrix $\mathbf{A}$. Capacity vector $\textbf{m}$. Booking limit $\mathbf{b}$.}
\KwOut{Acceptance decision for customers over time and update of capacity vector $\textbf{m}$.}
\textbf{Initialize:} Number of accepted customers $s_i=0$, $i \in [n]$.\;
\For{$t \in \{1,2,...,T\}$}{
    Observe customer of type $i \in [n]$.\;
    \uIf{$\textbf{A}_i \leq \textbf{m}(t)$ and $s_i < b_i$}{
        Accept the arriving customer, and set $s_i \gets s_i+1$, $\textbf{m}(t+1) \gets \textbf{m}(t) -  \textbf{A}_{i}$.\;
    }
    \Else{
        Reject the arriving customer.\;
    }
}
\end{algorithm}

\begin{algorithm}[h]
\caption{\AlgN}
\label{alg:nesting}

\DontPrintSemicolon 
\KwIn{Capacity $m$. Nesting quota $\mathbf{b}$.}
\KwOut{Acceptance decision for customers over time and update of capacity.}
\textbf{Initialize:} Number of accepted customers $s_i=0$, $i \in [n]$.\;
\For{$t \in \{1,2,...,T\}$}{
    Observe customer of type $i \in [n]$.\;
    \uIf{$m(t) \geq 1$ and $\sum_{j=i}^{n}s_j < b_i$}{
        Accept the arriving customer, and set $s_i \gets s_i+1$, $m(t+1) \gets m(t) -  1$.\;
    }
    \Else{
        Reject the arriving customer.\;
    }
}
\end{algorithm}

\newpage

\endproof

\section{Supplementary Material For Section \ref{sec:gpd}} \label{append:gpd}

\newpage
\section{Supplementary material for Section \ref{sec:stochatic}} \label{append:Sec5}

\subsection{Proof of Theorem \ref{thm:RDLPrevise}}

\begin{proof}{Proof of Theorem \ref{thm:RDLPrevise}}
The proof of Theorem \ref{thm:RDLPrevise} is split into two parts: in the first one, we show that Algorithm \ref{alg:RDLPrevise} is $(\alpha,\delta)$-fair. In the second part, we show that Algorithm \ref{alg:RDLPrevise} has a regret of $O(\max\{T^{\eta}, \log_{1-\beta} \delta\})$. Define the length of the uniform time segment as $\hat{T}=\max\{T^{\eta}, 2\bar a n \gamma\}$.

\textit{Part 1: }
Define $\tau$ as the stopping time, where one of the resources depletes. Let $k(\tau)$ be the index of time segment containing $\tau$. Note that here both $\tau$ and $k(\tau)$ are random variables. For any type $i$ customer $i(u)$ arriving within $[1,(k(\tau)-1)\hat{T}]$, by definition, both decreasing grace period and increasing grace period satisfy the $(\alpha,\delta)$-fair metrics. For any customer $i(u)$ arriving within $[k(\tau)\hat{T},T]$, Theorem \ref{thm:FCFSnetwork} shows that Algorithm \ref{alg:RDLPrevise} is $(\alpha,\delta)$-fair. Note that Theorem \ref{thm:FCFSnetwork} only proves the fairness of the decreasing grace period towards resource depletion. We can use a symmetric statement to prove the fairness of the increasing grace period towards resource depletion, and we omit the proof of this statement. 

\textit{Part 2: }
To show the regret bound, we have 
\begin{align*}
\Reg&=\E[\HO]-\Revp =\left(\E[\HO]-\RDLP\right)+\left(\RDLP-\Revp\right)\\&=O(T^{\eta})+\left(\RDLP-\Revp\right).
\end{align*}
Next, we upper bound the value of $\RDLP-\Revp$.

Define $\tau$ as the stopping time, where one of the resources depletes, and $t^{*}$ as the re-solving time point. Let $k(\tau)$, $k(t^{*})$ be the index of time segment containing $\tau$, $t^{*}$, respectively. For any arriving instance $I$, for the time segment $1,2,\ldots,k(\tau)-1$, we make the following coupling between R-DLP and GP Enhanced R-DLP: let $u_i(k)$ be the number of type $i$ customers accepted by DLP in the time segment $k$. We have $u_i(k) \sim \text{Bin}(\lambda_i\hat{T},\frac{x^{*}_i}{\lambda_i})$. Then, we couple $y_i$ in time segment $k$ to have the value $u_i(k)$. We can make this coupling because $y_i$ in time segment $k$ and $u_i(k)$ follow the same distribution. Let $w_i(k)$ be the number of rejected type $i$ customers due to the grace period, and let $Y_i(k)$ be the realization of total number of accepted type $i$ customers in time segment. For any instance $\mathcal I$, any customer type $i$, and any time segment index $s \in [k(t^{*})-1]$, we have
\begin{equation} \label{eq:couple2}
\sum_{k=1}^{s} u_i(k)-Y_i(k)= \sum_{k=1}^{s} u_i(k)-\left(y_i-w_i(k)+w_i(k-1)\right) = w_i(s) \leq 2\frac{\bar a}{\underline a} \gamma,
\end{equation}
which implies that
\begin{align}
    \nonumber& \RDLP(1,(k(t^{*})-1)\hat{T})-\Revp(1,(k(t^{*})-1)\hat{T}) \\ \nonumber&= \sum_{k=1}^{k(t^{*})-1} \left(\RDLP\left((k-1)\hat{T}+1,k\hat{T}\right)-\Revp\left((k-1)\hat{T}+1,k\hat{T}\right) \right) \\ \label{eq:beforeresolve}&= \sum_{i=1}^{n} w_i(k(t^{*})-1) \leq 2\frac{\bar a}{\underline a} n \gamma \bar r.
\end{align}

Next, let $\tilde{\mathbf{x}}^{*}$, $\tilde{\mathbf{y}}^{*}$ be the re-solved optimal solution by Algorithm 
\AlgRDLP and Algorithm \GPRDLP at time $t^{*}$ respectively. By Equation \eqref{eq:couple2}, the difference of remaining capacity Algorithm 
\AlgRDLP and Algorithm \GPRDLP at time $t^{*}$ is at most $O(\log 1/\delta)$. Therefore, 
\[
| \tilde{\mathbf{x}}^{*} - \tilde{\mathbf{y}}^{*} | = O\left(\frac{\log 1/\delta}{T}\right).
\]
For any instance $\mathcal I$, any customer type $i$, and any time segment index $s \in \{k(t^{*}), k(t^{*})+1,\ldots, k(\tau)-1\}$, we have
\begin{align} 
\sum_{k=k(t^{*})}^{s} u_i(k)-Y_i(k) \nonumber&= \sum_{k=k(t^{*})}^{s} u_i(k)-\left(y_i-w_i(k)+w_i(k-1)+O\left(\frac{\log T}{T}\right)\right) \nonumber \\
&= w_i(s) + (s-k(t^{*}))O\left(\frac{\log 1/\delta}{T}\right) \nonumber \\ \label{eq:couple3}&\leq 2\frac{\bar a}{\underline a} \gamma +O(\log 1/\delta).
\end{align}
Therefore, we have
\begin{align}
    \nonumber& \RDLP(k(t^{*})\hat{T}),(k(\tau)-1)\hat{T})-\Revp(k(t^{*})\hat{T}),(k(\tau)-1)\hat{T}) \\ \nonumber&= \sum_{k=k(t^{*})}^{k(\tau)-1} \left(\RDLP\left((k-1)\hat{T}+1,k\hat{T}\right)-\Revp\left((k-1)\hat{T}+1,k\hat{T}\right) \right) \\ \label{eq:afterresolve}&= \sum_{i=1}^{n} w_i(k(\tau)-1) + O(\log 1/\delta) \leq 2\frac{\bar a}{\underline a} n \gamma \bar r + O( \log 1/\delta).
\end{align}

As we can have at most $2\hat{T}$ loss in the time segment $k(t^{*})$ and $k(\tau)$, by Equations \eqref{eq:beforeresolve} and \eqref{eq:afterresolve}, we can obtain:
\begin{align*}
\RDLP-\Revp &= 2\max\{T^{\eta}, 2\bar a n \log 1/\delta\} +O(\log 1/\delta) \\&= O(\max\{T^{\eta}, \log_{1-\beta} \delta\}).
\end{align*}
\Halmos
\endproof

\subsection{Proof of Theorem \ref{thm:reviseogd}}
\proof{Proof of Theorem \ref{thm:reviseogd}}
First, in the period $(0,T^{2/3})$, Algorithm \GPDBPC rejects all customers, which is obviously $(\alpha,\delta)$-fair. As is shown in Theorem \ref{thm:RDLPrevise}, Algorithm \AlgRDLP satisfies $(\alpha,\delta)$-fair metrics, we can obtain that Algorithm \GPDBPC is $(\alpha,\delta)$-fair.

Secondly, to address the regret bound, we observe that the regret bound can be formulated as
\begin{align} \label{eq:defregretogd}
    \Reg&=\E[\HO]-\Revp  = \nonumber \\
    &=\Big(\E\big[\HO(0,T^{2/3})\big] - \Revp(0,T^{2/3})\Big) + \Big(\E\big[\HO(T^{2/3},T)\big] - \Revp(T^{2/3},T)\Big),
\end{align}
where $\Revp$ is revenue generated by Algorithm \GPDBPC.

Upon rejecting all customers within $[0,T^{2/3}]$, the maximum loss incurred is $T^{2/3}$ during the initial $T^{2/3}$ periods, which implies that the first term in Equation \eqref{eq:defregretogd} is bounded by $T^{2/3}$. To bound the second term, we show that the Euclidean distance between the rate of accepted type $i$ customers by \AlgDBPC ($\frac{u_i}{\Lambda_i}$), and that by DLP ($\frac{x_i^{*}}{\lambda_i}$) is close. Let the expected total revenue generated by \AlgDBPC between $(t_1,t_2)$ be $\Rogd(t_1,t_2)$. By \cite{balseiro2023best}, we can obtain that \AlgDBPC has the regret bound of $O(\sqrt{T})$, which implies that for $t \in [0,T^{2/3}]$, 
\begin{equation} \label{eq:sketchogd1}
\HO(0,T^{2/3})-\Rogd(0,T^{2/3})=\HO(0,T^{2/3})-\sum_{i \in [n]}r_iu_i = O\big(\sqrt{T^{2/3}}\big) = O\big(T^{1/3}\big).
\end{equation}
A similar deduction holds for DLP given its regret bound of $O(\sqrt{T})$, leading to 
\begin{equation} \label{eq:sketchogd2}
\HO(0,T^{2/3})-\DLP(0,T^{2/3})=\HO(0,T^{2/3})-\sum_{i \in [n]}r_ix_i^{*}T^{2/3} = O\big(\sqrt{T^{2/3}}\big) = O\big(T^{1/3}\big),
\end{equation}
where $x_i^{*}$ is the $i^{\text{th}}$ element of the optimal solution of DLP. By Equations \eqref{eq:sketchogd1} and \eqref{eq:sketchogd2}, we have
\begin{equation} \label{eq:sketchogd3}
\vert \sum_{i \in [n]}r_iu_i - \sum_{i \in [n]}r_ix_i^{*}T^{2/3} \vert = O\big(T^{1/3}\big).
\end{equation}

Building on assumption 3 of \cite{agrawal2014dynamic}, according to \cite{sun2020near}, for any $i \in [n]$, the following equation holds with probability $1-O(\frac{1}{T})$:
\begin{equation} \label{eq:sketchogd4}
\vert u_i - x_i^{*}T^{2/3} \vert = O\big(T^{1/3} \log T\big).
\end{equation}

For a given $t \in [T^{2/3},T]$, Algorithm \GPRDLP is executed with a time segment length of $T^{2/3}$, and with $x_i^{*}=\frac{u_i}{T^{2/3}}$. By Equation \eqref{eq:sketchogd4}, we have $\vert \frac{x_i^{*}}{\lambda_i}-\frac{u_i}{\lambda_i T^{2/3}} \vert = O\big(T^{-1/3}\big)$, and by Hoeffding's inequality, we have $\vert \frac{u_i}{\Lambda_i}-\frac{u_i}{\lambda_i T^{2/3}} \vert = O\big(T^{-1/3} \log T\big)$. Then, with triangle inequality, we have
\begin{equation}\label{eq:sketchogd5}
    \big\vert \frac{x_i^{*}}{\lambda_i}-\frac{u_i}{\Lambda_i} \big\vert = O\big(T^{-1/3} \log T\big),
\end{equation}
Therefore, to bound the second term in Equation \eqref{eq:defregretogd}, we have
\begin{align}\label{eq:sketchogd6}
    & \nonumber \HO(T^{2/3},T) - \Revp(T^{2/3},T) \\&= \nonumber \Big(\HO(T^{2/3},T) - \Rdlpp(T^{2/3},T)\Big) \\ \nonumber &+  \Big(\Rdlpp(T^{2/3},T)-\Revp(T^{2/3},T)\Big) \\&= \nonumber O(\sqrt{T}) + \Big(\Rdlpp(T^{2/3},(k(\tau)-1)T^{2/3})-\Revp(T^{2/3},(k(\tau)-1)T^{2/3}) \Big) \\&+ \Big(\Rdlpp((k(\tau)-1)T^{2/3},k(\tau)T^{2/3}) - \Revp((k(\tau)-1)T^{2/3},k(\tau)T^{2/3}) \Big).
\end{align}
By Equation \eqref{eq:sketchogd5}, we can obtain
\begin{align*}
\Rdlpp(T^{2/3},(k(\tau)-1)T^{2/3})-\Revp(T^{2/3},(k(\tau)-1)T^{2/3}) &= (T-T^{2/3}) O\big(T^{-1/3} \log T\big) \\&= O\big(T^{2/3} \log T\big).
\end{align*}
Moreover, as the length of the time segment $k(\tau)$ is $\max\{T^{2/3}, 2 \bar a n \log 1/\delta\}$, we have
\begin{align*}
& \Rdlpp((k(\tau)-1)\max\{T^{2/3}, 2 \bar a n \log 1/\delta\},k(\tau)\max\{T^{2/3}, 2 \bar a n \log 1/\delta\}) \\&- \Revp((k(\tau)-1)\max\{T^{2/3}, 2 \bar a n \log 1/\delta\},k(\tau)\max\{T^{2/3}, 2 \bar a n \log 1/\delta\}) \\&= O(\max\{T^{2/3}, 2 \bar a n \log 1/\delta\}).
\end{align*}

Therefore, by Equation \eqref{eq:sketchogd6}, the second term in Equation \eqref{eq:defregretogd} is bounded by $O(\sqrt{T})+O(T^{2/3}\log T)+O(\max\{T^{2/3}, 2 \bar a n \log 1/\delta\})=O(\max\{T^{2/3}\log T, 2 \bar a n \log 1/\delta\})$. Finally, by Equation \eqref{eq:defregretogd}, we have the regret of Algorithm \GPDBPC is 
\[
\HO - \Revp = O(\max\{T^{2/3}, 2 \bar a n \log 1/\delta\})+O(T^{2/3}\log T)=O(\max\{T^{2/3} \log T,  \log 1/\delta\}).
\]
\Halmos
\endproof

\end{proof}

\section{Supplementary material for Section \ref{sec:adversarial}} \label{append:Sec6}


\proof{Proof of Theorem \ref{thm:revisebooking}}
First, we show that \GPBL is $(\alpha,\delta)$-fair. Without the capacity constraint, for each customer type $i$, \GPBL is the same as the \FCFS algorithm under a resource capacity $b_i$, where $b_i$ is the booking limit. Therefore, by Theorem \ref{thm:FCFSnetwork}, \GPBL is $(\alpha,\delta)$-fair if the demand is less than the capacity. Then, let's consider the capacity constraint. Suppose that at a certain time period $t$, $\min_{j \in [L]}m_j(t) - \bar a n \gamma \leq 0$, for each customer type $i$, if $s_i<b_i-\gamma$, then \GPBL starts a decreasing grace period to type $i$. If $s_i \geq b_i-\gamma$, then type $i$ has been already started a decreasing grace period, and this will not have any impact on the decision. Therefore, by Theorem \ref{thm:FCFSnetwork}, \GPBL is $(\alpha,\delta)$-fair.

Second, we show that \GPBL has a competitive ratio of $C-O(\frac{\log 1/\delta }{m})$. If the total number of arrivals is $o(m)$, then with probability $1$, no decreasing grace period will start because the number of accepted customers is much fewer than the booking limit or resource capacity. This leads to a competitive ratio of $1$ because Algorithm \ref{alg:bookingrevised} accepts everyone. 

If the total number of arrivals is $\Omega(m)$, the offline optimal revenue is $\Theta(m)$ as the resource capacity scales with $m$. Compare to Algorithm \AlgBL, the loss is only from the decreasing grace period. The maximum number of customers who are rejected due to the decreasing grace period is the total length of decreasing grace periods: $\bar a n \gamma + n \gamma = (\bar a+1) n \gamma$, which will incur a maximum revenue loss of $(\bar a+1) n \gamma \bar r$. Therefore, the competitive ratio is 
\[
\inf_{I \in \mathcal{I}}\frac{\text{Rev}(\pi(I))}{\text{OPT}(I)} \geq \inf_{I \in \mathcal{I}}\frac{\text{Rev}(\textsf{BL}(I))-(\bar a+1) n  \gamma \bar r}{\text{OPT}(I)} = C-O\left(\frac{\log 1/\delta }{m}\right).
\]
\Halmos
\endproof

\section{Supplementary material for Section \ref{sec:extension}} \label{append:Sec8}

\proof{Proof of Theorem \ref{thm:pricebased}}
First, we show that by assigning a decreasing grace period $[t_1(i), t_2(i)]$, where $t_1 (i) = \inf\{t: \min_{j \in [L]}m_j(t) \leq \bar a n \gamma\}$ and $t_2(i) = T$ to each type $i$ customer, the static pricing algorithm is $(\alpha, \delta)$-fair. 

Define event $E$ as $E:= \{\text{no resource is depleted in the time interval } [0, T] \}$. We show that $E$ happens with probability at least $1-\delta$ by showing that the complement of $E$ happens with probability at most $\delta$.
\begin{align}
    \mathbb{P}(E ^{\mathsf{C}}) \nonumber&\leq \mathbb{P}(\text{more than $\gamma n$ customers are accepted in the grace period}) \\ \label{eq:pfthm7eq1}& \leq \mathbb{P}(\exists i \in [n], \text{ s.t. the number of type } i \text{ customers accepted in the grace period } \geq \gamma) \\  \nonumber &=1 - \mathbb{P}(\exists i \in [n], \text{ s.t. the number of type } i \text{ customers accepted in the grace period } < \gamma) \\ \label{eq:pfthm7eq2}&= (1-\beta)^{\gamma} \\ \nonumber &= (1-\beta)^{\log_{1-\beta}\delta} = \delta,
\end{align}
where \eqref{eq:pfthm7eq1} is due to the pigeonhole principle, and \eqref{eq:pfthm7eq2} is because for each type $i$, the number of type $i$ customers accepted after the grace period starts is a geometric random variable with success probability $\alpha$. Therefore, the cdf is $1-(1-\alpha)^{\gamma}$.

Conditional on event $E$ happening, for any customer $i(u)$ arriving within $[1,t_1]$, the algorithm gives the same price $p_u^{(i)}$ to all of them. For any customer $i(u)$ arriving within $[t_1,t_2]$, by the definition of the grace period, we obtain that $\mathbb{P}\left(  p_u^{(i)} \neq p_u^{(i+1)}  \right) \leq \alpha$.

To complete the proof of the theorem, we need to show that the revenue loss is bounded by $\frac{\bar a}{\underline a} n \gamma \bar p$. In the worst case for the revenue, the algorithm rejects all customers after $t_1$. By the definition of $t_1$, the remaining units for each $j$ are $\bar a n \gamma$. Since $\bar a n \gamma$ units of resource can serve at most $\frac{\bar a}{\underline a} n \gamma$ customers, we have that the revenue loss is at most $\frac{\bar a}{\underline a} n \gamma \bar p$. \Halmos
\endproof

\section{Supplementary Material For Section \ref{sec:discussion}} \label{append:def}

\begin{proposition} \label{prop:p1}
Given a time horizon \( [1, T] \), under the fairness constraint in Equation~\eqref{eq:tempfair}, any algorithm must incur a total revenue of \( 0 \).
\end{proposition}

\proof{Proof of Proposition \ref{prop:p1}}
Suppose there is only one type of customer. Let \( \mathcal{A} \) be any algorithm satisfying the fairness constraint in Equation~\eqref{eq:tempfair}. Define \( u_0 \) as the index of the first customer whose acceptance probability under \( \mathcal{A} \), $p_{u_0}$, is strictly greater than $0$. 

We now analyze how the lower bound of acceptance probabilities evolve for customers with indices \( v > u_0 \). For any such customer \( v \), we have:

\begin{align}
&\mathbb{P}[\text{customer } v \text{ is accepted}] \nonumber \\
&= \mathbb{P}[\text{accepted} \mid \text{prev accepted}] \cdot \mathbb{P}[\text{prev accepted}] \nonumber \\
&\quad + \mathbb{P}[\text{accepted} \mid \text{prev rejected}] \cdot \mathbb{P}[\text{prev rejected}] \nonumber \\
&\geq \mathbb{P}[\text{accepted} \mid \text{prev accepted}] \cdot \mathbb{P}[\text{prev accepted}] \nonumber \\
&= \left(1 - \mathbb{P}[1(v-1) \succ_{\mathcal{A}} 1(v)]\right) \cdot \mathbb{P}[\text{customer } v-1 \text{ is accepted}] \nonumber \\
&\geq (1 - \alpha) \cdot \mathbb{P}[\text{customer } v-1 \text{ is accepted}], \label{eq:defappend1}
\end{align}
where \( 1(v) \) denotes the \( v \)-th customer of type 1 (and we only have one customer type here). From the recursive inequality \eqref{eq:defappend1}, we see that for any \( x \geq 0 \), the acceptance probability of customer \( u_0 + x \) is lower bounded by:
\[
\mathbb{P}[\text{customer } u_0 + x \text{ is accepted}] \geq p_{u_0} (1 - \alpha)^x.
\]

Next, we claim that any feasible algorithm must eventually assign acceptance probability zero to at least one customer before the resource is fully depleted. We prove this by contradiction. Suppose instead that all customers have strictly positive acceptance probabilities before the resource runs out. Consider the point when only one unit of resource remains, and let \( \hat{u} \) denote the index of the customer arriving at that moment. Since the algorithm accepts customers with positive probability, it is possible that \( \hat{u} \) is accepted. But this implies that the next customer \( \hat{u} + 1 \) must be rejected with probability 1 due to zero remaining capacity. This means:
\[
\mathbb{P}[1(\hat{u}) \succ_{\mathcal{A}} 1(\hat{u} + 1)] = 1 > \alpha,
\]
which violates the fairness constraint in Equation~\eqref{eq:tempfair}.

Therefore, to satisfy the fairness constraint, the algorithm must assign zero acceptance probability to some customer before the resource is depleted. But from \eqref{eq:defappend1}, if the acceptance probability of customer \( u_0 \) is non-zero, then no subsequent customer will ever have acceptance probability exactly zero, since \( (1 - \alpha)^x > 0 \) for all finite \( x \). This contradiction implies that the acceptance probability of customer \( u_0 \) must be zero. Finally, if the first customer \( u_0 \) is already assigned zero acceptance probability, then the algorithm rejects all customers, resulting in zero revenue. \Halmos
\endproof

\end{document}